\documentclass{article}

\usepackage[utf8x]{inputenc}

\usepackage{xspace}
\usepackage{graphicx}
\usepackage{listings}
\usepackage{multirow}
\usepackage{amsthm}
\usepackage{amssymb}
\usepackage{amsfonts}
\usepackage{amsmath}
\usepackage{url}
\usepackage[ruled,linesnumbered,noline,noend]{algorithm2e}
\usepackage{tikz}
\usetikzlibrary{arrows,automata}

\theoremstyle{definition}
\newtheorem{thm}{Theorem}%[section]
\newtheorem{lem}{Lemma}%[section]
\newtheorem{cor}{Corollary}%[section]
\newtheorem{exm}{Example}%[section]
\newcommand{\false}{\ensuremath{\mathit{false}}}
\newcommand{\true}{\ensuremath{\mathit{true}}}
\newcommand{\ite}{\ensuremath{\mathbf{ite}}}
\newcommand{\bbt}{\ensuremath{\mathcal{B}}}
\newcommand{\lst}{\ensuremath{\mathbf{lst}}}
\newcommand{\pref}{\ensuremath{\mathbf{pref}}}
\newcommand{\exits}{\ensuremath{\mathbf{exits}}}
\newcommand{\var}[1]{{\tt #1}}
\newcommand{\sym}[1]{\ensuremath{\underline{#1}}}
\newcommand{\arrsym}[1]{\ensuremath{\lambda \vec{\chi}.\sym{#1}(\vec{\chi})}}
\newcommand{\matr}[1]{\ensuremath{ \left( \begin{smallmatrix} #1 \end{smallmatrix} \right) }}

\newcommand{\N}{\mathbb{N}}

\newcommand{\Pex}{\textsc{Pex}\xspace}
\newcommand{\Sage}{\textsc{Sage}\xspace}
\newcommand{\Cute}{\textsc{Cute}\xspace}
\newcommand{\Klee}{\textsc{Klee}\xspace}
\newcommand{\Exe}{\textsc{Exe}\xspace}
\newcommand{\APC}{\textsc{Apc}\xspace}

\DontPrintSemicolon
\newcommand{\aarg}[2]{#1~{\it //~#2}}
\newcommand{\aargm}[2]{\\ \qquad \aarg{#1}{#2}}
\newcommand{\aset}{\ensuremath{\longleftarrow}}

\newcommand{\note}[1]{}
\newcommand{\todo}[1]{}

\begin{document}

%%%%%%%%%%%%%%%%%%%%%%%%%%%%%%%%%%%%%%%%%%%%%%%%%%%%%%%%%%%%%%%%%%%%%%%%
%%%%%%%%%%%%%%%%%%%%%%%%%%%%%%%%%%%%%%%%%%%%%%%%%%%%%%%%%%%%%%%%%%%%%%%%
%%%%%%%%%%%%%%%%%%%%%%%%%%%%%%%%%%%%%%%%%%%%%%%%%%%%%%%%%%%%%%%%%%%%%%%%
\title{Abstracting Path Conditions for Effective Symbolic Execution}
\author{Marek Trt\'{i}k\\trtik@fi.muni.cz}
\maketitle

%%%%%%%%%%%%%%%%%%%%%%%%%%%%%%%%%%%%%%%%%%%%%%%%%%%%%%%%%%%%%%%%%%%%%%%%
%%%%%%%%%%%%%%%%%%%%%%%%%%%%%%%%%%%%%%%%%%%%%%%%%%%%%%%%%%%%%%%%%%%%%%%%
%%%%%%%%%%%%%%%%%%%%%%%%%%%%%%%%%%%%%%%%%%%%%%%%%%%%%%%%%%%%%%%%%%%%%%%%
\begin{abstract}
We present an algorithm for tests generation tools based on symbolic execution. The algorithm is supposed to help in situations, when a tool is repeatedly failing to cover some code by tests. The algorithm then provides the tool a necessary condition strongly narrowing space of program paths, which must be checked for reaching the uncovered code. We also discuss integration of the algorithm into the tools and we provide experimental results showing a potential of the algorithm to be valuable in the tools, when properly implemented there.
\end{abstract}

%%%%%%%%%%%%%%%%%%%%%%%%%%%%%%%%%%%%%%%%%%%%%%%%%%%%%%%%%%%%%%%%%%%%%%%%
%%%%%%%%%%%%%%%%%%%%%%%%%%%%%%%%%%%%%%%%%%%%%%%%%%%%%%%%%%%%%%%%%%%%%%%%
%%%%%%%%%%%%%%%%%%%%%%%%%%%%%%%%%%%%%%%%%%%%%%%%%%%%%%%%%%%%%%%%%%%%%%%%
\section{Introduction} 

Symbolic execution serves as a basis in many successful tools for test generation, including  \Klee~\cite{CDE08}, \Exe~\cite{Cadar08}, \Pex~\cite{TdH08}, \Sage~\cite{GLM08:fuzzing}, or \Cute~\cite{SMA05}. These tools can relatively quickly find tests that cover majority of code close to program entry location. But then the ratio of covered code increases very slowly or not at all. The reason for that is a huge number of program paths to be explored. And it typically becomes very difficult to find a path to a given yet uncovered program location among all those paths. We speak about the \emph{path explosion problem}.

In this paper we introduce an algorithm for the tools mentioned above, which can be very useful in situations, when all attempts to cover a particular program location are repeatedly failing (so a tool stops making a progress). Given that program location, our algorithm computes a nontrivial \emph{necessary condition} over-approximating a set of program paths leading to that location. The intention is to have the over-approximation as small as possible, while still keeping the condition simple for SMT solvers. Having this condition a tool can quickly recover from the failing situation by exploring only paths satisfying it.

For a given program and a target location in it we construct the necessary condition by collecting constraints appearing along acyclic program paths from entry location to the target one, while summarizing effects of loops along them. It is well known that loops are the main source of the path explosion problem. Therefore, the key part of the algorithm is the computation of loop summaries.

The algorithm is supposed to be integrated into the tools as their another heuristic. Therefore, complexity of the integration also matters. We show that since the algorithm actually computes a formula, the integration is very straightforward. And we show on a small set of representative benchmarks that \Pex could benefit from the algorithm, when properly implemented there. The results can be extrapolated to the remaining tools, since they have the common theoretical background. 

%%%%%%%%%%%%%%%%%%%%%%%%%%%%%%%%%%%%%%%%%%%%%%%%%%%%%%%%%%%%%%%%%%%%%%%%
%%%%%%%%%%%%%%%%%%%%%%%%%%%%%%%%%%%%%%%%%%%%%%%%%%%%%%%%%%%%%%%%%%%%%%%%
%%%%%%%%%%%%%%%%%%%%%%%%%%%%%%%%%%%%%%%%%%%%%%%%%%%%%%%%%%%%%%%%%%%%%%%%
\section{Program}

\paragraph{Program definition} A \emph{program} is a tuple $P = (V_P, E_P, l_s, l_t, \iota_P)$ such that $(V_P, E_P)$ is a connected oriented graph, vertices $V_P$ represent program locations and edges $E_P$ represent control flow between them. $P$ has a single \emph{start vertex} $l_s \in V_P$ and a single \emph{target vertex} $l_t \in V_P$, satisfying $l_s \ne l_t$. Each vertex has out-degree at most $2$. A vertex is a \emph{branching vertex} if its out-degree is exactly $2$. All other vertices, except $l_t$, have out-degree $1$. In-degree of $l_s$ and out-degree of $l_t$ are both $0$. Function $\iota_P: E_P \rightarrow \mathcal{I}$ assigns to each edge $e$ a single \emph{instruction} $\iota(e)$ from the set $\mathcal{I}$ of all instructions. Out-edges of any branching vertex are labelled with instructions $\texttt{assume}(\gamma)$ and $\texttt{assume}(\neg \gamma)$, where $\gamma$ is a boolean expression. Any other edge (i.e.~non-branching one) is labelled either with an assignment instruction $e_1 \aset e_2$, where $e_1, e_2$ are l-value and r-value expressions respectively, or with an assertion $\texttt{assert}(\psi)$ or an assumption $\texttt{assume}(\psi)$, for some boolean expression $\psi$, or \texttt{skip} instruction, which does nothing. We assume that expressions in program instructions have no side-effects. Without loss of generality we require that boolean expressions in \texttt{assume} and \texttt{assert} instructions contain no logical connective (i.e.~they are predicates). We further require that semantic of all instructions in $\mathcal{I}$ uses only linear integer arithmetic and arrays. Note that $P$ does not contain neither function calls nor pointer arithmetic.  We can supply a precondition $\varphi$ and a postcondition $\psi$ for $P$ by introducing new vertices $l_s, l_t$ and connecting them to old ones by two edges. And the labelling of the only out-edge of $l_s$ and of the only in-edge to $l_t$ are $\texttt{assume}(\varphi)$ and $\texttt{assert}(\psi)$ instructions respectively. When program $P$ is known form a context we often abbreviate $V, E, l_s, l_t$, and $\iota$.

\paragraph{Treating lists as arrays} Let us first consider an array \var{A}. If we define successor function $\textit{succ}$ on elements of \var{A}, the $k$-th element of \var{A}, commonly described as \texttt{A[$k$]}, can be identified by $\textit{succ}^k(\var{A})$. Note that $\textit{succ}^k(x)$ represents a composition of $k$ applications of $\textit{succ}$ starting on $x$. Let us now consider a list \var{L} with successor function $\textit{next}$. Then the $k$-th element of \var{L} can be identified by $\textit{next}^k(\var{L})$. Therefore, we can also use notation \texttt{L[$k$]} even for lists. Because of this equivalence in treating lists and arrays, we consider only arrays in the remainder of the text. It is also important to note that we do not provide shape analysis. Thus shape of lists and arrays are immutable. 

\paragraph{Assertions and assumptions in a program} Suppose that we execute symbolically a program. Let $\varphi$ be a path condition. Then $\texttt{assert}(\gamma)$ forces validity check of formula $\varphi \rightarrow \gamma$. The execution may continue only if the check succeeded. Note that the path condition is not updated. On the other hand $\texttt{assume}(\gamma)$ updates the path condition such that $\varphi \aset \varphi \wedge \gamma$ and the execution may continue, if updated $\varphi$ is satisfiable.

\paragraph{Program variables and expressions} Let $P$ be a program. Then $\mathcal{V}_P = \{ \var{a}, \var{A}, \var{b}, \var{B}, \ldots \}$ is a finite set of \emph{program variables}. We suppose each program variable has its type. We further define a countable set $\mathcal{E}_P$ of all syntactically correct expressions of $P$ over variables $\mathcal{V}_P$. We also suppose that each expression in $\mathcal{E}_P$ has its type. When $P$ is known from the context we write $\mathcal{V}$ and $\mathcal{E}$.

\paragraph{Path in $P$} A sequence $\pi = v_1 v_2 \cdots v_k$ is a \emph{path} in a program $P$, if for all $1 \leq i < k$ a pair $(v_i, v_{i+1})$ is an edge of $P$. We denote the empty path by $\varepsilon$. We identify the $i$-th vertex in $\pi$ as $\pi(i)$ and $|\pi|$ denotes the total number of vertices in $\pi$. But instead of $\pi(|\pi|)$ we write $\lst(\pi)$. For each path $\pi$ we define a set $\pref(\pi) = \{ \beta \mid \pi = \beta\gamma \}$ of all prefixes of $\pi$. A path $\pi$ from $l_s$ in a program $P$ is \emph{feasible} if there exists an input, such that execution of $P$ on the input follows $\pi$. Otherwise $\pi$ is \emph{infeasible}.

\paragraph{Backbone paths in $P$} Each acyclic path from $l_s$ to $l_t$ in $P$ is a \emph{backbone path} in $P$. Let $\pi_1, \pi_2$ be two backbone paths. Since each backbone path starts in $l_s$, there always exists a non-empty common prefix of $\pi_1, \pi_2$.

\paragraph{Reduction to a backbone path} Let $\pi$ be a path from $l_s$ to $l_t$ in a program $P$. We say that $\pi$ is \emph{reducible} to a backbone path $\pi'$, when a result of the following procedure applied to $\pi$ produces exactly the path $\pi'$: Let $k$ be the least index in $\pi$ such that vertex $\pi(k)$ occurs in $\pi$ once again (i.e.~at an index bigger then $k$). If no such $k$ exists, then we are done, as $\pi$ is a backbone path. Otherwise let $l$ be the greatest index such that $\pi(l) = \pi(k)$. If we denote $\pi(k)$ by $u$, then $\pi$ is of the form $\pi = \alpha u \beta u \gamma$, where $\beta \ne \varepsilon$ (since $l > k$). We set $\pi$ to $\alpha u \gamma$ and repeat the procedure.

Note that for each path $\pi$ from $l_s$ to $l_t$ there exists exactly one backbone path the path $\pi$ is reducible to.

\paragraph{Loop, loop entry vertex, and loop exit vertex of $P$} Let $P$ be a program and $\alpha v$ be an acyclic path from $l_s$ in $P$. Let $\mathcal{C}_v$ be the smallest subset of $V_P$ such that for each path $v v_1 \cdots v_n v$ of two or more vertices in $P$, where none of $v_1, \ldots, v_n \in V_P$ appears in the path $\alpha v$, all the vertices $v, v_1, \ldots, v_n \in \mathcal{C}_v$. If $\mathcal{C}_v \ne \emptyset$, then $\mathcal{C}_v$ is a \emph{loop at $v$ in $P$}, and $v$ is a \emph{loop entry vertex} of $P$. And a vertex $u \in V_P \smallsetminus \mathcal{C}_v$ is a \emph{loop exit from} $\mathcal{C}_v$, if there exists $w \in \mathcal{C}_v$ such that $(w,u) \in E_P$. We denote a set of all exit vertices from a loop $\mathcal{C}_v$ as $\exits(\mathcal{C}_v)$.

\paragraph{Program equivalence} Programs $P$ and $Q$ are \emph{equivalent}, if there exists a bijection between all paths from start to target vertex in $P$ and all paths from start to target vertex in $Q$ such that sequences of instructions along related paths are exactly the same, when ignoring \texttt{skip} instructions.

\paragraph{Normalized program} Program $P$ is \emph{normalized}, if each in-edge of each loop entry vertex of $P$ is labelled by an instruction \texttt{skip}. Given a program $P$, it is easy to compute a normalized program $P'$ which is equivalent with $P$: We start with $P'$ as a copy of $P$. For each loop entry vertex $v$ we create its copy $v'$ and then we replace every edge $(u,v)$ by a new edge $(u,v')$ with the same label. Finally we connect $v'$ with $v$ by a new edge $(v',v)$ labelled with \texttt{skip} instruction. 

\paragraph{} In the remainder of the text, whenever we speak about program we always assume it is normalized.

\paragraph{Program induced by a loop} Let $P$ be a program, $v$ be a loop entry vertex of $P$, and $\mathcal{C}$ be a loop at $v$. We can compute a program $P(\mathcal{C}, v)$, representing reachability in $\mathcal{C}$, as follows. We start with $P(\mathcal{C}, v)$ as a copy of $\mathcal{C}$. To get a program we only need to set right start and target vertices. In $P(\mathcal{C}, v)$ there must be a copy of $v$. Let $v'$ be the copy. Then we set $v'$ as the start vertex $l_s$ of $P'$. Further, we add a new vertex into $P(\mathcal{C}, v)$ and we set it as a target vertex $l_t$. Finally we replace each edge $(u,l_s)$ of $P(\mathcal{C}, v)$ by a new edge $(u,l_t)$ with the same label. We call the resulting program $P(\mathcal{C}, v)$ a \emph{program induced by a loop $\mathcal{C}$ at $v$}. \todo{toto by chtelo zpresnit}

\paragraph{Iterating $\bbt$} Let $P$ be a program induced by some loop $\mathcal{C}$ at a loop entry vertex of some bigger program. Further let $\bbt$ be a backbone tree of $P$. Since $\bbt$ represents all acyclic paths along $\mathcal{C}$, thus any execution looping in $\mathcal{C}$ actually iterates backbone paths in $\bbt$. Therefore, it is relevant to speak about \emph{iterating} $\bbt$. Similarly, when we consider a single backbone path $\pi$, then we can speak about \emph{iterating} $\pi$. Note that in case of presence of nested loops in $\mathcal{C}$ we can extend the definition recursively for iterating backbones trees of sub-induced programs, and so on.

%%%%%%%%%%%%%%%%%%%%%%%%%%%%%%%%%%%%%%%%%%%%%%%%%%%%%%%%%%%%%%%%%%%%%%%%
%%%%%%%%%%%%%%%%%%%%%%%%%%%%%%%%%%%%%%%%%%%%%%%%%%%%%%%%%%%%%%%%%%%%%%%%
%%%%%%%%%%%%%%%%%%%%%%%%%%%%%%%%%%%%%%%%%%%%%%%%%%%%%%%%%%%%%%%%%%%%%%%%
\section{Over-approximation $\hat{\varphi}$ of Feasible Paths}

\paragraph{Partitioning feasible path} There can be a huge or even infinite number of feasible paths from $l_s$ to $l_t$ in $P$. However, we can partition them into a finite number of classes according to the following lemma.

\begin{lem}
Let $\Phi_P$ be a set of all feasible paths form $l_s$ to $l_t$ in $P$. If $\Phi_P \ne \emptyset$, then there exists a finite partitioning $\Pi_P$ of $\Phi_P$ such that for each partition class $A$ of $\Pi_P$, there exists a unique backbone path $\pi_A$ in $P$ such that each path $\pi \in A$ is reducible to $\pi_A$.
\begin{proof}
Obvious.
\end{proof}
\end{lem}

\begin{cor}
Let $P$ be a program. Then $|\Pi_P| \leq |B_P|$.
\end{cor}

\paragraph{Partitioning of path conditions} A basic property of symbolic execution says that when symbolic execution on $P$ terminates, then each path condition uniquely identifies one feasible path in $P$ and vice versa. This bijection between feasible paths and related path conditions implies that the partitioning $\Pi_P$ actually also represents partitioning of related path conditions. Since both partitions are equivalent we do not distinguish between them.

\paragraph{Over-approximating $A \in \Pi_P$} Any formula $\hat{\varphi}_A$ is called an \emph{over-approximation} of a partition class $A \in \Pi_P$ if for each path condition $\varphi_A \in A$ a formula $\varphi_A \rightarrow \hat{\varphi}_A$ is valid. Note that $\hat{\varphi}_A \equiv \true$ is an abstraction of $A$.

\paragraph{Over-approximating $\Phi_P$} Any formula $\hat{\varphi}_P$ is called an \emph{over-approximation} of non-empty $\Phi_P$ if for each path condition $\varphi \in \Phi_P$ a formula $\varphi \rightarrow \hat{\varphi}_P$ is valid. Note that $\hat{\varphi}_P \equiv \true$ is an abstraction of $\Phi_P$. We write $\hat{\varphi}$, when $P$ is known from a context.

\paragraph{} We compute $\hat{\varphi}_P$ as a disjunction $\hat{\varphi}_{A_1} \vee \ldots \vee \hat{\varphi}_{A_n}$, where $\hat{\varphi}_{A_1}, \ldots, \hat{\varphi}_{A_1}$ are over-approximations of all partition classes $A_1, \ldots, A_n$ of $\Pi_P$ respectively.

%%%%%%%%%%%%%%%%%%%%%%%%%%%%%%%%%%%%%%%%%%%%%%%%%%%%%%%%%%%%%%%%%%%%%%%%
%%%%%%%%%%%%%%%%%%%%%%%%%%%%%%%%%%%%%%%%%%%%%%%%%%%%%%%%%%%%%%%%%%%%%%%%
%%%%%%%%%%%%%%%%%%%%%%%%%%%%%%%%%%%%%%%%%%%%%%%%%%%%%%%%%%%%%%%%%%%%%%%%
\section{Computation of $\hat{\varphi}$}
\label{sec:Algorithm}

\paragraph{Overview} Let $P$ be a program. If $P$ does not contain a loop and target location is reachable, then each partition class $A \in \Pi_P$ contains a single path, which is a backbone one. Therefore, if $\pi \in B_P$ is the only backbone path in a class $A$, then we can compute an over-approximation $\hat{\varphi}_A$ of $A$ as follows. We symbolically execute $\pi$. So we receive a path condition $\varphi$ and symbolic state $\theta$ from the execution. Since $\varphi \rightarrow \varphi$ is valid, we conclude that $\hat{\varphi}_A \equiv \varphi$.

Let us now consider a case, when $P$ contains a single loop at some loop entry vertex $w$, the target location is reachable, and $A \in \Pi_P$ is a partition class such that each $\alpha \in A$ is reducible to a backbone path $\pi = v_1 \cdots v_j w v_{j+1} \cdots v_n \in B_P$. We can compute an over-approximation $\hat{\varphi}_A$ of $A$ as follows. The class $A$ may contain even infinitely many feasible paths $\alpha_1, \alpha_2, \ldots$. But each path $\alpha_i$ is of a form $\alpha_i  = v_1 \cdots v_j \beta_i v_{j+1} \cdots v_n$, where $\beta_i \ne \varepsilon$ represents a different cyclic path along the loop from $w$ back to $w$. So the paths $\beta_i$ differ in number of iterations along the loop, and in interleaving of paths along the loop in separate iterations. Let us observe symbolic execution of paths $\alpha_i$. The execution proceeds exactly the same for a common prefix $v_1 \cdots v_j$. But then we reach the loop entry vertex $w$. Then the symbolic execution proceeds differently for all paths $\beta_i$. As a result we can get even infinitely many different path conditions and symbolic states. To prevent this, we compute an over-approximation of all those symbolic executions along the loop, so we get a single over-approximated path condition and a single over-approximated symbolic state. We compute the over-approximation as follows.

We build an induced program $P'$ of the loop at $w$ and we recursively call symbolic execution of its backbone paths as we do here for $P$. For each backbone path of $P'$ we receive a single path condition and single symbolic state. The path conditions and single symbolic states represent all possibilities, how to symbolically execute the loop once from $w$ back to $w$. But paths $\beta_i$ may go along the loop arbitrary number of times with arbitrary interleaving of paths through the loop. To maintain arbitrary iterations of backbone paths, we express values of all program variables of $P'$ as functions of number of iterations of backbone paths of $P'$. Then, to handle arbitrary interleaving of backbone paths in different iterations along the loop, we ``merge'' symbolic states of different backbone paths (separately and independently for each variable) into a single resulting symbolic state. Then we insert values in the resulting state into computed path conditions. We use them to build a formula stating that symbolic execution will keep looping in the loop, until proper number of iterations of individual backbone paths of $P'$ are met. The formula is a single resulting path condition. Both the resulting formula and symbolic state over-approximate sets of path conditions and symbolic states of the paths $\beta_i$ respectively, because their computation typically involve some lose of precision. We discuss in details the computation of the resulting formula and symbolic state in separate sections later. Only note that when $P'$ also contains some loops, then we resolve the situation by another recursive calls in all loop entry vertices in backbone paths of $P'$. This is the same process as we did at vertex $w$ of the backbone paths $\alpha_i$.
 
Let us now suppose that we already have the over-approximation, i.e~a single over-approximated path condition and a single over-approximated symbolic state. We can use them to proceed to symbolic execution of the common remainder $v_{j+1} \cdots v_n$ of paths $\alpha_i$. Obviously, we receive a single over-approximated path condition $\varphi$ and a single over-approximated symbolic state $\theta$ at the end. We show later in the section, that such a computed $\varphi$ is indeed an over-approximation $A$, i.e~$\hat{\varphi}_A \equiv \varphi$.

A case, when the backbone path $\pi$ contains more then one loop entry vertex (i.e.~$\pi$ goes through more then one loop), is now simple. Symbolic executions of the common parts of paths $\alpha_i$ (i.e.~those between loop entry vertices) are the same for all the paths $\alpha_i$. Whenever we reach a loop entry vertex, we call the over-approximation procedure to get a single over-approximated path condition and a single over-approximated symbolic state. At the end we again receive a single over-approximated path condition $\varphi$ representing $\hat{\varphi}_A$.

When the partitioning $\Pi_P$ has $n$ classes $A_1, \ldots, A_n$, then we apply the described procedure $n$ times, once for each class. We receive path conditions $\varphi_1, \ldots, \varphi_n$. Then a formula $\varphi_1 \vee \cdots \vee \varphi_n$ is an over-approximation $\hat{\varphi}$ of $\Phi_P$, since each $\varphi_i$ is an over-approximation of $A_i$.

In case the target location is not reachable, then none of the computed formulae $\varphi_1, \ldots, \varphi_n$ is satisfiable. Therefore, $\varphi_1 \vee \cdots \vee \varphi_n$ is unsatisfiable as well.

\paragraph{} It remains to discuss individual parts of the presented algorithm, in details. First of all, the algorithm is based on symbolic execution. Therefore, we need a formal definition of a symbolic expressions and symbolic state. We provide the definitions in Sections~\ref{sec:SymbolExpr} and~\ref{sec:SymbolState}. Since we symbolically execute backbone paths of a program, we provide their compact representation in a tree structure, called a backbone tree. The definition of the tree and its construction can be found in Section~\ref{sec:BBT}. The key property of the algorithm is a collection of  path conditions computed along backbone path. Since we work intensively with their structure we decompose their structure along vertices of a backbone tree. Therefore we discuss definition and handling path conditions separately in Section~\ref{sec:PathCond}. Symbolic execution of a backbone tree is then depicted in Section~\ref{sec:ExecBBT}. The key part of the algorithm -- the computation of an over-approximation of a loop at an entry vertex -- is described in details in Section~\ref{sec:OverApproxLoop}. And finally, an algorithm building the formula $\hat{\varphi}$ from results of the symbolic execution of a backbone tree is described in Section~\ref{sec:BuildResult}.

%%%%%%%%%%%%%%%%%%%%%%%%%%%%%%%%%%%%%%%%%%%%%%%%%%%%%%%%%%%%%%%%%%%%%%%%
%%%%%%%%%%%%%%%%%%%%%%%%%%%%%%%%%%%%%%%%%%%%%%%%%%%%%%%%%%%%%%%%%%%%%%%%
\subsection{Symbolic Expressions}
\label{sec:SymbolExpr}

\paragraph{Symbolic expressions} Let $P$ be a program, $\mathsf{V}$ be a set of variable names such that $\mathcal{V}_P \subseteq \mathsf{V}$, and $T_P$ be a first order theory that captures the constants of $\mathcal{I}_P$ (like \texttt{0,1,true}, etc.), the functions of $\mathcal{I}_P$ (like \texttt{+,-}, etc.), predicates of $\mathcal{I}_P$ (like \texttt{<,=}, etc.), and it also is a combination of several theories including theory of equality and uninterpreted functions, and theory of integers. We extend $T_P$ as follows
\begin{itemize}
	\item[(1)] For each program variable $\var{a} \in \mathsf{V}$ of a scalar type $\tau$ we introduce a new constant symbol $\sym{a}$ ranging over data domain of the type $\tau$.
	\item[(2)] For each program variable $\var{A} \in \mathsf{V}$ of an array type $\texttt{int}^n \rightarrow \tau$ we introduce a new function symbol $\sym{A}$ identifying a function from $n$-tuples of integers into data domain of the type $\tau$.
	\item[(3)] For all data types of $P$ we extend their data domains such that they have a special new value $\bot$ in common. We also introduce constant symbol $\star$, which is supposed to be always interpreted to $\bot$.
	\item[(4)] For all three terms $t, t', t''$ of extended $T_P$ such that $t$ is of type \texttt{bool}, and $t', t''$ has a same type $\tau$, we introduce term $\ite(t,t',t'')$ of type $\tau$, whose value is $t'$ if $t$ is $\true$, and $t''$ otherwise.
	\item[(5)] If $t$ is a term of the extended $T_P$ containing symbol $\star$ in it, then we require that $t = \star$ is a valid formula in the theory. If $p$ is a predicate symbol of the extended $T_P$ containing symbol $\star$ as one of its arguments, then we require that $p \leftrightarrow A$ is a valid formula in the theory, where $A$ is a fresh propositional variable (in other words, $p$ can be replaced by a fresh propositional variable).
\end{itemize}
A set $\mathcal{S}_P(\mathsf{V})$ of all terms and formulae of extended $T_P$ is a \emph{set of symbolic expressions of a program} $P$. Each $e \in \mathcal{S}_P(\mathsf{V})$ is a \emph{symbolic expression of a program} $P$. Note that each such $e$ has its type (i.e.~if $e$ is a term, then type of $e$ is a type of an element of a data domain defined by any interpretation of $T_P$, and if $e$ is a formula, then type of $e$ is \texttt{bool}). When program $P$ is known from a context, then we write $\mathcal{S}(\mathsf{V})$. And if we do not care what superset of $\mathcal{V}_P$ the set $\mathsf{V}$ exactly is, then we omit it as well. So we write $\mathcal{S}_P$  or even $\mathcal{S}$ (when $P$ is known from a context).

\paragraph{Basic symbols and variables of basic symbols} Let $P$ be a program and $\mathsf{V}$ be a set of variable names such that $\mathcal{V}_P \subseteq \mathsf{V}$. Then $\Sigma(\mathcal{S}_P(\mathsf{V})) = \{ \sym{a} \mid \var{a} \in \mathsf{V} \}$ is a set of \emph{basic symbols} of $\mathcal{S}_P(\mathsf{V})$ and $\mathcal{V}(\mathcal{S}_P(\mathsf{V})) = \mathsf{V}$ is a set of \emph{variables of basic symbols} of $\mathcal{S}_P(\mathsf{V})$.

\paragraph{Substitution into symbolic expression} Let $h, e, e' \in \mathcal{S}_P$ be symbolic expressions of $P$ of the same type. Then $h[e/e']$ is such a symbolic expression $h$, where all occurrences of $e$ in $h$ were replaced by the expression $e'$. An expression $h[e_1/e'_1, \ldots, e_n/e'_n]$ denotes simultaneous substitution of all pairs $e_i / e'_i$ in $h$.

\paragraph{Expression equivalence} Let $e, e'\in \mathcal{S}_P$ be a two symbolic expressions of $P$. Then $e$ is \emph{equal} to $e'$, if (1) $e, e'$ are both terms of extended $T_P$ and $e = e'$ is a valid formula in extended $T_P$, or (2) $e, e'$ are both formulae of extended $T_P$ and $e \leftrightarrow e'$ is a valid formula in extended $T_P$. \todo{doplnit par prikladu na ekvivalentni vyrazy: $\sym{a} \equiv \sym{a} + 1 - 1$, atd.}

\paragraph{Special variable names of $T_P$} Let $P$ be a program. We distinguish the following sets: (1) $\mathcal{K} = \{ \kappa_i \mid i \in \N\}$ is a set of \emph{path counters}. Each $\kappa_i$ is a variable of $T_P$ and it ranges over $\N_0$. (2) $\mathcal{T} = \{ \tau_i \mid i \in \N\}$ is a set of \emph{parameters}. Each $\tau_i$ is a variable of $T_P$ and it ranges over $\N_0$. (3) $\mathcal{X} = \{ \chi_i \mid i \in \N\}$ is a set of \emph{argument placeholders}. Each $\chi_i$ is a variable of $T_P$ and it ranges over integers. We assume all the sets are disjunctive. We further use the following notation. Let $e \in \mathcal{S}_P$ be a symbolic expression. Then we denote by $\mathcal{K}(e)$ a set of all path counters appearing in $e$, and by $\mathcal{T}(e)$ a set of all the parameters appearing in $e$. 

\paragraph{$\tau$-substitution} Let $e, h\in \mathcal{S}_P$ be two symbolic expressions of $P$ and $\{ \tau_1, \ldots, \tau_n \} \subseteq \mathcal{T}$ be all parameters contained in them. And let $g \in \mathcal{S}_P$ be any symbolic expression containing none of the parameters $\{ \tau_1, \ldots, \tau_n \}$. if both $h$ and $g$ are of the same integer type, then $e\{h/g\}$ is a symbolic expression computed from $e$ as follows. Let $e'$ be a symbolic expression equal to $e$ with the same parameters and same number of their occurrences as in $e$ and with a maximal number of occurrences of $h$ as subexpressions. Then $e\{h/g\} = e'[h/g][\tau_1/\star, \ldots, \tau_n/\star]$. \todo{doplnit par prikladu: $(\sym{a} + \tau)\{\tau + 2 / \chi \} = (\sym{a} + \tau + 2 - 2)[\tau + 2 / \chi][\tau/\star] = \sym{a} + \chi - 2$, atd.} We naturally extend the subexpression substitution to vector expressions: $e\{\vec{h}/\vec{g}\}$ is an expression $e\{h_1/g_1\}\ldots\{h_n/g_n\}$. Note that we require that vectors $\vec{h}$ and $\vec{g}$ have the same dimension.

\paragraph{Comparison of vectors of symbolic expressions} Let $\vec{u} = (u_1, \ldots, u_n)$ and $\vec{v} = (v_1, \ldots, v_n)$ be two vectors of some symbolic expressions $u_1, \ldots, u_n$ and $v_1, \ldots, v_n$ respectively. Then we use the following notation
\begin{align*}
\vec{u} \leq \vec{v} & \equiv u_1 \leq v_1 \wedge \ldots \wedge u_n \leq v_n \\
\vec{u} < \vec{v} & \equiv \vec{u} \leq \vec{v} \wedge \sum_{i = 1}^n{u_i} < \sum_{i = 1}^n{v_i}
\end{align*}

%%%%%%%%%%%%%%%%%%%%%%%%%%%%%%%%%%%%%%%%%%%%%%%%%%%%%%%%%%%%%%%%%%%%%%%%
%%%%%%%%%%%%%%%%%%%%%%%%%%%%%%%%%%%%%%%%%%%%%%%%%%%%%%%%%%%%%%%%%%%%%%%%
\subsection{Symbolic State}
\label{sec:SymbolState}

\paragraph{Symbolic noname functions} Let $P$ be a program, $\mathsf{V}$ be a set of variable names such that $\mathcal{V}_P \subseteq \mathsf{V}$. Then $\mathcal{S}_{\lambda P}(\mathsf{V}) = \{ \lambda \chi_1, \ldots, \chi_n~.~e \mid e \in \mathcal{S}_P(\mathsf{V}) \wedge n \in \N \wedge \chi_1, \ldots, \chi_n \in \mathcal{X} \}$ is a set of \emph{symbolic noname functions of $\mathcal{S}_P(\mathsf{V})$}. Let $\lambda \chi_1, \ldots, \chi_n~.~e \in \mathcal{S}_{\lambda P}(\mathsf{V})$ and $e, e_1, \ldots, e_n \in \mathcal{S}_P(\mathsf{V})$. Then $(\lambda \chi_1, \ldots, \chi_n~.~e)(e_1, \ldots, e_n) \in \mathcal{S}_P(\mathsf{V})$ is a symbolic expression $e[\chi_1/e_1, \ldots, \chi_n/e_n]$. When program $P$ is known from a context, then we write $\mathcal{S}_\lambda(\mathsf{V})$. And if we do not care what superset of $\mathcal{V}_P$ the set $\mathsf{V}$ exactly is, then we omit it as well. So we write $\mathcal{S}_{\lambda P}$  or even $\mathcal{S}_\lambda$ (when $P$ is known from a context).

\paragraph{Symbolic state} Let $P$ be a program, $\mathsf{V}$ be a set of variable names such that $\mathcal{V}_P \subseteq \mathsf{V}$. A function $\theta : \mathsf{V} \rightarrow \mathcal{S}_P(\mathsf{V}) \cup \mathcal{S}_{\lambda P}(\mathsf{V})$ is a \emph{symbolic state of} $P$, if it satisfies the following
\begin{itemize}
	\item If $\var{a} \in \mathsf{V}$ is of a scalar type $\tau$, then $\theta(\var{a}) \in \mathcal{S}_P(\mathsf{V})$ is also of a scalar type $\tau$.
	\item If $\var{A} \in \mathsf{V}$ is of an array type $\texttt{int}^n \rightarrow \tau$, then $\theta(\var{A}) \in \mathcal{S}_{\lambda P}(\mathsf{V})$ is also of a type $\texttt{int}^n \rightarrow \tau$, and it is of a form $\lambda \chi_1, \ldots, \chi_n~.~e$, for some $e \in \mathcal{S}_P(\mathsf{V})$ of type $\tau$. We often use an abbreviated vector notation $\lambda \vec{\chi}~.~e$.
\end{itemize}

\paragraph{Symbolic states} Let $P$ be a program, $\mathcal{S}_P$ be a set of symbolic expressions and $\mathcal{S}_{\lambda P}$ be a set of symbolic noname functions of $\mathcal{S}_P$. Then we denote by $\mathcal{M}(\mathcal{S}_P) = \{\theta \mid \theta : \mathcal{V}(\mathcal{S}_P) \rightarrow \mathcal{S}_P \cup \mathcal{S}_{\lambda P} \}$ a set of \emph{symbolic states of $P$}.

\paragraph{The most general symbolic state} Let $P$ be a program. We distinguish a special symbolic state $\theta_G$ of $P$. It has the following properties:
\begin{itemize}
  \item For each $\var{a} \in \mathcal{V}(\mathcal{S}_P)$ of a scalar type we have $\theta_G( \var{a} ) = \sym{a}$.
  \item For each $\var{A} \in \mathcal{V}(\mathcal{S}_P)$ of an array type we have $\theta_G( \var{A} ) = \lambda \vec{\chi}~.~\sym{A}(\vec{\chi})$.
\end{itemize}

\paragraph{The most unknown symbolic state} Let $P$ be a program. We distinguish a special symbolic state $\theta_\star$ of $P$. It has the following properties:
\begin{itemize}
  \item For each $\var{a} \in \mathcal{V}(\mathcal{S}_P)$ of a scalar type we have $\theta_\star( \var{a} ) = \star$.
  \item For each $\var{A} \in \mathcal{V}(\mathcal{S}_P)$ of an array type we have $\theta_\star( \var{A} ) = \lambda \vec{\chi}~.~\star$.
\end{itemize}

\paragraph{Substitution into symbolic state} Let $\theta$ be a symbolic state of $P$ and $e, e'$ some symbolic expressions of $P$ of the same type. Then $\theta[e/e']$ is a symbolic state of $P$ such that for each variable $\var{a} \in \mathcal{V}(\mathcal{S}_P)$ we have $\theta[e/e'](\var{a}) = \theta(\var{a})[e/e']$. A symbolic state $\theta[e_1/e'_1, \ldots, e_n/e'_n]$ denotes simultaneous substitution of all pairs $e_i / e'_i$ into $\theta$.

\paragraph{Change in symbolic state} Let $\theta$ be a symbolic state of $P$, $\var{a} \in \mathcal{V}(\mathcal{S}_P)$ be a program variable of a scalar type $\tau$, $\var{A} \in \mathcal{V}(\mathcal{S}_P)$ be a program variable of an array type $\texttt{int}^n \rightarrow \tau$, and $e$ be a symbolic expression of $P$ of the type $\tau$. Then $\theta[\var{a} \rightarrow e]$ is a symbolic state equal to $\theta$ except for variable \var{a}, where $\theta[\var{a} \rightarrow e](\var{a}) = e$, and $\theta[\var{A} \rightarrow e]$ is a symbolic state equal to $\theta$ except for variable \var{A}, where $\theta[\var{A} \rightarrow e](\var{A}) = \lambda \chi_1, \ldots, \chi_n~.~e$.

\paragraph{Extending symbolic state to program expressions} Let $\theta$ be a symbolic state of $P$ and $e \in \mathcal{E}_P$ be a program expression. Then $\theta(e) \in \mathcal{S}_P$ is a symbolic expression received from $e$ such that (1) Each occurrence of each variable \var{a} appearing in $e$ is replaced by symbolic expression $\theta(\var{a})$, where we assume that all substitutions are applied simultaneously. (2) We replace all constant, operator and function symbols appearing in $e$ by their counterparts in $T_P$.

\paragraph{Substituting symbolic state into symbolic expression} Let $e \in \mathcal{S}_P$ and $\theta$ be a symbolic state of $P$. Then $e\theta \in \mathcal{S}_P$ is a symbolic expression received from $e$ such that each occurrence of each basic symbol $\sym{a} \in \Sigma(S_P)$ appearing in $e$ is replaced by a symbolic expression $\theta(\var{a})$. We assume that all the substitutions are applied simultaneously.

\paragraph{Merge of symbolic states} Let $\theta$ and $\theta'$ be two symbolic states of $P$. Then $\theta \theta'$ denotes a symbolic state of $P$ such that for each program variable \var{a} we have $(\theta \theta')(\var{a}) = (\theta(\var{a})) \theta'$.

%%%%%%%%%%%%%%%%%%%%%%%%%%%%%%%%%%%%%%%%%%%%%%%%%%%%%%%%%%%%%%%%%%%%%%%%
%%%%%%%%%%%%%%%%%%%%%%%%%%%%%%%%%%%%%%%%%%%%%%%%%%%%%%%%%%%%%%%%%%%%%%%%
\subsection{Backbone Tree}
\label{sec:BBT}

Any two backbone paths of a program $P$ always have some non-empty prefix in common. Therefore, we effectively store the backbone paths in a tree defined as follows.

\paragraph{Backbone tree of $P$} Let $V_\bbt$ be a set of all non-empty prefixes of backbone paths of a program $P$. Let $E_\bbt \subseteq V_\bbt \times V_\bbt$ be a set of all pairs $(\alpha,\alpha v)$. Then we call a rooted tree $\bbt_P=(V_\bbt,E_\bbt,l_s)$, where $l_s$ is the root, a \emph{backbone tree} of $P$. \note{Honza: Definici $\bbt$ by slo zredukovat na pouhou $V_\bbt$ (tj. $\bbt$ je $V_\bbt$), protoze vnitrni struktura vrcholu ty hrany vlastne presne definuje} Note that vertices of $\bbt_P$ identify acyclic paths from $l_s$ in $P$. We denote the set of all leaf vertices of $\bbt_P$ by $B_\bbt$. Note that $B_\bbt$ is actually a set of all the backbone paths of $P$. When a program is known from a context or it is not important, then we simply write $V, E, B$ and $\bbt$. Algorithm \ref{alg:buildBBT} computes $\bbt$ for a program $P$. \todo{doplnit slovni popis toho algoritmu a spocitat jeho  casovou a prostorovou slozitost}

\begin{algorithm}[!htb]
\newcommand{\buildBBT}{\texttt{buildBackboneTree($P$)}}
\caption{\buildBBT\label{alg:buildBBT}}
\KwIn{ \aarg{$P$}{a normalized program} }
\KwOut{ \aarg{$\bbt$}{a backbone tree for $P$} }
\BlankLine
$V_\bbt \aset \{ l_s \}$\;
$E_\bbt \aset \emptyset$\;
$D \aset \emptyset$\;
\While{there is a leaf $\alpha u \in V_\bbt \smallsetminus D$ such that $u \ne l_t$}{
  \ForEach{vertex $v$ such that $(u,v) \in E_P$}{
    \If{$\exists k$ such that $v = (\alpha u)(k)$}{
      $D \aset D \cup \{ \alpha u \}$
    } \Else {
      $V_\bbt \aset V_\bbt \cup \{ \alpha u v \}$\;
      $E_\bbt \aset E_\bbt \cup \{ (\alpha u, \alpha u v) \}$\;
    }
  }
}
\While{$D \ne \emptyset$}{
  $\alpha u$ \aset any element of $D$\;
  $D \aset D \smallsetminus \{ \alpha u \}$\;
  \If{$\alpha u$ is a leaf vertex in current $\bbt$}{
    $D \aset D \cup \{ \alpha \}$\;
    $E_\bbt \aset E_\bbt \smallsetminus \{ (\alpha, \alpha u) \}$\;
    $V_\bbt \aset V_\bbt \smallsetminus \{ \alpha u \}$\;
  }
}
\Return{$\bbt$}\;
\end{algorithm}

\paragraph{Loop entry vertex, and loop} Let $\bbt$ be a backbone tree of a program $P$. Then each vertex $\alpha v \in V(\bbt)$ such that $v$ is an entry vertex of $P$ is a \emph{loop entry vertex of $\bbt$}. Let $\mathcal{C}$ be a loop at $v$. Then $\mathcal{C}$ is also a \emph{loop at $\alpha v$}.

\paragraph{Counting backbone paths of induced programs} Let $\bbt$ be a backbone tree of a program $P$. We define a function $\eta_\bbt : V_\bbt \rightarrow \N$ as follows. Let $\alpha v \in V_\bbt$ be a vertex of $\bbt$. If $\alpha v$ is not a loop entry vertex $\bbt$, then $\eta_\bbt(\alpha v) = 0$. Otherwise, let $\mathcal{C}$ be a loop at a loop entry vertex $v$, and $B_{\bbt'}$ be a set of all leaf vertices of a backbone tree $\bbt'$ of an induced program $P(\mathcal{C},v)$. Then $\eta_\bbt(\alpha v) = |B_{\bbt'}|$. When $\bbt$ is known from a context we write $\eta$.

%%%%%%%%%%%%%%%%%%%%%%%%%%%%%%%%%%%%%%%%%%%%%%%%%%%%%%%%%%%%%%%%%%%%%%%%
%%%%%%%%%%%%%%%%%%%%%%%%%%%%%%%%%%%%%%%%%%%%%%%%%%%%%%%%%%%%%%%%%%%%%%%%
\subsection{Path Condition}
\label{sec:PathCond}

\paragraph{Function $\Psi$} Let $\bbt$ be a backbone tree of a program $P$. In Section~\ref{sec:ExecBBT} we show, how to execute $\bbt$ symbolically. In our analysis path conditions from these executions play a crucial role. Since we work with them intensively, and we examine their internal structure, it is not effective to represent them as a whole formulae (as typical in original symbolic execution). We rather attach their parts to vertices of $\bbt$. This can be explain as follows. Let $\pi \in B_\bbt$ be a backbone path of $P$. When executing $\pi$ symbolically, we execute instructions occurring along the path $\pi$. Execution of some instructions may cause extension of current path condition $\varphi$ by some formula $\gamma$ such that extended path condition is of a form $\varphi \wedge \gamma$. Other instruction only change symbolic state, but keep path condition $\varphi$ unchanged. To unify the approach for all instructions, we want that also these instructions extend path condition $\varphi$ by some formula $\gamma$. If the formula $\gamma \equiv \true$ for these instructions, then we are done. Now we can assign to each vertex along $\pi$ a formula $\gamma$ received from executing an instruction. More precisely, path condition is initially set to $\true$. Therefore, we assign formula $\true$ to the first vertex of $\pi$. Now suppose that symbolic execution reached a vertex $\alpha u$ of $\pi$ and $\alpha u v$ is next one in $\pi$. Then execution of an instruction $\iota((u,v))$ produce a formula $\gamma$ which we attach to the vertex $\alpha u v$. It is important to note, that we can always reconstruct actual path condition $\varphi$ in each step of symbolic execution from the formulae attached to vertices along currently processed path such that we return conjunction of those formulae.

The situation is different in loop entry vertices of $\bbt$. There we enter a loop and we call the over-approximation algorithm. The result of the call is a single (over-approximated) path condition and single(over-approximated) symbolic state. We assign the resulting formula to the loop vertex. Note that there is always place for the formula, since we assume only normalized programs, so all in-edges to loop entry vertices are labelled with \texttt{skip} instruction.

So all parts of path conditions can indeed be assigned to vertices of $\bbt$. We formally introduce a function $\Psi_\bbt: V_\bbt \rightarrow \mathcal{S}_P$ assigning each vertex of $\bbt$ an symbolic expression of type \texttt{bool}. We build a content of the function during symbolic executions of backbone paths of $\bbt$. We discuss a details of the execution in Section~\ref{sec:ExecBBT}. But since $\Psi_\bbt$ contains formula from which we construct the path conditions, therefore this function is a key property of whole algorithm. When a backbone tree is known from a context we simply write $\Psi$.

\paragraph{Path counters at loop entry vertex} The key part of the algorithm is computation of an over-approximation of a loop at some loop entry vertex. We already know that we compute the over-approximation such that we express values of program variables as functions of how many times backbone paths of induced program of the loop are executed. For this purpose we introduce for each such a backbone path a single and unique path counter. A path counter is a variable of a theory $T_P$ of an integer type. We have already distinguish the infinite set $\mathcal{K}$ of variable symbols for the path counters.

For each loop entry vertex of $\bbt$ we know exactly how many fresh path counters we need to introduce. The count is equal to a number of backbone paths of an induced program at the loop entry vertex. We use the following naming convention for identifying path counters introduced at a loop entry vertices: Let $\alpha$ be a loop entry vertex of $\bbt$. Then we identify the fresh paths counters introduced at $\alpha$ as $\kappa_{\alpha,1}, \ldots, \kappa_{\alpha, \eta(\alpha)}$. We assume, that order of backbone paths in induced program is fixed to provide unique mapping between the path counters and related backbone paths.

\paragraph{Path condition part at vertex of $\bbt$} Let $\alpha \in V_\bbt$ be a vertex of $\bbt$ and $\vec{\kappa}_\alpha = (\kappa_{\alpha,1}, \ldots, \kappa_{\alpha,\eta(\alpha)} )^T$ identify all the path counters introduced at $\alpha$. Then formula
\begin{equation*}
\hat{pc}(\alpha, \Psi, \varphi) \equiv
\begin{cases}
\varphi & \alpha = \varepsilon \\
\Psi(\alpha) \wedge \varphi & \alpha \ne \varepsilon \wedge \eta(\alpha) = 0 \\
\exists \vec{\kappa}_\alpha~(\vec{0} \leq \vec{\kappa}_\alpha \wedge \Psi(\alpha) \wedge \varphi) & \text{Otherwise},
\end{cases}
\end{equation*}
is a \emph{path condition part at vertex} $\alpha$. When a backbone tree is known from a context we simply write $\hat{pc}(\alpha, \Psi, \varphi)$. Note that $\hat{pc}$ has additional parameter $\varphi$ to allow insertion of a formula into a scope of the existential quantifier introduced in the last case of the definition.

\paragraph{Path condition at vertex of $\bbt$} Let $v_1 v_2 \cdots v_k \in V_\bbt$, where $k > 0$, be a vertex of $\bbt$. Then recursively defined formula
\begin{equation*}
pc_\bbt(v_1 v_2 \cdots v_k, \Psi) \equiv \hat{pc}_\bbt(v_1, \Psi, \hat{pc}_\bbt(v_1 v_2, \Psi, \ldots \hat{pc}_\bbt(v_1 v_2 \cdots v_k, \Psi, \true) \ldots ))
\end{equation*}
is a \emph{path condition at vertex} $v_1 v_2 \cdots v_k$. When a backbone tree is known from a context we simply write $pc(\alpha, \Psi)$.

%%%%%%%%%%%%%%%%%%%%%%%%%%%%%%%%%%%%%%%%%%%%%%%%%%%%%%%%%%%%%%%%%%%%%%%%
%%%%%%%%%%%%%%%%%%%%%%%%%%%%%%%%%%%%%%%%%%%%%%%%%%%%%%%%%%%%%%%%%%%%%%%%
\subsection{Symbolic Execution of Backbone Tree}
\label{sec:ExecBBT}

Let $\bbt$ be a backbone tree of a program $P$. To execute a $\bbt$ symbolically means that we symbolically execute all its backbone paths. To symbolically execute a backbone path $\pi$ of $\bbt$ means the following. We start at the first vertex $l_s$ of $\pi$. There we set $\Psi(l_s) = \true$ and we set actual symbolic state $\theta$ to be the most general one, i.e.~$\theta_G$. Then we proceed along $\pi$ per vertex until we process the last one. Let $u$ be a vertex of $\pi$ lastly processed and let $v$ be its successor in $\pi$. If $v$ is a loop entry vertex of $P$, then we call an algorithm, depicted in details in Section~\ref{sec:OverApproxLoop}, computing an over-approximation of the loop. If $v$ is not a loop vertex of $P$, then we symbolically execute instruction $\iota((u,v))$. We discuss symbolic execution of individual instructions in details later in this section. In both cases we receive a formula which we put into $\Psi$ and we also receive updated symbolic state. Then we proceed to another vertex of $\pi$ with the updated state. It may also happen at some vertex during symbolic execution of $\pi$ that path condition, composed of formulae assigned to already processed vertices of $\pi$, is not satifiable. Then there is no feasible path in $P$ reducible to $\pi$. Therefore, we stop the execution at that vertex. We can also remove this path $\pi$ from the tree $\bbt$, since we have discovered it is useless for reachability of the target location of $P$.

In Algorithm~\ref{alg:executeBBT} we present symbolic execution of $\bbt$ in more details. The algorithm works as described above. But we do not execute backbone paths separately one by one. We rather execute them simultaneously, all at once. Therefore, we maintain a set $Q$ of lastly processed vertices of all backbone paths. Since we also need to save actual symbolic states at those vertices, the elements of $Q$ are actually pairs, i.e.~vertex plus symbolic state. Another difference is, that the algorithm also computes function $\Theta$ assigning final symbolic states to leaves of $\bbt$. This function is a by product of the algorithm. It is only used by the over-approximation algorithm of Section~\ref{sec:OverApproxLoop}. There it is used to compute an over-approximated symbolic state such that a backbone tree of induced program of a loop to be over-approximated is symbolically executed first (by this algorithm). Let us discuss all three cases which may occur at each vertex during the execution.

At line~\ref{l:loopVtxOrNot} we determine, whether successor vertex $\alpha u v$ of $\alpha u$ is a loop entry or not. If so, then we identify a loop $\mathcal{C}$ at $v$ and at line~\ref{l:callOverApprox} we call the over-approximation algorithm \texttt{overapproximateLoop}, discussed in Section~\ref{sec:OverApproxLoop}, to obtain a formula $\varphi^{\vec{\kappa}}$, which is an over-approximation of path conditions of all feasible paths looping in $\mathcal{C}$, and symbolic state $\theta^{\vec{\kappa}}$, which is an over-approximation of all changes in symbolic state made by all feasible paths looping in $\mathcal{C}$. Having these over-approximations, we need to integrate them into current symbolic execution. It means, that we assign the formula $\varphi^{\vec{\kappa}}$ into function $\Psi$ at vertex $\alpha u v$, and we store $\alpha u v$ plus $\theta^{\vec{\kappa}}$ to be later able to process successors of $\alpha u v$ in $\bbt$. Note that both $\varphi^{\vec{\kappa}}$ and $\theta^{\vec{\kappa}}$ are updated by symbolic state $\theta$ before they are integrated. This is because the over-approximation of $\mathcal{C}$ is computed independently form the remainder of $P$. And symbolic state $\theta$ captures the current progress of symbolic execution up to the loop vertex $v$. We need to incorporate that progress into the over-approximation, before we integrate it into symbolic execution of $\bbt$.

If $\alpha u v$ is not a loop entry vertex, then we must symbolically execute an instruction $\iota((u,v))$ labelling a program edge $(u,v)$. Since it is purely technical matter, we leave its detailed description to the end of this section. Having the instruction executed we receive a formula representing an add-on to a current path condition. Therefore, we can directly assign it into $\Psi$ at $\alpha u v$. As the second value form execution of $\iota((u,v))$, we receive an updated symbolic state, capturing an effect of $\iota((u,v))$ on original symbolic state $\theta$. Next we check, whether a path condition, composed of all formulae assigned to vertices along the path $\alpha u v$ so far.

Let us suppose first the path condition is satisfiable. If we have not reached the target vertex yet, we store current progress in $Q$. Otherwise we store final symbolic state into function $\Theta$ for the leaf $\alpha u v$ and we are done executing current backbone path.

In case the path condition is not satisfiable, we stop symbolic execution at $\alpha u v$. We know that any further progress form $\alpha u v$ along any backbone path with prefix $\alpha u v$ cannot represent feasible path to the target location. Therefore we can reduce $\bbt$ such that we remove from it exactly those backbone paths with a prefix $\alpha u v$, while keeping there all the others. Such reduced set of vertices of $\bbt$ is computed at line~\ref{l:pruneBBTVtxs}. Then we need to update all the remaining sets forming $\bbt$. We cannot forget to update also function $\Psi$ to be defined only on proper set of vertices of $\bbt$ at the end.

\paragraph{} Note that symbolic execution of $\bbt$ is always finite, since $\bbt$ is a finite binary tree of backbone paths and the same holds for backbone tree of induced programs of its loops. There is a finite number of loops in a program.

\begin{algorithm}[!htb]
\newcommand{\executeBBT}{\texttt{executeBackboneTree}}
\caption{\executeBBT\texttt{(}$\bbt,P$\texttt{)}\label{alg:executeBBT}}
\KwIn{
\aargm{$\bbt$}{a backbones tree of $P$}
\aargm{$P$}{a normalized program}
}
\KwOut{
\aargm{$\Psi$}{function assigning parts of path conditions to vertices of $\bbt$}
\aargm{$\Theta: B \rightarrow \mathcal{M}(\mathcal{S})$}{final symbolic states at leaves of $\bbt$}
}
\BlankLine
$\Psi$ \aset $\{ (l_s, \true) \}$\;
$\Theta$ \aset $\emptyset$\;
$Q$ \aset $\{ (l_s, \theta_G) \}$\;
\Repeat{ $Q = \emptyset$ }{
  $(\alpha u, \theta)$ \aset any element of $Q$\;
  $Q$ \aset $Q \smallsetminus \{ (\alpha u, \theta) \}$\;
  \ForEach{$\alpha u v \in V_\bbt$}{
    \If{$\alpha u v$ is a loop entry vertex of $\bbt$}{ \label{l:loopVtxOrNot}
      Let $\mathcal{C}$ be a loop at $v$\;
      $(\varphi^{\vec{\kappa}}, \theta^{\vec{\kappa}})$ \aset $\texttt{overapproximateLoop}(\mathcal{C},v)$\; \label{l:callOverApprox}
      $\Psi(\alpha u v)$ \aset $\varphi^{\vec{\kappa}}\theta$\;
      $Q$ \aset $Q \cup \{ (\alpha u v, \theta^{\vec{\kappa}}\theta) \}$\;
    }\Else{
      $(\Psi(\alpha u v), \theta) \aset \iota((u,v))(\theta, pc(\alpha u, \Psi))$\;
      \If{$pc(\alpha u v, \Psi)$ is satisfiable}{
	      \If{ $v \ne l_t$ }{
	        $Q$ \aset $Q \cup \{ (\alpha u v, \theta) \}$\;
	      }
	      \Else{
	        $\Theta(\alpha u v)$ \aset $\theta$\;
	      }
	    } \Else{
	      $V_\bbt \aset \{ \beta \mid \exists \pi \in B_\bbt \wedge \alpha u v \not \in \pref(\pi) \wedge \beta \in \pref(\pi) \}$\; \label{l:pruneBBTVtxs}
	      $E_\bbt \aset {E_\bbt}|_{V_\bbt}$\;
	      $B_\bbt \aset V_\bbt \cap B_\bbt$\;
	      $\Psi \aset \Psi|_{V_\bbt}$\;
	    }
    }
  }
}
\Return{$(\Psi, \Theta)$}\;
\end{algorithm}

\paragraph{Symbolic execution of a program instruction} Let $P$ be a program, $\varphi$ be a symbolic expression of $P$ of \texttt{bool} type representing a path condition, $\theta$ be a symbolic state of $P$, and let $I$ be an instruction. Then we compute a result $I(\theta, \varphi)$ of symbolic execution of $I$ in $\theta$ and $\varphi$ according to a syntax structure of $I$ as follows. We assume \var{a} is a variable a scalar type $\tau$, \var{A} is a variable an array type $\texttt{int}^n \rightarrow \tau$, $\gamma$ is a program expression of type \texttt{bool}, $e$ is a program expression of type $\tau$, and $e_1, \ldots, e_n$ are program expressions of $P$ of type \texttt{int}.
\begin{itemize}
	\item $I$ is an assumption $\texttt{assume}(\gamma)$: If a formula $\varphi \rightarrow \theta(\gamma)$ is satisfiable, then $I(\theta, \varphi)$ is a pair $(\theta(\gamma), \theta)$, and $(\false, \theta)$ otherwise.
	\item $I$ is an assertion $\texttt{assert}(\gamma)$: If a formula $\varphi \rightarrow \theta(\gamma)$ is valid, then $I(\theta, \varphi)$ is a pair $(\true, \theta)$, and $(\false, \theta)$ otherwise.
	\item $I$ is an instruction \texttt{skip}: Then $I(\theta, \varphi)$ is a pair $(\true, \theta)$.
	\item $I$ is an assignment $\var{a} \aset e$: Then $I(\theta, \varphi)$ is a pair $(\true, \theta[\var{a} \rightarrow \theta(e)])$.
	\item $I$ is an assignment $\var{A}(e_1, \ldots, e_n) \aset e$: Then $I(\theta, \varphi)$ is a pair $(\true, \theta[\var{A} \rightarrow \ite( \chi_1 = \theta(e_1) \wedge \cdots \wedge \chi_n = \theta(e_n), e, \theta(\var{A})(\chi_1, \ldots, \chi_n))])$.
\end{itemize}

%%%%%%%%%%%%%%%%%%%%%%%%%%%%%%%%%%%%%%%%%%%%%%%%%%%%%%%%%%%%%%%%%%%%%%%%
%%%%%%%%%%%%%%%%%%%%%%%%%%%%%%%%%%%%%%%%%%%%%%%%%%%%%%%%%%%%%%%%%%%%%%%%
\subsection{Building $\hat{\varphi}$}
\label{sec:BuildResult}

After symbolic execution of a backbone tree $\bbt$ of a program $P$ we have computed all the information we need to build resulting over-approximation $\hat{\varphi}$ of $\Psi_P$. The information is stored in function $\Psi$ as formulae attached to vertices of $\bbt$. We know, that for a backbone path $\pi \in B_\bbt$ a formula $pc_\bbt(\pi,\Psi)$ is an over-aproximated path condition for all feasible paths reducible to $\pi$. Therefore, the over-approximation $\hat{\varphi}$ is given by the formula
\begin{equation*}
\hat{\varphi} \equiv
\begin{cases}
\false & B_\bbt = \emptyset \\
\bigvee_{\pi \in B_\bbt} pc_\bbt(\pi,\Psi) & Otherwise.
\end{cases}
\end{equation*}
Since backbone paths have always non-empty common prefix, it is usually the case that we can simplify the formula $\hat{\varphi}$. For each pair of backbone paths we move common part of their path conditions in front of the disjunction of their remainders.

Observe, that composition of backbone paths in $\bbt$ precisely matches structure of such simplified formula. Therefore, we can infer a simple algorithm on $\bbt$, which build $\hat{\varphi}$ in already simplified form. We depict its pseudo-code in Algorithm~\ref{alg:buildSimplified}. The algorithm is recursive. It accepts the backbone tree $\bbt$, function $\Psi$ already filled in during symbolic execution of $\bbt$, and a vertex $\alpha$ of $\bbt$. To receive $\hat{\varphi}$ we need to call the algorithm with the root vertex $l_s$.

\begin{algorithm}[!htb]
\newcommand{\buildSimplified}{\texttt{buildSimplified}}
\caption{\buildSimplified\texttt{(}$\bbt, \Psi, \alpha$\texttt{)} \label{alg:buildSimplified}}
\BlankLine
%\lIf{$\alpha = \varepsilon$}{ \label{l:emptyU}
%  \Return{$\false$} \;
%}
$\gamma$ \aset $\ite(\textit{$\alpha$ is a leaf of $\bbt$},\true,\false)$\;
\ForEach{$\alpha v \in V_\bbt$}{
  $\gamma$ \aset $(\gamma \vee \texttt{buildSimplified}(\bbt,\Psi,\alpha v))$\;
}
\Return{$\hat{pc}_\bbt(\alpha, \Psi, \gamma)$}\;
\end{algorithm}

Note that the Algorithm~\ref{alg:buildSimplified} cannot be used, when $\bbt$ is empty. Nevertheless, this case is trivial, since $\hat{\varphi}$ is $\false$. We use Algorithm~\ref{alg:buildSimplified} only for non-empty backbone trees.

%%%%%%%%%%%%%%%%%%%%%%%%%%%%%%%%%%%%%%%%%%%%%%%%%%%%%%%%%%%%%%%%%%%%%%%%
%%%%%%%%%%%%%%%%%%%%%%%%%%%%%%%%%%%%%%%%%%%%%%%%%%%%%%%%%%%%%%%%%%%%%%%%
%%%%%%%%%%%%%%%%%%%%%%%%%%%%%%%%%%%%%%%%%%%%%%%%%%%%%%%%%%%%%%%%%%%%%%%%
\section{Loop Over-approximation}
\label{sec:OverApproxLoop}

Let $\mathcal{C}$ be a loop at a loop entry vertex $v$ of a backbone tree $\bbt$ of program $P$. We want to over-approximate all feasible paths representing all possible looping in $\mathcal{C}$ by a single formula $\varphi^{\vec{\kappa}}$ and single symbolic state $\theta^{\vec{\kappa}}$. The formula $\varphi^{\vec{\kappa}}$ is an over-approximation of path conditions of all those feasible paths and it is supposed to ensure, that none of these feasible paths is early terminated. In other words, it prunes out all those input to the loop $\mathcal{C}$ such that an execution of the loop for any such an input would terminate in some vertex of $\mathcal{C}$ different to $v$. Therefore, we call the formula $\varphi^{\vec{\kappa}}$ a \emph{looping condition of} $\mathcal{C}$. The symbolic state $\theta^{\vec{\kappa}}$ over-approximates all changes into symbolic state which could be made by the feasible paths looping in $\mathcal{C}$. Since its computation is based on expressing values of program variables as functions of how many times backbone paths of induced program of $\mathcal{C}$ are iterated, we call the symbolic state $\theta^{\vec{\kappa}}$ an \emph{iterated symbolic state of} $\mathcal{C}$.

We depict a computation of the over-approximation $(\varphi^{\vec{\kappa}}, \theta^{\vec{\kappa}})$ of $\mathcal{C}$ in Algorithm~\ref{alg:overapproximateLoop}. We first build an induced program $P'$ for the loop $\mathcal{C}$ at $v$ and then we construct a backbone tree $\bbt'$ of $P'$. When we have $\bbt'$, we can execute it symbolically as described in Section~\ref{sec:ExecBBT}. As a result from the execution we receive functions $\Psi'$ and $\Theta'$. At line~\ref{l:loopInfeasible} we resolve a trivial case, when the backbone tree $\bbt'$ becomes empty after its symbolic execution. That indicates, there is no feasible path iterating in $\mathcal{C}$. Therefore, returned value at that line is indeed an over-approximation of $\mathcal{C}$. If $\bbt'$ is not empty, we can proceed further in the computation. We compute the over-approximation $(\varphi^{\vec{\kappa}}, \theta^{\vec{\kappa}})$ from the functions $\Psi'$ and $\Theta'$. First we compute the iterated symbolic state $\theta^{\vec{\kappa}}$. This is done at lines~\ref{l:introArtif}--\ref{l:restrictThetaKappa}. A step at line~\ref{l:introArtif} is technical. The computation of $\theta^{\vec{\kappa}}$ involves presence of some artificial program variables and basic symbols in functions $\Psi'$ and $\Theta'$. To save the original functions, we build copies $\bar{\Psi}$ and $\bar{\Theta}$ of functions $\Psi'$ and $\Theta'$, where we introduce those artificial variables and symbols. The computation of $\theta^{\vec{\kappa}}$ itself is done at line~\ref{l:iterState}. We postpone the detailed description of both the introduction of artificial variables and the computation of $\theta^{\vec{\kappa}}$ into Section~\ref{sec:IterateTheta}. The returned iterated symbolic state $\theta^{\vec{\kappa}}$ is defined also for the artificial variables. Therefore, we restrict $\theta^{\vec{\kappa}}$ into regular program variables at line~\ref{l:restrictThetaKappa}. And finally, having function $\Psi'$ and iterated symbolic state $\theta^{\vec{\kappa}}$ we can compute the looping condition $\varphi^{\vec{\kappa}}$ at line~\ref{l:comLoopCond}. We describe its computation in details in Section~\ref{sec:loopingCondition}.

\begin{algorithm}[!htb]
\newcommand{\overapproximateLoop}{\texttt{overapproximateLoop}}
\caption{\overapproximateLoop\texttt{(}$\mathcal{C},v$\texttt{)}\label{alg:overapproximateLoop}}
\KwIn{
\aargm{$\mathcal{C}$}{a loop at a loop entry vertex $v$ of $\bbt$}
\aargm{$v$}{the loop entry vertex}
}
\KwOut{
\aargm{$\varphi^{\vec{\kappa}}$}{a looping condition of $\mathcal{C}$}
\aargm{$\theta^{\vec{\kappa}}$}{an iterated symbolic state of $\mathcal{C}$}
}
\BlankLine
$P' \aset P(\mathcal{C}, v)$\;
$\bbt'$ \aset \texttt{buildBackboneTree($P'$)}\;
$(\Psi', \Theta')$ \aset \texttt{executeBackboneTree}($\bbt', P'$)\;
\lIf{$V_{\bbt'} = \emptyset$}{\Return{$(\true,\theta_G)$}}\; \label{l:loopInfeasible}
$(\bar{\Psi}, \bar{\Theta}) \aset \texttt{introduceArtificials}(\Psi', \Theta')$\; \label{l:introArtif}
$\theta^{\vec{\kappa}}$ \aset \texttt{computeIteratedState}$(\bar{\Psi}, \bar{\Theta})$\; \label{l:iterState}
$\theta^{\vec{\kappa}} \aset \theta^{\vec{\kappa}}|_{\mathcal{V}}$\; \label{l:restrictThetaKappa} 
$\varphi^{\vec{\kappa}}$ \aset {\it compute looping condition from $\Psi'$ and $\theta^{\vec{\kappa}}$}\; \label{l:comLoopCond}
\Return{$(\varphi^{\vec{\kappa}}, \theta^{\vec{\kappa}})$}\;
\end{algorithm}

%%%%%%%%%%%%%%%%%%%%%%%%%%%%%%%%%%%%%%%%%%%%%%%%%%%%%%%%%%%%%%%%%%%%%%%%
%%%%%%%%%%%%%%%%%%%%%%%%%%%%%%%%%%%%%%%%%%%%%%%%%%%%%%%%%%%%%%%%%%%%%%%%
\subsection{Computation of iterated symbolic state $\theta^{\vec{\kappa}}$}
\label{sec:IterateTheta}

Let $P$ be a program, $\bbt$ be a backbone tree of $P$, and let $\mathcal{C}$ be a loop at a loop entry vertex $\alpha$ of $\bbt$. We assume in this section, that we have already build a backbone tree $\bbt'$ of an induced program $P'$ of the loop $\mathcal{C}$, and that we have also executed $\bbt'$ symbolically. So we have also computed functions $\Psi'$ and $\Theta'$. We further assume that $\pi_1', \ldots, \pi_n'$, where $n = \eta_\bbt(\alpha)$, are all backbone paths of $\bbt'$ and that $\kappa_{\alpha,1}, \ldots, \kappa_{\alpha,n}$ are all path counters introduced at $\alpha$ for the backbone paths of $\bbt'$ respectively.

Our goal in the section is to describe algorithm computing iterated symbolic state $\theta^{\vec{\kappa}}$. $\theta^{\vec{\kappa}}$ is a symbolic state, where values of program variables are expressed as functions of how many times the backbone paths of $\bbt'$ are iterated. Those numbers of iterations are captured in introduced path counters. Therefore, the resulting iterated symbolic state $\theta^{\vec{\kappa}}$ will be parametrized by the path counters. It means that for any concrete values substituted into the path counters in $\theta^{\vec{\kappa}}$, we obtain a symbolic state over-approximating those received by symbolic execution of the backbone paths of $\bbt'$ as many times as defined by the values of counters.

When $\bbt'$ contains loop entry vertices, then values of some variables may depend on concrete number of iterations along loops at those loop entry vertices. Since these numbers of iterations may be arbitrary in different iterations of backbone paths of $\bbt'$, it is difficult to infer functions of path counters for values of such variables. Of course, we can always express the values as unknown value $\star$. But we would loose a lot of precision. On the other hand, very precise analysis might be computationally expensive. Therefore, we provide an analysis still remaining simple, but precise enough for majority of programs our technique is designed for. We want to be precise in cases, when there is a linear relationship between number of iterations of a loop of $\bbt'$ and values of path counters introduced at $\alpha$. In all other cases we use that unknown value $\star$. Let $\gamma$ be a loop entry vertex of $\bbt'$. Then path counters $\kappa_{\gamma,1}, \ldots, \kappa_{\gamma,m_\gamma}$, where $m_\gamma = \eta_{\bbt'}(\gamma)$, introduced at $\gamma$ identify a number of iterations of backbone paths of an induced program at $\gamma$. An expression $\kappa_{\gamma,1} + \cdots + \kappa_{\gamma,m_\gamma}$ defines a number of iterations backbone paths of the induced program at $\gamma$. Therefore, in our analysis we intend to express values of the expression $\kappa_{\gamma,1} + \cdots + \kappa_{\gamma,m_\gamma}$ as a linear function of path counters $\kappa_{\alpha,1}, \ldots ,\kappa_{\alpha,n}$. Note that we do not try to compute values of individual path counters $\kappa_{\gamma,j}$.

Because all of this, we do not work directly with functions $\Psi'$ and $\Theta'$, but we first compute their updated versions $\bar{\Psi}$ and $\bar{\Theta}$. The update lies basically in replacement of all occurrences of expression $\kappa_{\gamma,1} + \cdots + \kappa_{\gamma,m_\gamma}$ in $\bar{\Psi}$ and $\bar{\Theta}$ by newly introduced basic symbol $\sym{s}_\gamma$. This replacement is done for each loop vertex $\gamma$ of $\bbt'$. We do not have to forget to eliminate all remaining occurrences of path counters $\kappa_{\gamma,1}, \ldots, \kappa_{\gamma,m_\gamma}$ form both $\bar{\Psi}$ and $\bar{\Theta}$. Since we cannot express their values, we replace them by unknown symbol $\star$. We depict the computation of functions $\bar{\Psi}$ and $\bar{\Theta}$ in more details in Algorithm~\ref{alg:introduceArtificials}. There we first set $\bar{\Psi}$ and $\bar{\Theta}$ to be copies of functions $\Psi'$ and $\Theta'$. Then we apply the substitutions for each loop entry vertex $\gamma$ of $\bbt'$. At line~\ref{l:lpCnd} we declare general structure of a looping condition stored in $\bar{\Psi}$ at the loop entry vertex $\gamma$. Its structure is not important now. We discuss a structure of a looping condition later in Section~\ref{sec:loopingCondition}. We replace this formula by one stored at line~\ref{l:newLpCnd}. Disregarding of meaning of these formulae, we can check that validity of the formula at line~\ref{l:lpCnd} implies validity of the formula at line~\ref{l:newLpCnd}. We replace the original formula in $\bar{\Psi}$ by the weaker one at line~\ref{l:newLpCnd} to save some precision: If we applied the substitutions on the original formula, we would receive a formula where antecedents of all implications in it would be of a form $0 \leq \tau_{\gamma,i} < \star$. We can see, that weaker formula at line~\ref{l:newLpCnd} prevent such a substitution and brings therefore more precision after the substitution. At line~\ref{l:loopPerVertices} we enumerate all remaining vertices of $\bbt'$ such that vertex $\gamma$ is their prefix. For each such vertex $\gamma\beta$ we apply the substitutions in function $\bar{\Psi}$ at line~\ref{l:updatePsi} and if $\gamma\beta$ is a leaf vertex of $\bbt$, then we apply the substitutions in function $\bar{\Theta}$ at line~\ref{l:updateMem}. Note that each artificial symbol $s_\gamma$ represents an expressions $\kappa_{\gamma,1} + \cdots + \kappa_{\gamma,m_\gamma}$. Therefore, we later compute those linear relationships between artificial symbols $s_\gamma$ and the path counters $\kappa_{\alpha,1}, \ldots ,\kappa_{\alpha,n}$.

We denote by $\mathcal{V}_s$ a set of all fresh artificial program variables $\var{s}_\gamma$ introduced into $\bar{\Theta}$, and we denote by $\Sigma_s$ a set of all fresh artificial basic symbols $\sym{s}_\gamma$ substituted into functions $\bar{\Psi}$ and $\bar{\Theta}$. Note that $\mathcal{V}_s \cap \mathcal{V}(\mathcal{S}_P) = \emptyset$ and $\Sigma_s \cap \Sigma(\mathcal{S}_P) = \emptyset$.

\begin{algorithm}[!htb]
\newcommand{\introduceArtificials}{$\texttt{introduceArtificials}$}
\caption{\introduceArtificials\texttt{(}$\Psi', \Theta'$\texttt{)} \label{alg:introduceArtificials}}
$\bar{\Psi} \aset \Psi'$\;
$\bar{\Theta} \aset \Theta'$\;
\ForEach{loop entry vertex $\gamma$ of $\bbt'$}{ \label{l:loopPerEntries}
  Let $\bar{\Psi}(\gamma) \equiv \bigwedge_{i=1}^{m_\gamma} (\forall \tau_{\gamma,i}~(0 \leq \tau_{\gamma,i} < \kappa_{\gamma,i} \rightarrow \exists \vec{\tau}_{\gamma,i}~(\vec{0} \leq \vec{\tau}_{\gamma,i} \leq \vec{\kappa}_{\gamma,i}) \wedge \psi_{\gamma,i})$\; \label{l:lpCnd}
  $\bar{\Psi}(\gamma) \aset \forall s~(0 \le s < \sym{s}_{\gamma,i} \rightarrow \bigvee_{i=1}^{m_\gamma} (\psi_{\gamma,i}[(\sum_{k=1}^{m_\gamma}\tau_{\gamma,k})/s] [\tau_{\gamma,1}/\star, \ldots \tau_{\gamma,m_\gamma}/\star]))$\; \label{l:newLpCnd}
  \ForEach{vertex $\gamma \beta$ of $\bbt$}{ \label{l:loopPerVertices}
    $\bar{\Psi}(\gamma\beta) \aset \bar{\Psi}(\gamma\beta) [(\sum_{k=1}^{m_\gamma}\kappa_{\gamma,k})/\sym{s}_\gamma] [\kappa_{\gamma,1}/\star, \ldots \kappa_{\gamma,m_\gamma}/\star]$\; \label{l:updatePsi}
    \If{$\gamma\beta$ is a leaf vertex of $\bbt$}{ \label{l:loopPerEntries2}
	    $\bar{\Theta}(\gamma\beta) \aset \bar{\Theta}(\gamma\beta) [(\sum_{k=1}^{m_\gamma}\kappa_{\gamma,k}) / \sym{s}_\gamma] [\kappa_{\gamma,1}/\star, \ldots \kappa_{\gamma,m_\gamma}/\star]$\; \label{l:updateMem}
    }
  }
}
\Return{$(\bar{\Psi}, \bar{\Theta})$}
\end{algorithm}

We can now move on to computation of the iterated state $\theta^{\vec{\kappa}}$ itself. We define a semi-lattice of all symbolic states, where we compute $\theta^{\vec{\kappa}}$ as a least fix-point of a monotone function defined later. Let us first describe the semi-lattice. Having $\mathcal{S}(\mathcal{V} \cup \mathcal{V}_s)$ we can define an order $\leq = \{ (\star,s)~|~s \in \mathcal{S}(\mathcal{V} \cup \mathcal{V}_s) \}$ on it. Then $\mathcal{L}_0 =(\mathcal{S}(\mathcal{V} \cup \mathcal{V}_s),\leq)$ is a semi-lattice, where symbol $\star$ is the least element. Note that $\mathcal{L}_0$ has finite height $2$. We can define an order $\leq$ on $\mathcal{M}(\mathcal{S}(\mathcal{V} \cup \mathcal{V}_s))$ such that $\leq = \{ (r,s)~|~r,s \in \mathcal{M}(\mathcal{S}(\mathcal{V} \cup \mathcal{V}_s)) \wedge \forall \var{a} \in \mathcal{V} \cup \mathcal{V}_s~r(\var{a}) \leq s(\var{a}) \}$, then $\mathcal{L} = (\mathcal{M}(\mathcal{S}(\mathcal{V} \cup \mathcal{V}_s)), \leq)$ is a map semi-lattice. The least element of $\mathcal{L}$ is a symbolic state $\theta_\star$ and also note that $\mathcal{L}$ is of finite height $|\mathcal{V} \cup \mathcal{V}_s|$, since $\mathcal{V} \cup \mathcal{V}_s$ is finite.

\begin{algorithm}[!htb]
\newcommand{\computeIteratedState}{\texttt{computeIteratedState}}
\caption{\computeIteratedState\texttt{(}$\bar{\Psi}, \bar{\Theta}$\texttt{)} \label{alg:compFixpoint}}
$\theta^{\vec{\kappa}} \aset \theta_\star$\; \label{l:setitBottom}
\Repeat{$\var{change} = \false$}{ \label{l:fixpointLoop}
  $\var{change} \aset \false$\; \label{l:monotoneFnBegin}
  \ForEach{$\var{a} \in \mathcal{V}$}{ \label{l:monotFirstLoop}
    $e \aset  \texttt{iterateVariable}(\var{a}, \bar{\Psi}, \bar{\Theta}, \theta^{\vec{\kappa}}[\var{a} \rightarrow \theta_G(\var{a})])$\; \label{l:iterExpr}
    \If{$\theta^{\vec{\kappa}}(\var{a}) < e$}{
      $\theta^{\vec{\kappa}}(\var{a}) \aset e$\;
      $\var{change} \aset \true$\;
    }
  }
  \ForEach{$\var{s}_\gamma \in \mathcal{V}_s$}{ \label{l:monotSecondLoop}
    $e \aset $ \texttt{iterationsOfLoop}$(\gamma, \bar{\Psi},\theta^{\vec{\kappa}}[\var{s}_\gamma \rightarrow \theta_G(\var{s}_\gamma)])$\;
    \If{$\theta^{\vec{\kappa}}(\var{s}_\gamma) < e$}{
      $\theta^{\vec{\kappa}}(\var{s}_\gamma) \aset e$\;
      $\var{change} \aset \true$\; \label{l:monotoneFnEnd}
    }
  }
}
\Return{$\theta^\kappa$}\;
\end{algorithm}

The symbolic state $\theta^{\vec{\kappa}}$ is an element of the semi-lattice $\mathcal{L}$ and it is computed by Algorithm~\ref{alg:compFixpoint} as a least fix-point of a monotone function depicted at lines \ref{l:monotoneFnBegin}--\ref{l:monotoneFnEnd} in the algorithm. The algorithm computes Kleene's sequence leading to $\theta^{\vec{\kappa}}$ as follows. At line~\ref{l:setitBottom} we set $\theta^{\vec{\kappa}}$ to be the least element $\theta_\star$ of $\mathcal{L}$. Then the loop at line~\ref{l:fixpointLoop} computes the following elements of the Kleene's sequence. Note that this sequence is always finite, since $\mathcal{L}$ is of finite height. The monotone function is computed in two loops. The first loop at line~\ref{l:monotFirstLoop} computes for each program variable \var{a} an iterated value of its values stored in $\Theta$. The iterated value for a variable is a function of path counters $\kappa_{\alpha,1}, \ldots, \kappa_{\alpha,n}$ expressing values of the variable for any number of iterations of backbone paths of $\bbt'$. We discuss the details of this computation in Section~\ref{sec:iterateVar}. If the iterated value $e$ is more precise, then the current value $\theta^{\vec{\kappa}}(\var{a})$, we overwrite it with the iterated one. The second loop at line~\ref{l:monotSecondLoop} computes for each loop entry vertex $\gamma$ of $\bbt'$ a linear function between an artificial basic symbol $\sym{s}_\gamma$, representing an expression $\kappa_{\gamma,1} + \cdots + \kappa_{\gamma,m_\gamma}$, and path counters $\kappa_{\alpha,1}, \ldots, \kappa_{\alpha,n}$. We discuss the details of that computation in Sections~\ref{sec:iterNestedLoop}. Whenever the result $e$ is more precise then the value already stored in $\theta^{\vec{\kappa}}$, then the content of $\theta^{\vec{\kappa}}$ is updated.

%%%%%%%%%%%%%%%%%%%%%%%%%%%%%%%%%%%%%%%%%%%%%%%%%%%%%%%%%%%%%%%%%%%%%%%%
\subsubsection{Computing iterated value of a program variable}
\label{sec:iterateVar}

Algorithm~\ref{alg:iterateVar} computes an iterated value $e$ for a given program variable \var{a}. We start with expression $e$ set to its related a basic symbol at line~\ref{l:justSymbol}. Then in loop at line~\ref{l:loopPerLeaves} we enumerate backbone paths $\pi_1', \ldots, \pi_n'$ of $\bbt'$ in order as they are marked. Remember that $n = \eta_{\bbt}(\alpha)$. Let a backbone path $\pi_i'$ be just enumerated. Then we update $e$ according to a content of \var{a} in $\bar{\Theta}$ for the current path $\pi_i'$. We read the content of \var{a} from $\bar{\Theta}(\pi_i')$ at line~\ref{l:readContentTheta} and store it into $e'$. Note that the result of the read is immediately followed by substituting $\theta^{\vec{\kappa}}$ into it. By the substitution we incorporate already iterated values of other variables into $e'$. Note that the value of \var{a} may depend on other variables. Then at line~\ref{l:branchByType} we proceed differently for variables of a scalar and array types. Nevertheless, both branches are supposed to look up related table to get an iterated value for \var{a} (computed by combining $e$ and $e'$ in the table). This is done at lines~\ref{l:useScalarTab} and~\ref{l:useArrayTab}. It remains to discuss the use of Tables~\ref{tab:iterateVarValue_scalar}~and~\ref{tab:iterateVarValue_array}. We do that in separate paragraphs. 

\begin{algorithm}[!htb]
\newcommand{\iterateVariable}{\texttt{iterateVariable}}
\caption{\iterateVariable\texttt{(}$\var{a}, \bar{\Psi}, \bar{\Theta}, \theta^{\vec{\kappa}}$\texttt{)} } \label{alg:iterateVar}
$e \aset \theta_G(\var{a})$\; \label{l:justSymbol}
\ForEach{$i = 1,2, \ldots, n$}{ \label{l:loopPerLeaves}
  $e' \aset \bar{\Theta}(\pi_i')(\var{a})\theta^{\vec{\kappa}}$\; \label{l:readContentTheta}
  \If{\var{a} is of a scalar type}{ \label{l:branchByType}
    $e \aset$ apply Table~\ref{tab:iterateVarValue_scalar} for values $(e,e')$\; \label{l:useScalarTab}
  } \Else {
    $e \aset$ apply Table~\ref{tab:iterateVarValue_array} for values $(e,e')$\; \label{l:useArrayTab}
  }
}
\Return{$e$}
\end{algorithm}

\paragraph{Iterating values of scalar type} We combine the expressions $e$ and $e'$ of Algorithm~\ref{alg:iterateVar} for a variable \var{a} of a scalar type according to Table~\ref{tab:iterateVarValue_scalar} into a single iterated value. The expression $e$ represents an iterated value of \var{a} of all already enumerated backbone paths $\pi_1', \ldots, \pi_{i-1}'$. And the expression $e'$ represents symbolic value of \var{a} after symbolic execution of the backbone path $\pi_i'$ as the last one. We use Table~\ref{tab:iterateVarValue_scalar} to compute a resulting iterated value as follows. We try to match expressions $e$ and $e'$ to an expression in the first column and first row respectively. In case either $e$ or $e'$ fails to match any of the expressions, then the resulting iterated value is $\star$. Otherwise we pick an expression from the table common to matched column and row.

The expressions in the first row have a structure of all symbolic expressions we are interested in. We want to compute precise iterated values for them. The first expression identifies the case, when \var{a} is not written to along $\pi_i'$ at all. The second expression matches syntactic structure of expressions, whose values follow some arithmetic progression. The arithmetic progression are the most common for variables of programs we are focusing on. For example majority of sequential traversals of arrays typically involve at least one variable whose values follow some arithmetic progression. And the third expression in the first row identify symbolic expressions whose values do not depend on iterations of other backbone paths. Typical examples are variables storing intermediate results, and more importantly flag variables. For example programs typically set or remove flags when scanning an array to check whether the array matches some property or not.

To fully understand the content of the table we need to discuss meaning of symbols appearing in it. First of all we must say, that all the occurrences of the basic symbol $\sym{a}$, all the path counters $\kappa_{\alpha,1}, \ldots, \kappa_{\alpha,n}$, and the expression $\star$ are explicit in the table. $k$ is a natural number such that $k < i$, indices $i_1, \ldots, i_k$ are all natural numbers, they are all distinct, and also less the $i$. They represent indices of some of already enumerated backbone paths $\pi_1', \ldots, \pi_{i-1}'$. Symbols $d_i, d_{i_1}, \ldots, d_{i_k}$ are symbolic expressions of $P$. Any $\rho_j$ is a symbolic expression which may contain at most a path counter $\kappa_{\alpha,j}$ from the path counters $\kappa_{\alpha,1}, \ldots, \kappa_{\alpha,n}$. Expressions $\psi_{j,k}$ are defined as follows.
\begin{equation*}
\psi_{j,k} \equiv
\begin{cases}
\kappa_{\alpha,j} > 0 & k = 1 \\
\kappa_{\alpha,j} > 0 \wedge \exists \vec{\tau}_j (\vec{0} \leq \vec{\tau}_j \leq \vec{\kappa}_j \wedge pc_j \wedge \forall \vec{\tau}_j' (\vec{\tau}_j < \vec{\tau}_j' \leq \vec{\kappa}_j \rightarrow \bigwedge_{\substack{r = 1 \\ r \ne j}}^k \neg pc_{i_r}')) & \text{Otherwise},
\end{cases}
\end{equation*}
where
\begin{align*}
\vec{\tau}_j & = (\tau_1, \ldots, \tau_{j-1}, \tau_{j+1}, \ldots, \tau_k)^T, \\ \vec{\tau}_j' & = (\tau_1', \ldots, \tau_{j-1}', \tau_{j+1}', \ldots, \tau_k')^T, \\
\vec{\kappa}_j & = (\kappa_{\alpha,1}, \ldots, \kappa_{\alpha,j-1}, \kappa_{\alpha,j+1}, \ldots, \kappa_{\alpha,k})^T, \\
pc_j & = pc_{\bbt'}(\pi_j',\bar{\Psi}) \theta^{\vec{\kappa}} [\vec{\kappa}_j / \vec{\tau}_j], \\
pc_j' & = pc_{\bbt'}(\pi_j',\bar{\Psi}) \theta^{\vec{\kappa}} [\vec{\kappa}_j / \vec{\tau}_j'].
\end{align*}
A condition $\psi_{j,k}$ determines whether a backbone path $\pi_j'$ was symbolically executed at least once, and if so, then whether it was executed as the last one of already examined backbone paths where \var{a} is modified. Also note that we substitute $\theta^{\vec{\kappa}}$ into the formula $pc_{\bbt'}(\pi_j',\bar{\Psi})$. The substitution incorporates values of already iterated program variables into the formula.

We also need to clarify a notation used in the expression in the last row and column, where we assume that index $i_{k+1}$ represents the value $i$.

\begin{table}[!htb]
\begin{center}
\begin{tabular}{r|ccc}
&
$\sym{a}$ &
$\sym{a} + d_i$ &
$\rho_i$
\\ \hline

$\sym{a}$ &
$\sym{a}$ &
$\sym{a} + d_i \kappa_{\alpha,i}$ &
$\ite( \psi_{i,1}, \rho_i, \sym{a})$
\\
\\

$\sym{a} + \sum_{j=1}^k d_{i_j} \kappa_{\alpha,i_j}$ &
$\sym{a} + \sum_{j=1}^k d_{i_j} \kappa_{\alpha,i_j}$ &
$\sym{a} + d_i \kappa_{\alpha,i} + \sum_{j=1}^k d_{i_j} \kappa_{\alpha,i_j}$ &
$\star$
\\
\\

\begin{tabular}{l}
$\ite(\psi_{i_1,k}, \rho_{i_1},$ \\
\hspace{0.5cm} \ldots \\
\hspace{0.5cm} $\ite(\psi_{i_k,k}, \rho_{i_k},$ \\
\hspace{1.2cm} $\sym{a}) \ldots )$
\end{tabular} &
\begin{tabular}{l}
$\ite(\psi_{i_1,k}, \rho_{i_1},$ \\
\hspace{0.5cm} \ldots \\
\hspace{0.5cm} $\ite(\psi_{i_k,k}, \rho_{i_k},$ \\
\hspace{1.2cm} $\sym{a}) \ldots )$
\end{tabular} &
$\star$ &
\begin{tabular}{l}
$\ite(\psi_{i_1,k+1}, \rho_{i_1},$ \\
\hspace{0.5cm} \ldots \\
\hspace{0.5cm} $\ite(\psi_{i_{k+1},{k+1}}, \rho_{i_{k+1}},$ \\
\hspace{1.2cm} $\sym{a}) \ldots )$
\end{tabular}
\end{tabular}
\end{center}
\caption{Combining values $e$ and $e'$ of a variable \var{a} of a scalar type.}
\label{tab:iterateVarValue_scalar}
\end{table}

\paragraph{Iterating values of array type} In Table~\ref{tab:iterateVarValue_array} we assume that variable \var{a} is of an array type. We combine the expressions $e$ and $e'$ (computed in Algorithm~\ref{alg:iterateVar}) according to their syntactical structure. We try to match $e$ and $e'$ to an expression in the first column and first row respectively. In case either $e$ or $e'$ fails to match any of the expressions, then the resulting iterated value is $\lambda\vec{\chi}.\star$. Otherwise we pick an expression from the table common to matched column and row. We again use the vector notation. In particular, vector $\vec{\chi}$ represent formal parameters of a value of \var{a} and its dimension therefore matches dimension of the array. The syntax structure of the first expressions in the first row identify the case, when \var{a} is not written to along $l$ at all. The second expression in the first row captures sequence of $n$ writes along backbone path $l$. The outer-most $\ite$ expression represents the last write along $l$, while the most nested one represents the first write. The expressions of the first column have very similar meaning as those in the first row. The only difference is, that expressions in the first column capture arbitrary iteration of all already processed backbone paths (and not only a single the current one). The most complicated expressions in the table lies in the last column. The expression in the second row computes iteration of all writes along the path $l$. The iterated expression has similar structure as the one at the first row. The only difference is that expressions $d_i, s_i$ are transformed into iterated versions $h_i, t_i$. The expression in the last row and column combines iterations of all writes along all already iterated path including the current path $l$. Since it does not mater on the order of writes from different backbone paths, we append iterated versions of writes along $l$ as the most nested $\ite$ expressions in the result, i.e.~$\ite(h_{m+1}, t_{m+1} \ldots \ite(h_{m+n}, t_{m+n}, \arrsym{a}) \ldots )$.

\begin{table}[!htb]
\begin{center}
\begin{tabular}{r|cc}
&
$\arrsym{a}$ &
\begin{tabular}{l}
$\lambda \vec{\chi}.\ite(d_1, s_1,$ \\ \hspace{1.2cm} $\ldots$ \\ \hspace{1.2cm} $\ite(d_n, s_n,$ \\ \hspace{1.9cm} $\arrsym{a}) \ldots )$
\end{tabular}
\\ \hline

$\arrsym{a}$ &
$\arrsym{a}$ &
\begin{tabular}{l}
$\lambda \vec{\chi}.\ite(h_1, t_1,$ \\ \hspace{1.2cm} $\ldots$ \\ \hspace{1.2cm} $\ite(h_{n}, t_{n},$ \\ \hspace{1.9cm} $\arrsym{a}) \ldots )$
\end{tabular}
\\

\begin{tabular}{l}
$\lambda \vec{\chi}.\ite(c_1, r_1,$ \\ \hspace{1.2cm} $\ldots$ \\ \hspace{1.2cm} $\ite(c_m, r_m,$ \\ \hspace{1.9cm} $\arrsym{a}) \ldots )$
\end{tabular} &

\begin{tabular}{l}
$\lambda \vec{\chi}.\ite(c_1, r_1,$ \\ \hspace{1.2cm} $\ldots$ \\ \hspace{1.2cm} $\ite(c_m, r_m,$ \\ \hspace{1.9cm} $\arrsym{a}) \ldots )$
\end{tabular} &

\begin{tabular}{l}
$\lambda \vec{\chi}.\ite(h_1, t_1,$ \\ \hspace{1.2cm} $\ldots$ \\ \hspace{1.2cm} $\ite(h_{m+n}, t_{m+n},$ \\ \hspace{1.9cm} $\arrsym{a}) \ldots )$
\end{tabular}

\end{tabular}
\end{center}
\caption{Combining values $e$ and $e'$ using $\varphi$ of a variable \var{a} of an array type.}
\label{tab:iterateVarValue_array}
\end{table}

To fully understand meaning of the table, we also need to discuss structure of expressions $c_i, d_i, h_i$ and $r_i, s_i, t_i$ appearing inside $\ite$ expressions in the table. None of these expressions has an implicit occurrence of the basic symbol $\sym{a}$, a path counter $\kappa \in \mathcal{K}$ or parameter $\tau \in \mathcal{T}$. And none of them is equal to $\star$. All occurrences of mentioned symbols in the expressions are always stated explicitly in their description.

We start with the description of expressions $c_i$. Each $c_i$ declaratively identifies all those indices into the array, where $i$-th nested $\ite$ expression in $e$ writes during all iteration of all already examined backbone paths.	The indices can be expressed as follows
\begin{equation*}
c_i = \exists \vec{\tau}~.~\vec{\chi} = \vec{u}_i \wedge \vec{0} \leq \vec{\tau} < \vec{p}_i \wedge \phi_i(U,P) \wedge \gamma_i.
\end{equation*}
Vector $\vec{\chi}$ serves only as a placeholder, where actual parameters are substituted, when we read from the array. Vector $\vec{u}_i$ may contain parameters from $\vec{\tau}$ and identifies possible indices, where the $i$-th $\ite$ expression may write its value $r_i$ during all the iteration of already examined backbone paths. Therefore, if for some concrete vector $\vec{I}$ of indices into the array there exists a $\vec{\tau}$ such that the formula $(\vec{\chi} = \vec{u}_i)[\vec{\chi}/\vec{I}]$ is $\true$, then we know, that $\vec{I}$ identifies element of the array overwritten by the $i$-th $\ite$ expression during the iteration. But values of parameters $\vec{\tau}$ must be real -- they capture only iterations of examined backbone paths where number of their iterations do not exceed values of counters. Therefore, parameters $\vec{\tau}$ are restricted from top by a vector of expressions $\vec{p}_i$. Expressions in $\vec{p}_i$ possibly (and typically) contain some of the counters $\kappa_1, \ldots, \kappa_p$. Formula $\phi_i(U,P)$ checks whether $r_i$ is the value, which was written to the array last at an element identified by $\vec{\chi}$. We discuss structure of $\phi_i(U,P)$ later. Only note that $U$ is a sequence of all vectors $\vec{u}_i$ and $P$ is a sequence of all vectors $\vec{p}_i$. Each formula $\gamma_i$ uniquely identifies a single path $l$ in $\bbt$ from $l_s$ down to location of $i$-th write. We use formula $pc(l)$ to express the condition for $l$.

Expressions $d_i$ have similar structure as expressions $c_i$. But they capture writes only along current path $l$. They have the following structure.
\begin{equation*}
d_i =
\begin{cases}
\exists \vec{\tau}~.~\vec{\chi} = \vec{v}_i \wedge \vec{0} \leq \vec{\tau} < \vec{q}_i \wedge \phi_i(V,Q) \wedge \gamma_i & \text{if $\vec{v}_i$ contains at least one parameter} \\
\vec{\chi} = \vec{v}_i & \text{otherwise}
\end{cases}
\end{equation*}
The first case matches the situation, when path $l$ contains at least one component vertex. Analysing related SCCs recursively we receive imported counters. Therefore, value of the array was iterated by just discussed procedure and it implies the more complicated structure. The second case identifies common symbolic write into the array along $l$. Only note that $V$ is a sequence of all vectors $\vec{v}_i$ and $Q$ is a sequence of all vectors $\vec{q}_i$. We discuss structure of $\phi_i(V,Q)$ later.

Expressions $h_i$ express iterated versions of expressions $d_i$. Since we iteratively combine backbone paths of $\bbt$ into resulting $\theta^{\vec{\kappa}}$ the structure of $e$ actually represent iterated version of value of the array (but only for already examined backbone paths). Therefore, expressions $c_i$ already represent iterated versions of writes along examined path. Since we want to extend iterated value of \var{a}  in $\theta^{\vec{\kappa}}$ by writes along current $l$ it is obvious that structure of expressions $h_i$ is the same as for the expressions $c_i$:
\begin{align*}
h_i = \exists \vec{\tau}~.~\vec{\chi} = \vec{w}_i & \wedge \vec{0} \leq \vec{\tau} < \vec{g}_i \wedge \phi_i(W,G) \wedge \ite(i \le m, \gamma_i, \psi[\vec{\kappa}/\vec{\tau}]), \\
\vec{w}_i & = \ite(i \le m, \vec{u}_i, \vec{v}_{i-m}[\vec{\kappa}/\vec{\tau}])\\
\vec{g}_i & = \ite( i \le m, \vec{p}_i, \ite(\mathcal{T}(\vec{v}_{i-m}) \ne \emptyset, (\vec{g}^T_{i - m}, \mathcal{K}(\vec{v}_{i-m}) \smallsetminus \mathcal{K}(\vec{g}_{i-m}))^T, (\mathcal{K}(\vec{g}_{i-m})^T)).
\end{align*}
Note that vectors $\vec{w}_i, \vec{g}_i$ are defined to choose right expression either from $e$ or $e'$. Also note, that formula $\gamma$ is extended by formula $\ite(i \le m, \neg\psi, \psi)$, where $\psi$ was computed in Algorithm~\ref{alg:iterateVar}. It distinguishes writes along the current path $l$ from writes along other already examined backbone paths. We discuss structure of $\phi_i(W,G)$ bellow. Only note that $W$ is a sequence of all vectors $\vec{w}_i$ and $G$ is a sequence of all vectors $\vec{G}_i$.

Now we can discuss structure of formulae $\phi_i(Z,B)$. The sequence $Z = \{\vec{z}_1, \ldots, \vec{z}_k\}$ contains all those indices to the array, where the array is written to during the iteration of all already examined backbone paths. Note that each such index is a vector of symbolic expressions of dimension $m$, if $m$ is a dimension of the array \var{a}. The sequence $B = \{\vec{b}_1, \ldots, \vec{b}_k\}$ containts vectors restricting values of parameters $\vec{\tau}$ appearing in related indices in $Z$. The formula $\phi_i(Z,B)$ has the following structure:
\begin{align*}
\phi_i(\{\vec{z}_1, \ldots, \vec{z}_k\}, \{\vec{b}_1, \ldots, \vec{b}_k\}) & = \forall \vec{\tau}_1', \ldots, \vec{\tau}_k'~.~(\bigwedge_{j = 1}^k \vec{\tau}_j < \vec{\tau}_j' < \vec{b}_j) \rightarrow (\bigwedge_{j = 1}^k \zeta(\vec{z}_j,i,j)[\vec{\tau}_j/\vec{\tau}_j']),\\
\zeta(\vec{z},i,j) & = 
\begin{cases}
\vec{\chi} \ne \vec{z} & \text{if $i \ne j$ or some $\tau$ appears in $\vec{z}$} \\
\true & \text{otherwise}
\end{cases}
\end{align*}
We see, that the formula is not sentence. It contains free variables -- parameters $\tau_i \in \mathcal{T}$ -- which are exposed in the formula through vectors $\vec{\tau}_j$. Note that two different $\vec{\tau}_j, \vec{\tau}_k, j \ne k$ may share some parameters. But all these free variables (parameters) can be stored in a single vector $\vec{\tau}$, which is exactly the one existentially quantified in expressions $c_i, d_i$ and $h_i$. The formula $\phi_i(Z,B)$ states for given parameters $\vec{\tau}$ that each write to the array in any future iterations of already examined backbone paths will store its value to the different element of the array then to the one indexed by $\vec{z}_i$. In other words, the formula says that a value lastly overwritten in the array at index $\vec{z}_i$ was done by $i$-th $\ite$ expression in iteration identified by parameters $\vec{\tau}$.

Each $s_i$ denotes a symbolic expression written to the array along current backbone path $l$. We do not restrict their syntactic structure in any way.

Expressions $r_i$ and $t_i$ have similar structure, since each $r_i$ represent iterated version of an expression written to the array, and each $t_i$ has the same meaning, but it also includes iterated versions of expressions written along the current backbone path $l$, i.e.~iterated versions of expressions $c_i$. We therefore discuss only structure of $t_i$:
\begin{equation*}
t_i = \ite(i \le m, r_i, \delta_i).
\end{equation*}
We see, that for all $i \le m$ we have $t_i = r_i$. Since expressions $r_i$ already are iterated, we do not need to do any action for them. For all the remaining expressions $s_{i-m}$ (i.e.~$m + 1 \le i \le m + n$) we need to compute their iterated versions, before we put them into $t_i$. We express the iterated version of $s_{i-m}$ by the expression $\delta_i$:
\begin{equation*}
\delta_i =
\begin{cases}
(f \tau_l + \sym{a}(\vec{z}))\{\vec{w}_i/\vec{\chi}\} & \text{if $s_{i-m}$ is of a form $\sym{a}(\vec{z}) + f$} \\
\ite(h_i \wedge \phi_i'(W), \rho_i(s_{i-m},\vec{v}_{i-m})[\kappa_l/\tau_l]\{\vec{w}_i/\vec{\chi}\}, \star) & \text{if $s_{i-m}$ is of a form $\sym{a}(\vec{z}\kappa_l+\vec{y}) + f$} \\
\star & \text{if $s_{i-m}$ is any other expression containing $\sym{a}$} \\
s_{i-m}[\kappa_l/\tau_l]\{\vec{w}_i/\vec{\chi}\} & \text{otherwise}
\end{cases}
\end{equation*}
The first case identifies a situation, when a single element of the array indexed by $\vec{w}_i$ is updated several times during iteration of $\bbt$ such that the values in the element follow an arithmetic progression. The second case identifies a situation, when sequences of elements of the array follow some arithmetic progressions. We discuss details of this case bellow. Whenever $s_{i-m}$ contains $\sym{a}$, but it is not of the syntactic form of neither the first nor the second case, then we return $\star$. The forth case matches any symbolic expressions without $\sym{a}$ inside them.

In the second case of $\delta_i$ each written element of the array is a part of a linear function. The single write can produce several lines, and each written element of the array belongs to exactly one of the lines. The iterated version of $s_i$ is thus expression declaratively describing all the lines. But it is not only about describing the lines. We must also ensure, that other writes to the array (along any backbone paths) do not corrupt them during the iteration. That is the reason for the $\ite$ expression for this case. The condition of the $\ite$ expression checks, whether lines are not corrupted during the iteration of backbone paths. We have already discussed structure of $h_i$. Therefore, it remains to describe structure of boolean expression $\phi_i'(W)$. Remember, that $W$ is a sequence of all vectors $\vec{w}_i$. Structure of $\phi_i'(W)$ is very similar to $\phi_i(W,G)$, since they have the same purpose -- to detect accidental writes to selected array elements.
\begin{equation*}
\phi_i'(\{\vec{z}_1, \ldots, \vec{z}_k\}) = \forall \vec{\tau}_1', \ldots, \vec{\tau}_k'~.~(\bigwedge_{j = 1}^k 0 \leq \vec{\tau}_j' < \vec{\tau}_j) \rightarrow (\bigwedge_{j = 1}^k \zeta(\vec{z}_j,i,j)[\vec{\tau}_j/\vec{\tau}_j'])
\end{equation*}
Since formula $\phi_i(W,G)$ detects accidental writes in future iterations, the formula $\phi_i'(W)$ can only check overwrites in the previous iterations (see the antecedent of the implication). Note that there are free variables in $\phi_i'(W)$ (exactly those which are free in $\phi_i(W,G)$), which are existentially bind through $h_i$ (i.e.~scope of $\exists \vec{\tau}$ in $h_i$ covers also $\phi_i'(W)$).

Expression $\rho_i(s_{i-m},\vec{v}_{i-m})$ identifies the lines. Note that $s_{i-m}$ is of a form $\sym{a}(\vec{z}\kappa_l+\vec{y}) + f$. When $\vec{p} = \vec{v}_{i-m}[\kappa_l/(\kappa_l + 1)] - \vec{v}_{i-m}$ is a vector identifying differences in indices between subsequent iterations, and $\vec{q} = \vec{v}_{i-m} - \vec{z}\kappa_l - \vec{y}$ is a vector identifying differences in indices between l-value and r-value, then we require that $\mathbf{Diag}(p) q  \ge \vec{0}$ and at least one element of the vector $\mathbf{Diag}(p) q$ must be strictly greater then $0$. If one of these requirements is not met, then we evaluate $\rho(s_{i-m},\vec{v}_{i-m})$ to $\star$. Otherwise we define $\rho_i(s_{i-m},\vec{v}_{i-m})$ as an expression
\begin{center}
\begin{tabular}{l}
$\ite((z_1\kappa_l + y_1) \mod |q_1| \equiv \hat{q}_{1,0} \wedge \cdots \wedge (z_k\kappa_l + y_k) \mod |q_k| \equiv \hat{q}_{k,0},$ \\ \qquad\qquad $f\kappa + \sym{a}(z_1\hat{q}_{1,0}+y_1, \ldots, z_k\hat{q}_{k,0}+y_k),$ \\ \qquad $\ldots $ \\ \qquad  $\ite((z_1\kappa_l + y_1) \mod |q_1| \equiv \hat{q}_{1,N} \wedge \cdots \wedge (z_k\kappa_l + y_k) \mod |q_k| \equiv \hat{q}_{k,N},$ \\ \qquad\qquad\qquad  $f\kappa + \sym{a}(z_1\hat{q}_{1,N}+y_1, \ldots, z_k\hat{q}_{k,N}+y_k),$ \\ \qquad\qquad\qquad  $\star) \ldots )$,
\end{tabular}
\end{center}
where $k$ is dimension of the array, $\hat{q}_{i,j} = \min\{ |q_i|, j \}$ and $N = \max\{ |q_1|, \ldots, |q_k| \} - 1$. Note that presence of $\star$ in the expression is only technical -- to simplify listing of the formula. The formula can never be evaluated to that $\star$.

As an example, consider a program expression $\texttt{a[i]} \aset \texttt{a[i-5] + 1}$. Then $\rho_i(\sym{a}(\kappa + \sym{i} - 5) + 1,\kappa + \sym{i})$, where iterated value of \var{i} is $\kappa + \sym{i}$, has five composed $\ite$ expressions. There are generated five independent lines in array \var{a} during the iteration. Array elements of these lines are interleaved modulo 5 in \var{a}. And formula $\rho_i$ captures this property.

%%%%%%%%%%%%%%%%%%%%%%%%%%%%%%%%%%%%%%%%%%%%%%%%%%%%%%%%%%%%%%%%%%%%%%%%
\subsubsection{Computing number of iterations of nested Loop}
\label{sec:iterNestedLoop}

Here we compute a linear function between a symbol $\sym{s}_\gamma$, representing the expression $\kappa_{\gamma,1} + \cdots + \kappa_{\gamma,m_\gamma}$, and the path counters $\kappa_{\alpha,1}, \ldots ,\kappa_{\alpha,n}$. Of course, if there is no such linear function, we fail to infer the function. It may also be the case, that there is a linear relationship, but coefficients of the function do not form liner functions over input symbols. It that case, we also fails to compute the function. The main idea behind our algorithm computing the linear function can be explained as follows.

We start with precise formulation of a condition identifying, whether value of $\sym{s}_\gamma$ is linearly dependent on path counters $\kappa_{\alpha,1}, \ldots ,\kappa_{\alpha,n}$ or not. Formula $\bar{\Psi}(\gamma)$ is weakened looping condition of $\mathcal{C}$. It ensures, that each iteration along the loop $\mathcal{C}$ gets back to the entry vertex $\lst(\gamma)$, until it is a time to leave it. Leaving the loop means to follow some path in $\mathcal{C}$ from the loop entry vertex to one of its exit vertices. Formulae in $\bar{\Psi}$ along all these paths identify the the leaving condition after that successful iteration in $\mathcal{C}$. So we need all these formulae to describe the iterations of $\mathcal{C}$. But these formulae describe the iterations only for single (you can imagine the last) iteration of $\bbt'$. To capture arbitrary (previous) iteration of $\bbt'$ we need to substitute $\theta^{\vec{\kappa}}$ into these formulae. Therefore, the discussed condition identifying iterations of $\mathcal{C}$ can be formally expressed as
\begin{equation*}
\Gamma(\gamma,\zeta,\theta) \equiv \bigvee_{\beta \in \{ \beta~\mid~\gamma\beta \in \bbt~\wedge~\lst(\beta) \in \exits(\mathcal{C}) \}} \left( \bigwedge_{\alpha \in \pref(\beta)} \zeta(\gamma\alpha)\theta\right).
\end{equation*}
It only remains to state, that whenever we have a proper iteration of $\mathcal{C}$, identified by $\Gamma(\gamma,\bar{\Psi},\theta^{\vec{\kappa}})$, then number of its iterations $\sym{s}_\gamma$ is linearly dependent on path counters $\kappa_{\alpha,1}, \ldots ,\kappa_{\alpha,n}$. Let us first discuss a case, when there is no occurrence of a basic symbol of an array type in $\Gamma(\gamma,\bar{\Psi},\theta^{\vec{\kappa}})$. We describe how to deal with arrays at the end of the section.

Let $\vec{\sym{a}}$ be a vector of all basic symbols of scalar types appearing in $\Gamma(\gamma,\bar{\Psi},\theta^{\vec{\kappa}})$. We want to state, that for each concrete input (i.e.~for each assignment of concrete values to symbols in $\vec{\sym{a}}$), there is a vector $\vec{p}$ of integers and some integer $q$, such that $\sym{s}_\gamma = \vec{p}^T \vec{\kappa} + q$, for each possible choice of concrete values for $\sym{s}_\gamma$ and path counters $\vec{\kappa}$ appearing in $\Gamma(\gamma,\bar{\Psi},\theta^{\vec{\kappa}})$. We can formally write the linear relationship as
\begin{equation*}
\forall \vec{\sym{a}} \exists \vec{p},q \forall \vec{\kappa}, \sym{s}_\gamma~\left( (\vec{\kappa} \ge \vec{0} \wedge \sym{s}_\gamma \ge 0 \wedge \Gamma(\gamma,\bar{\Psi},\theta^{\vec{\kappa}})) \rightarrow \sym{s}_\gamma = \max\{ 0, \vec{p}^T \vec{\kappa} + q \}\right).
\end{equation*}
Presence of function $\max$ in the formula solves cases, when linear relation would imply negative value for $\sym{s}_\gamma$. But $\sym{s}_\gamma$ is a natural number and negative value for $\sym{s}_\gamma$ only implies that $\mathcal{C}$ is not iterated at all. Therefore, in such situations we provide the alternative choice for $\sym{s}_\gamma$, to be equal to $0$.

In the presented formula the values $\vec{p}, q$ may vary for each choice of concrete input values $\vec{\sym{a}}$. Although an SMT solver may give us an answer that given formula is valid, we can only conclude that there indeed is a linear relationship between number of iterations of $\mathcal{C}$ and values of the path counters. But we do not know the relationship itself. To force an SMT solver to compute the linear relationship for us we do the following. We restrict ourselves only to those linear functions, where its coefficients are some fixed linear combinations of input values. In other words, we only focus on those relationships, where all variations of $\vec{p}, q$ for different inputs $\vec{\sym{a}}$ can be captured by a single (fixed) linear combination of input values $\vec{\sym{a}}$. This restriction allow us to move existential quantification to the front of the formula. And we get
\begin{equation*}
\exists M, \vec{w}, \forall \vec{\sym{a}},\vec{\kappa}, \sym{s}_\gamma~\left( (\vec{\kappa} \ge \vec{0} \wedge \sym{s}_\gamma \ge 0 \wedge \Gamma(\gamma,\bar{\Psi},\theta^{\vec{\kappa}})) \rightarrow \sym{s}_\gamma = \max\{ 0, (M \vec{\kappa} + \vec{w})^T \matr{\vec{\sym{a}} \\ 1} \} \right),
\end{equation*}
where $M$ and $\vec{w}$ are matrix and vector of unknown integers to be computed by an SMT solver respectively. If the formula is satisfiable, then we can get the integers as a part of model of the formula from an SMT solver. These integers define the linear combinations of input we wanted. Although a type of $M$ and a dimension of $\vec{w}$ might be clear from the formula, we rather discuss it. If number of basic symbols in the formula (i.e.~dimension of $\vec{a}$) is $m$ and number of path counters in the formula (i.e.~dimension of $\vec{\kappa}$) is $n$, type of matrix $M$ is $(m + 1) \times n$, and dimension of $\vec{w}$ is $m + 1$.

The last formula would be the result, if modern SMT solvers had performed well on it. We have experimented with powerful SMT solver Z3. But the performance was poor. We found very simple instances of the formula, where it took several minutes \todo{navic vysledek z SMT solveru byl spatny. Mozna bychom to meli poslat vyvojarum Z3, aby se na to podivali} for the SMT solver to check satisfiability for each of them. We discovered, that performance issue lies in nested general quantifiers brought to the formula through looping conditions. \todo{Mozna by bylo dobre se vyvojaru Z3 zeptat, proc ten vnoreny kvantifikator dela takove peklo. Ty formule nejsou nijak slozite} Fortunately, we do not need to express all iterations of $\mathcal{C}$ in each iteration of $\bbt$. It is sufficient for the relationship to ensure, that we stay in $\mathcal{C}$ in $(\sym{s}_\gamma - 1)$-st iteration of $\mathcal{C}$ in each iteration of $\bbt$. (Leaving of $\mathcal{C}$ in $\sym{s}_\gamma$-th iteration is then ensured by formulae collected from $\bar{\Psi}$ along paths to exit vertices). Therefore, if $\bar{\Psi}(\gamma)$ is a looping condition of a form $\forall s~(0 \le s < \sym{s}_\gamma \rightarrow \psi)$, then we can replace it by a condition $0 \le \sym{s}_\gamma - 1 \rightarrow \psi[s/(\sym{s}_\gamma - 1)]$. \todo{mozna by to chtelo zduvodnit trochu lepe...} Z3 SMT solver is able to decide satisfiable such updated formulae in tens of miliseconds. Which is significant performance improvement. To integrate the modification into our last formlula, we formally introduce a formula $\Omega$ defined on vertices of $\bbt$ as follows
\begin{equation*}
\Omega(\gamma) \equiv 
\begin{cases}
0 \le \sym{s}_\gamma - 1 \rightarrow \psi[s/(\sym{s}_\gamma - 1)] & \text{if $\gamma$ is a component vertex, where} \\ & \qquad \text{$\bar{\Psi}(\gamma) \equiv \forall s~(0 \le s < \sym{s}_\gamma \rightarrow \psi)$)} \\
\bar{\Psi}(\gamma) & \text{Otherwise}
\end{cases}
\end{equation*}
Using $\Omega$ we can finally define a formula $S_\gamma$
\begin{equation*}
S_\gamma \equiv \exists M, \vec{w}, \forall \vec{\sym{a}},\vec{\kappa}, \sym{s}_\gamma~\left( (\vec{\kappa} \ge \vec{0} \wedge \sym{s}_\gamma \ge 0 \wedge \Gamma(\gamma,\Omega,\theta^{\vec{\kappa}})) \rightarrow \sym{s}_\gamma = \max\{ 0, (M \vec{\kappa} + \vec{w})^T \matr{\vec{\sym{a}} \\ 1} \} \right),
\end{equation*}
whose satisfiability we check to compute the relationship.

The last thing to be discussed an occurrence of arrays in $\Gamma(\gamma,\bar{\Psi},\theta^{\vec{\kappa}})$. Although we try to express $\sym{s}_\gamma$ as a linear function, whose coefficients are some fixed linear combinations of input values of only scalar types, presence of array symbols in $\Gamma(\gamma,\bar{\Psi},\theta^{\vec{\kappa}})$ may strongly affect existence of such a relationship. We must ensure, that a relation exists not only for all possible values to symbols of scalar types, but also for all possible contents of arrays. Unfortunately, we cannot quantify a function symbol in first order language. Therefore, we solve the problem in two steps. First, we introduce a fresh function symbol $\rho$. This function accepts as arguments $\vec{a}$, i.e.~whole input to symbols of scalar types. The function returns for each input $\vec{a}$ a unique integer number. It means that $\rho$ is injective. Formally speaking we add the following axiom into extended theory $T_P$.
\begin{equation*}
\forall \vec{a}_1, \vec{a}_2~\rho(\vec{a}_1) = \rho(\vec{a}_2) \leftrightarrow \vec{a}_1 = \vec{a}_2.
\end{equation*}
The second step we need to do is to replace each function symbol application $\sym{A}(e_1, \ldots, e_k)$ occurring in $\Gamma(\gamma,\bar{\Psi},\theta^{\vec{\kappa}})$ by an application $\sym{A}(e_1, \ldots, e_k, \rho(\vec{a}))$. It means that basic symbols of array types have changed their type such that their dimension have been increased by one. This way we ensure, that for each assignment to $\vec{a}$ we have a fresh contents of all arrays for checking satisfiability of our formula $S_\gamma$.

We are ready to describe in Algorithm~\ref{alg:iterationsOfComponent} the computation of the expression identifying number of iterations of $\mathcal{C}$ as a linear function of path counters of $\bbt$. We assume, that the axiom for function $\rho$ is automatically inserted into extended theory $T_P$ of $P$. \todo{doplnit slovni popis pro Algorithm~\ref{alg:iterationsOfComponent}}

\begin{algorithm}[!htb]
\newcommand{\iterationsOfComponent}{$\texttt{iterationsOfComponent}$}
\caption{\iterationsOfComponent\texttt{(}$\gamma, \theta^{\vec{\kappa}}$\texttt{)} \label{alg:iterationsOfComponent}}
$S_\gamma \aset \exists M, \vec{w}, \forall \vec{\sym{a}},\vec{\kappa}, \sym{s}_\gamma~\left( (\vec{\kappa} \ge \vec{0} \wedge \sym{s}_\gamma \ge 0 \wedge \Gamma(\gamma,\Omega,\theta^{\vec{\kappa}})) \rightarrow \sym{s}_\gamma = \max\{ 0, (M \vec{\kappa} + \vec{w})^T \matr{\vec{\sym{a}} \\ 1} \} \right)$\;
Extend all function symbols applications in $S_\gamma$ by an extra parameter $\rho(\vec{a})$\;
\If{$S_\gamma$ is satisfiable (ask an SMT solver)}{
   retrieve $M, \vec{w}$ from a model of $S_\gamma$ computed by the SMT solver\;
   \Return{$\max\{ 0, (M \vec{\kappa} + \vec{w})^T \matr{\vec{\sym{a}} \\ 1} \}$}
}
\Return{$\star$}
\end{algorithm}

%%%%%%%%%%%%%%%%%%%%%%%%%%%%%%%%%%%%%%%%%%%%%%%%%%%%%%%%%%%%%%%%%%%%%%%%
%%%%%%%%%%%%%%%%%%%%%%%%%%%%%%%%%%%%%%%%%%%%%%%%%%%%%%%%%%%%%%%%%%%%%%%%
\subsection{Computation of looping condition $\varphi^{\vec{\kappa}}$}
\label{sec:loopingCondition}

Let $\alpha$ be a loop entry vertex of $\bbt$ of a program $P$. Then we can build a backbone tree $\bbt'$ of an induced program of the loop at the loop entry vertex. Let $\pi_1', \ldots, \pi_n'$, where $n = \eta_\bbt(\alpha)$, be all the backbone paths of $\bbt'$. After symbolic execution of $\bbt'$ we receive filled in function $\Psi'$ and according to Section~\ref{sec:IterateTheta}, we can then also compute an iterated symbolic state $\theta^{\vec{\kappa}}$ of $\bbt'$. Now we are ready to express a looping condition $\varphi^{\vec{\kappa}}$ of the loop over-approximating all path conditions representing feasible paths iterating tree $\bbt'$. The formula $\varphi^{\vec{\kappa}}$ is defined as follows
\begin{align*}
\varphi^{\vec{\kappa}} \equiv \bigwedge_{i = 1}^n (\forall \tau_i~&(0 \leq \tau_i < \kappa_{\alpha,i} \rightarrow \exists \vec{\tau}_i~(\vec{0} \leq \vec{\tau}_i \leq \vec{\kappa}_i \wedge pc_{\bbt'}(\pi_i',\Psi') \theta^{\vec{\kappa}} [\kappa_{\alpha,1}/\tau_1, \ldots, \kappa_{\alpha,n}/\tau_n] )), \quad \text{where} \\ &
\vec{\tau}_i = (\tau_1, \ldots, \tau_{i-1}, \tau_{i+1}, \ldots, \tau_n)^T, \\ &
\vec{\kappa}_i = (\kappa_{\alpha,1}, \ldots, \kappa_{\alpha,i-1}, \kappa_{\alpha,i+1}, \ldots, \kappa_{\alpha,n})^T.
\end{align*}
The formula can be explained as follows. Feasible path iterating in the induced program of the loop give us concrete values of the paths counters $\kappa_{\alpha,1}, \ldots, \kappa_{\alpha,n}$. For each path counter $\kappa_{\alpha,i}$ (i.e.~all its concrete values) the looping condition must ensure, that the backbone path $\pi_i'$ is executed at least $\kappa_{\alpha,i}$ times. Therefore, for each execution number $\tau_i$ between $0$ and $\kappa_{\alpha,i} - 1$ there must exist actual execution numbers $\tau_1, \ldots, \tau_{i-1}, \tau_{i+1}, \ldots, \tau_n$ of remaining backbone paths lying in their limits (i.e.~$0 \leq \tau_j \leq \kappa_{\alpha,j}$) such that execution of $\pi_i'$ is possible, i.e.~path condition of $\pi_i'$ is satisfiable. This must be ensured for execution numbers of all backbone paths. Note that we substitute $\theta^{\vec{\kappa}}$ into path condition $pc_{\bbt'}(\pi_i',\Psi')$. This is necessary, because values in the path condition capture only single execution along the path. Substitution converts those values into functions of path counters, so they represent any possible number of iterations of backbone paths. Also note that we do not have to ensure in the looping condition that path $\pi_i'$ is also executed at most $\kappa_{\alpha,i}$ times. This property is handled by backbone paths in $\bbt$, since they contains paths from loop entry vertices to all possible loop exits. Assertions along these paths do the job.

\begin{lem}
Only free variables in $\varphi^{\vec{\kappa}}$ are the path counters $\kappa_{\alpha,1}, \ldots, \kappa_{\alpha,n}$.
\begin{proof}
Obvious.
\end{proof}
\end{lem}

\begin{lem}
Let $\Psi(\alpha)$ be updated to a formula $\varphi^{\vec{\kappa}}$ computed as described above. Then for any path condition $\varphi$ representing a feasible path iterating $\bbt'$ the sentence $\varphi \rightarrow \hat{pc}_\bbt(\alpha,\Psi,\true)$ is valid.
\begin{proof}
It directly follows from the construction of $\hat{pc}_\bbt(\alpha,\Psi,\true)$.
\end{proof}
\end{lem}

%%%%%%%%%%%%%%%%%%%%%%%%%%%%%%%%%%%%%%%%%%%%%%%%%%%%%%%%%%%%%%%%%%%%%%%%
%%%%%%%%%%%%%%%%%%%%%%%%%%%%%%%%%%%%%%%%%%%%%%%%%%%%%%%%%%%%%%%%%%%%%%%%
\subsection{Discussing Relaxations}

We finish computation of loop over-approximation by a brief discussion of the relaxations we use in the computation of $\hat{\varphi}$.
There are many loops in real-world programs where interleaving of paths through the loops is not important for reasoning about conditions below them. For example many C++ programs manipulate sequential containers by calling Standard Template Library functions like \texttt{copy}, \texttt{find}, \texttt{find\_if}, \texttt{transform}, \texttt{for\_each}, \texttt{count}, \texttt{count\_if}. Loops in these functions commonly have the property. And it is also very common that iterations of loops are controlled by values following monotone progressions. Consider for example concept of iterators in C++ Standard Template Library. Branchings below such loops are mostly dependent on a final state of these progressions. Therefore, using the relaxations we can compute $\hat{\varphi}$ such that it is well balanced between complexity and precision.

%%%%%%%%%%%%%%%%%%%%%%%%%%%%%%%%%%%%%%%%%%%%%%%%%%%%%%%%%%%%%%%%%%%%%%%%
%%%%%%%%%%%%%%%%%%%%%%%%%%%%%%%%%%%%%%%%%%%%%%%%%%%%%%%%%%%%%%%%%%%%%%%%
%%%%%%%%%%%%%%%%%%%%%%%%%%%%%%%%%%%%%%%%%%%%%%%%%%%%%%%%%%%%%%%%%%%%%%%%
\section{Soundness and Incompleteness}\label{sec:soundness}

In this section we formulate and prove soundness and incompleteness theorems for our algorithm.

\begin{thm}[Soundness]
Let $\hat{\varphi}$ be the necessary condition computed by our algorithm for a given target program location. If $\hat{\varphi}$ is not satisfiable, then the target location is not reachable in that program.
\end{thm}
\begin{proof}[Informal proof]
We build any looping condition $\varphi^{\vec{\kappa}}$ such that it is implied by all path conditions of an analysed loop. And each formula $pc_\bbt(\pi_{i}, \Psi)$ collects all the predicated along backbone $\pi_i$ and it also collects all looping conditions at loop entries along the path. Therefore, $pc_\bbt(\pi_{i}, \Psi)$ must be implied by any path condition of any symbolic execution along $\pi_i$. We compute $\hat{\varphi}$ as a disjunction of formulae $\varphi_i$ for all backbones. Since any program path leading to the target location must follow some backbone (with possible temporary escapes into loops along the backbone), its path condition exists (i.e.~it is satisfiable formula) only if $\hat{\varphi}$ is satisfiable.
%\label{alg:execBB:addLoopCond}
%$\hat{\varphi}$ is 
\end{proof}

\begin{thm}[Incompleteness]
There is a program and an unreachable target location in it for which the formula $\hat{\varphi}$ computed by our algorithm is satisfiable.
\end{thm}
\begin{proof}
Let us consider the following C code:
\begin{verbatim}
   int i = 1; while (i < 3) { if (i == 2) i = 1; else i = 2; }
\end{verbatim}
The loop never terminates. Therefore, a program location below it is not reachable. But $\hat{\varphi}$ computed for that location is equal to $\true$, since variable $\var{i}$ does not follow a monotone progression.
\end{proof}
 
%%%%%%%%%%%%%%%%%%%%%%%%%%%%%%%%%%%%%%%%%%%%%%%%%%%%%%%%%%%%%%%%%%%%%%%%
%%%%%%%%%%%%%%%%%%%%%%%%%%%%%%%%%%%%%%%%%%%%%%%%%%%%%%%%%%%%%%%%%%%%%%%%
%%%%%%%%%%%%%%%%%%%%%%%%%%%%%%%%%%%%%%%%%%%%%%%%%%%%%%%%%%%%%%%%%%%%%%%%
\section{Dealing with Quantifiers}\label{sec:quantifiers}

We can ask an SMT solver whether a computed necessary condition $\hat{\varphi}$ is satisfiable or not. And if it is, we may further ask for some its model. As we will see in Section~\ref{sec:apps} such queries to a solver should be fast. Unfortunately, our experience with solvers shows that presence of quantifiers in $\hat{\varphi}$ usually causes performance issues. Although SMT technology evolves quickly, we show in this section how to overcome this issue now by unfolding universally quantified formulae the looping conditions $\varphi^{\vec{\kappa}}$ are made of.

Universally quantified variables $\tau_i$ in $\varphi^{\vec{\kappa}}$ are always restricted from above by path counters $\kappa_i$ counting iterations of backbones $\pi_i$ of analysed loop. Let us choose some upper limits $K_i > 0$ for the path counters $\kappa_i$. Since each $\tau_i$ ranges over a finite set of integers $\{ 0, \ldots, K_i - 1 \}$ now, we can unfold each universally quantified formula in $\varphi^{\vec{\kappa}}$ for each possible value of $\tau_i$. Having eliminated the universal quantification, we can also eliminate existential quantification of all $\kappa_i$ and all $\vec{\tau_i}$ by redefining them as uninterpreted integer constants. For given upper limits $\vec{K}$ for the path counters $\vec{\kappa}$ we denote an unfolded necessary condition $\hat{\varphi}$ by $\hat{\varphi}^{\vec{K}}$.

For any $\vec{K}$ the formula $\hat{\varphi}^{\vec{K}}$ represents wakened $\hat{\varphi}$. Higher values we choose in $\vec{K}$, then we get closer to the precision of $\hat{\varphi}$. In practice we must choose moderate values $\vec{K}$, since the unfolding process makes $\hat{\varphi}^{\vec{K}}$ much longer then $\hat{\varphi}$.

In some cases an SMT solver is able to quickly decide satisfiability of $\hat{\varphi}$. Therefore, we ask the solver for satisfiability of $\hat{\varphi}$ in parallel with the unfolding procedure described above. And there is a common timeout for both queries. We take the fastest answer. In case both queries exceeds the timeout, the condition $\hat{\varphi}$ cannot help a tool to cover given target location.

%%%%%%%%%%%%%%%%%%%%%%%%%%%%%%%%%%%%%%%%%%%%%%%%%%%%%%%%%%%%%%%%%%%%%%%%
%%%%%%%%%%%%%%%%%%%%%%%%%%%%%%%%%%%%%%%%%%%%%%%%%%%%%%%%%%%%%%%%%%%%%%%%
%%%%%%%%%%%%%%%%%%%%%%%%%%%%%%%%%%%%%%%%%%%%%%%%%%%%%%%%%%%%%%%%%%%%%%%%
\section{Integration into Tools}\label{sec:apps}

Tools typically explore program paths iteratively. At each iteration there is a set of program locations $\{ v_1, \ldots, v_k \}$, from which the symbolic execution may continue further. At the beginning the set contains only program entry location. In each iteration of the symbolic execution the set is updated such that actions of program edges going out from \emph{some} locations $v_i$ are symbolically executed. Different tools use different systematic and heuristic strategies for selecting locations $v_i$ to be processed in the current iteration. It is also important to note that for each $v_i$ there is available an actual path condition $\varphi_i$ capturing already taken symbolic execution from the entry location up to $v_i$.

When a tool detects difficulties in some iteration to cover a particular program location, then using $\hat{\varphi}$ it can restrict selection from the whole set $\{ v_1, \ldots, v_k \}$ to only those locations $v_i$, for which a formula $\varphi_i \wedge \hat{\varphi}$ is satisfiable. In other words, if for some $v_i$ the formula $\varphi_i \wedge \hat{\varphi}$ is \emph{not} satisfiable, then we are guaranteed there is no real path from $v_i$ to the target location. And therefore, $v_i$ can safely be removed from the consideration.

Tools like \Sage, \Pex or \Cute combine symbolic execution with concrete one. Let us assume that a location $v_i$, for which the formula $\varphi_i \wedge \hat{\varphi}$ is satisfiable, was selected in a current iteration. These tools require a concrete input to the program to proceed further from $v_i$. Such an input can directly be extracted from any model of the formula $\varphi_i \wedge \hat{\varphi}$.

%%%%%%%%%%%%%%%%%%%%%%%%%%%%%%%%%%%%%%%%%%%%%%%%%%%%%%%%%%%%%%%%%%%%%%%%
%%%%%%%%%%%%%%%%%%%%%%%%%%%%%%%%%%%%%%%%%%%%%%%%%%%%%%%%%%%%%%%%%%%%%%%%
%%%%%%%%%%%%%%%%%%%%%%%%%%%%%%%%%%%%%%%%%%%%%%%%%%%%%%%%%%%%%%%%%%%%%%%%
\section{Experimental Results}\label{sec:experiments}

We implemented the algorithm in an experimental program, which we call \APC. We also prepared a small set of benchmark programs mostly taken from other papers. In each benchmark we marked a single location as the target one. All the benchmarks have a huge number of paths, so it is difficult to reach the target. We run \Pex and \APC on the benchmarks and we measured times till the target locations were reached. This measurement is obviously unfair from \Pex perspective, since its task is to cover an analysed benchmark by tests and not to reach a single particular location in it. Therefore, we clarify the right meaning of the measurement now.

Our only goal here is to show, that \Pex could benefit from our algorithm. Typical scenario when running \Pex on a benchmark is that all the code except the target location is covered in few seconds (typically up to three). Then \Pex keeps searching space of program paths for a longer time without covering the target location. This is exactly the situation when our heuristic should be activated. We of course do not know the exact moment, when \Pex would activate it. Therefore, we can only provide running times of our heuristic as it was activated at the beginning of the analysis.

Before we present the results, we discuss the benchmarks. Benchmark HWM checks whether an input string contains four substrings \texttt{Hello}, \texttt{world}, \texttt{at} and \texttt{Microsoft!}. It does not matter at which position and in which order the words occur in the string. The target location can be reached only when all the words are presented in the string. This benchmark was introduced in~\cite{AGT08}. The benchmark consists of four loops in a sequence, where each loop searches for a single of the four words mentioned above. Each loop checks for an occurrence of a related word at each position in the input string starting from the beginning. Benchmark HWM is the most complicated one from our set of benchmarks. We also took its two lightened versions presented in~\cite{OT11}: Benchmark HW consists of two loops searching the input string for the first two words above. And benchmark Hello searches only for the first one.

Benchmark MatrIR scans upper triangle of an input matrix. The matrix can be of any rank bigger then $20 \times 20$. In each row we count a number of elements inside a fixed range $(10,100)$. When sum of counts from all the rows exceeds a fixed limit $15$, then the target location is reached.

Benchmarks OneLoop and TwoLoops originate from~\cite{OT11}. They are designed such that their target locations are not reachable. Both benchmarks contain a loop in which the variable \var{i} (initially set to 0) is increased by 4 in each iteration. The target location is then guarded by an assertion \texttt{i==15} in OneLoop benchmark and by a loop\texttt{~while (i != j + 7) j += 2~}in the second one. We note that \var{j} is initialized to $0$ before the loop.

The last benchmark WinDriver comes from a practice and we took it from~\cite{GLE09}. It is a part of a Windows driver processing a stream of network packets. It reads an input stream and decomposes it into a two dimensional array of packets. A position in the array where the data from the stream are copied into are encoded in the input stream itself. We marked the target location as a failure branch of a consistency check of the filled in array. It was discussed in the paper~\cite{GLE09} the consistency check can indeed be broken.

\begin{table}[!htb]
	\begin{center}
    \begin{tabular}{c||c|c|c|c|c}
      & \Pex & \multicolumn{4}{c}{\APC} \\
      \cline{2-6}
      \textbf{Benchmark} & \textbf{Total} & \textbf{Total} & Bld $\hat{\varphi}$ & Unf/SMT $\hat{\varphi}^{\vec{K}}$ & SMT $\hat{\varphi}$ \\
      \hline
      Hello       & 5.257 & 0.181 & 0.021 & 0.290 / S 0.060 & S 0.160 \\
      HW          & 25.05 & 0.941 & 0.073 & 0.698 / S 0.170 & S 13.84 \\
      HWM         & T/O   & 4.660 & 1.715 & 2.135 / S 0.810 & X M/O \\
      MatrIR      & 95.00 & 0.035 & 0.015 & 0.491 / S 70.80 & S 0.020 \\
      WinDriver   & 28.39 & 0.627 & 0.178 & 0.369 / S 0.080 & X 4.860 \\
      \hline
      OneLoop     & 134.0 & 0.003 & 0.001 & 0.001 / U 0.001 & U 0.010 \\
      TwoLoops    & 64.00 & 0.003 & 0.002 & 0.004 / U 0.010 & U 0.001
    \end{tabular}
	\end{center}
  \caption{Running times of \Pex and \APC on benchmarks.}
  \label{tab:experiments}
\end{table}

The experimental results are depicted in Table~\ref{tab:experiments}. They show running times in seconds of \Pex and \APC on the benchmarks. We did all the measurements on a single common desktop computer\footnote{Intel\textsuperscript{\textregistered} Core$^{\mathtt{TM}}$ i7 CPU 920 @ 2.67GHz 2.67GHz, 6GB RAM, Windows 7 Professional 64-bit, MS \Pex 0.92.50603.1, MS Moles 1.0.0.0, MS Visual Studio 2008, MS .NET Framework v3.5 SP1, MS Z3 SMT solver v3.2, and boost v1.42.0.}. The mark T/O in \Pex column indicates that it failed to reach the target location within an hour. For \APC we provide the total running times and also time profiles of different paths of the computation. In sub-column 'Bld $\hat{\varphi}$' there are times required to build the necessary condition $\hat{\varphi}$. In sub-column 'Unf/SMT $\hat{\varphi}^{\vec{K}}$' there are two times for each benchmark. The first number identifies a time spent by unfolding the formula $\hat{\varphi}$ into $\hat{\varphi}^{\vec{K}}$. We use a fixed number 25 for all the counters and benchmarks. The second number represent a time spent by Z3 SMT solver~\cite{Z3} to decide satisfiability of the unfolded formula $\hat{\varphi}^{\vec{K}}$. Characters in front of these times identify results of the queries: S for satisfiable, U for unsatisfiable and X for unknown. And the last sub-column 'SMT $\hat{\varphi}$' contains running times of Z3 SMT solver directly on formulae $\hat{\varphi}$. The mark M/O means that Z3 went out of memory. As we explained in Section~\ref{sec:quantifiers} the construction and satisfiability checking of $\hat{\varphi}^{\vec{K}}$ runs in parallel with satisfiability checking of $\hat{\varphi}$. Therefore, we take the minimum of the times to compute the total runing time of \APC.

%%%%%%%%%%%%%%%%%%%%%%%%%%%%%%%%%%%%%%%%%%%%%%%%%%%%%%%%%%%%%%%%%%%%%%%%
%%%%%%%%%%%%%%%%%%%%%%%%%%%%%%%%%%%%%%%%%%%%%%%%%%%%%%%%%%%%%%%%%%%%%%%%
%%%%%%%%%%%%%%%%%%%%%%%%%%%%%%%%%%%%%%%%%%%%%%%%%%%%%%%%%%%%%%%%%%%%%%%%
\section{Related Work}\label{sec:related}

Early work on symbolic execution~\cite{Kin76,BEL75,How77} showed its effectiveness in test generation.
King further showed that symbolic execution can bring more automation into Floyd's inductive proving method~\cite{Kin76,Floyd67}. Nevertheless, loops as the source of the path explosion problem were not in the center of interest.

More recent approaches dealt mostly with limitations of SMT solvers and the environment problem by combining the symbolic execution with the concrete one~\cite{PKS05,AGT08,SMA05,GLM08:active_props,G07,GLM08:fuzzing,GKL08,TdH08,GLM08:fuzzing,PRV11}. Although practical usability of the symbolic execution improved, these approaches still suffer from the path explosion problem.
An interesting idea
is to combine the symbolic execution with a complementary technique~\cite{GNRT10,GMR09,Beckmanetal08,NRTT09,Gulavanietal06}. Complementary techniques typically perform differently on different parts of the analysed program. Therefore, an information exchange between the techniques leads to a mutual improvement of their performance.
There are also techniques based on saving of already observed program behaviour and early terminating those executions, whose further progress will not explore a new one~\cite{BCE08,Cadar08,CDE08}.
Compositional approaches are typically based on computation of function summaries~\cite{G07,AGT08}. A function summary often consists of pre and post condition. Preconditions identify paths through the function and postconditions capture effects of the function along those paths. Reusing these summaries at call sites typically leads to an interesting performance improvement. In addition the summaries may insert additional symbolic values into the path condition which causes another improvement.% of performance.
And there are also techniques partitioning program paths into separate classes according to similarities in program states~\cite{QNR11,SH10}. Values of output variables of a program or function are typically considered as a partitioning criteria.
A search strategy Fitnex~\cite{XTHS09} implemented in \Pex~\cite{TdH08} uses state-dependent fitness values computed through a fitness function to guide a path exploration. The function measures how close an already discovered feasible path is to a particular target location (to be covered by a test). The fitness function computes the fitness value for each occurrence of a predicate related to a chosen program branching along the path. The minimum value is the resulting one. There are also orthogonal approaches dealing with the path explosion problem by introducing some assumptions about program input. There are, for example, specialized techniques for programs manipulating strings~\cite{BTV09,XGM08}, and techniques reducing input space by a given grammar~\cite{GKL08,SPmCS09}.

Although the techniques above showed performance improvements when dealing with the path explosion problem, they do not focus directly on loops.
The LESE~\cite{SPmCS09} approach introduces symbolic variables for the number of times each loop was executed and links these with features of a known input grammar such as variable-length or repeating fields. This allows the symbolic constraints to cover a class of paths that includes different number of loop iterations, expressing loop-dependent program values in terms of the input.
A technique presented in~\cite{GL11} analyses loops on-the-fly, i.e.~during simultaneous concrete and symbolic execution of a program for a concrete input. The loop analysis infers inductive variables. A variable is inductive if it is modified by a constant value in each loop iteration. These variables are used to build loop summaries expressed in a form of pre a post conditions. The summaries are derived from the partial loop invariants synthesized dynamically using pattern matching rules on the loop guards and induction variables.
In our previous work~\cite{OT11} we introduced an algorithm sharing the same goal as one presented here. Nevertheless, in~\cite{OT11} we transform an analysed program into chains and we do the remaining analysis there. For each chain with sub-chains we build a constraint system serving as an oracle for steering the symbolic execution in the path space towards the target location.

%%%%%%%%%%%%%%%%%%%%%%%%%%%%%%%%%%%%%%%%%%%%%%%%%%%%%%%%%%%%%%%%%%%%%%%%
%%%%%%%%%%%%%%%%%%%%%%%%%%%%%%%%%%%%%%%%%%%%%%%%%%%%%%%%%%%%%%%%%%%%%%%%
%%%%%%%%%%%%%%%%%%%%%%%%%%%%%%%%%%%%%%%%%%%%%%%%%%%%%%%%%%%%%%%%%%%%%%%%
\section{Conclusion}\label{sec:conclusion}

We presented algorithm computing for a given target program location the necessary condition $\hat{\varphi}$ representing an over-approximated set of real program paths leading to the target. We proposed the use of $\hat{\varphi}$ in tests generation tools based on symbolic execution. Having $\hat{\varphi}$ such a tool can cover the target location faster by exploring only program paths in the over-approximated set. We also showed that $\hat{\varphi}$ can be used in the tools very easily and naturally. And we finally showed by the experimental results that \Pex could benefit from our algorithm.

%%%%%%%%%%%%%%%%%%%%%%%%%%%%%%%%%%%%%%%%%%%%%%%%%%%%%%%%%%%%%%%%%%%%%%%%
%%%%%%%%%%%%%%%%%%%%%%%%%%%%%%%%%%%%%%%%%%%%%%%%%%%%%%%%%%%%%%%%%%%%%%%%
%%%%%%%%%%%%%%%%%%%%%%%%%%%%%%%%%%%%%%%%%%%%%%%%%%%%%%%%%%%%%%%%%%%%%%%%
\bibliographystyle{plain}
\bibliography{apc}

\begin{thebibliography}{10}

\bibitem{AGT08}
S.~Anand, P.~Godefroid, and N.~Tillmann.
\newblock Demand-driven compositional symbolic execution.
\newblock In {\em TACAS'08}, volume 4963 of {\em LNCS}, pages 367--381.
  Springer, 2008.

\bibitem{Beckmanetal08}
N.~E. Beckman, A.~V. Nori, S.~K. Rajamani, R.~J. Simmons, S.~Tetali, and A.~V.
  Thakur.
\newblock Proofs from tests.
\newblock In {\em ISSTA '08}, pages 3--14. ACM, 2008.

\bibitem{BTV09}
N.~Bj{\o}rner, N.~Tillmann, and A.~Voronkov.
\newblock Path feasibility analysis for string-manipulating programs.
\newblock In {\em TACAS '09}, pages 307--321. Springer-Verlag, 2009.

\bibitem{BCE08}
P.~Boonstoppel, C.~Cadar, and D.~Engler.
\newblock {RWset}: attacking path explosion in constraint-based test
  generation.
\newblock In {\em TACAS'08/ETAPS'08}, pages 351--366. Springer-Verlag, 2008.

\bibitem{BEL75}
R.~S. Boyer, B.~Elspas, and K.~N. Levitt.
\newblock {SELECT}---a formal system for testing and debugging programs by
  symbolic execution.
\newblock In {\em Proceedings of the international conference on Reliable
  software}, pages 234--245. ACM, 1975.

\bibitem{CDE08}
C.~Cadar, D.~Dunbar, and D.~Engler.
\newblock {KLEE}: Unassisted and automatic generation of high-coverage tests
  for complex systems programs.
\newblock In {\em OSDI'08}, pages 209--224. USENIX Association, 2008.

\bibitem{Cadar08}
C.~Cadar, V.~Ganesh, P.~M. Pawlowski, D.~L. Dill, and D.~R. Engler.
\newblock {EXE}: Automatically generating inputs of death.
\newblock {\em ACM Trans. Inf. Syst. Secur.}, 12(2):1--38, 2008.

\bibitem{Floyd67}
R.~W. Floyd.
\newblock Assigning meanings to programs.
\newblock In {\em Proceedings of a Symposium on Applied Mathematics},
  volume~19, pages 19--31, 1967.

\bibitem{G07}
P.~Godefroid.
\newblock Compositional dynamic test generation.
\newblock In {\em POPL '07}, pages 47--54. ACM, 2007.

\bibitem{GKL08}
P.~Godefroid, A.~Kiezun, and M.~Y. Levin.
\newblock Grammar-based whitebox fuzzing.
\newblock In {\em PLDI '08}, pages 206--215. ACM, 2008.

\bibitem{PKS05}
P.~Godefroid, N.~Klarlund, and K.~Sen.
\newblock {DART}: directed automated random testing.
\newblock In {\em PLDI '05}, pages 213--223. ACM, 2005.

\bibitem{GLM08:active_props}
P.~Godefroid, M.~Y. Levin, and D.~A. Molnar.
\newblock Active property checking.
\newblock In {\em EMSOFT '08}, pages 207--216. ACM, 2008.

\bibitem{GLM08:fuzzing}
P.~Godefroid, M.~Y. Levin, and D.~A. Molnar.
\newblock Automated whitebox fuzz testing.
\newblock In {\em Network Distributed Security Symposium (NDSS)}, pages
  151--166, 2008.

\bibitem{GLE09}
P.~Godefroid, X.~Levin, and X.~Elkarablieh.
\newblock Precise pointer reasoning for dynamic test generation.
\newblock In {\em ISSTA'09}, 2009.

\bibitem{GL11}
P.~Godefroid and D.~Luchaup.
\newblock Automatic partial loop summarization in dynamic test generation.
\newblock In {\em ISSTA '11}, pages 23--33. ACM, 2011.

\bibitem{GNRT10}
P.~Godefroid, A.~V. Nori, S.~K. Rajamani, and S.~D. Tetali.
\newblock Compositional may-must program analysis: unleashing the power of
  alternation.
\newblock In {\em POPL '10}, pages 43--56. ACM, 2010.

\bibitem{Gulavanietal06}
B.~S. Gulavani, T.~A. Henzinger, Y.~Kannan, A.~V. Nori, and S.~K. Rajamani.
\newblock {SYNERGY}: a new algorithm for property checking.
\newblock In {\em SIGSOFT '06/FSE-14}, pages 117--127. ACM, 2006.

\bibitem{GMR09}
A.~Gupta, R.~Majumdar, and A.~Rybalchenko.
\newblock From tests to proofs.
\newblock In {\em TACAS '09}, pages 262--276. Springer-Verlag, 2009.

\bibitem{How77}
William~E. Howden.
\newblock Symbolic testing and the {DISSECT} symbolic evaluation system.
\newblock {\em IEEE Trans. Software Eng.}, 3(4):266--278, 1977.

\bibitem{Kin76}
J.~C. King.
\newblock Symbolic execution and program testing.
\newblock {\em Commun. ACM}, 19(7):385--394, 1976.

\bibitem{NRTT09}
A.~V. Nori, S.~K. Rajamani, S.~Tetali, and A.~V. Thakur.
\newblock The {Yogi} project: Software property checking via static analysis
  and testing.
\newblock In {\em TACAS '09}, pages 178--181. Springer-Verlag, 2009.

\bibitem{OT11}
J.~Obdr\v{z}\'{a}lek and M.~Trt\i{i}k.
\newblock Efficient loop navigation for symbolic execution.
\newblock In {\em ATVA '11}, pages 453--462. LNCS, 2011.

\bibitem{PRV11}
C.~S. P\u{a}s\u{a}reanu, N.~Rungta, and W.~Visser.
\newblock Symbolic execution with mixed concrete-symbolic solving.
\newblock In {\em ISSTA '11}, pages 34--44. ACM, 2011.

\bibitem{QNR11}
D.~Qi, H.~D.~T. Nguyen, and A.~Roychoudhury.
\newblock Path exploration based on symbolic output.
\newblock In {\em ESEC/FSE '11}, pages 278--288. ACM, 2011.

\bibitem{SH10}
R.~Santelices and M.~J. Harrold.
\newblock Exploiting program dependencies for scalable multipe-path symbolic
  execution.
\newblock In {\em ISSTA '10}, pages 195--206. ACM, 2010.

\bibitem{SPmCS09}
P.~Saxena, P.~Poosankam, S.~McCamant, and D.~Song.
\newblock Loop-extended symbolic execution on binary programs.
\newblock In {\em ISSTA '09}, pages 225--236. ACM, 2009.

\bibitem{SMA05}
K.~Sen, D.~Marinov, and G.~Agha.
\newblock {CUTE}: a concolic unit testing engine for {C}.
\newblock In {\em ESEC/FSE-13}, pages 263--272. ACM, 2005.

\bibitem{TdH08}
N.~Tillmann and J.~de~Halleux.
\newblock Pex -- white box test generation for .{NET}.
\newblock In {\em TAP'08}, volume 4966 of {\em LNCS}, pages 134--153. Springer,
  2008.

\bibitem{XTHS09}
X.~Tao Xie, N.~Tillmann, J.~de~Halleux, and W.~Schulte.
\newblock Fitness-guided path exploration in dynamic symbolic execution.
\newblock In {\em DSN '09}, pages 359--368, 2009.

\bibitem{XGM08}
R.~G. Xu, P.~Godefroid, and R.~Majumdar.
\newblock Testing for buffer overflows with length abstraction.
\newblock In {\em ISSTA '08}, pages 27--38. ACM, 2008.

\bibitem{Z3}
\url{http://research.microsoft.com/projects/Z3}.

\end{thebibliography}
\clearpage

%%%%%%%%%%%%%%%%%%%%%%%%%%%%%%%%%%%%%%%%%%%%%%%%%%%%%%%%%%%%%%%%%%%%%%%%
%%%%%%%%%%%%%%%%%%%%%%%%%%%%%%%%%%%%%%%%%%%%%%%%%%%%%%%%%%%%%%%%%%%%%%%%
%%%%%%%%%%%%%%%%%%%%%%%%%%%%%%%%%%%%%%%%%%%%%%%%%%%%%%%%%%%%%%%%%%%%%%%%
\appendix

%%%%%%%%%%%%%%%%%%%%%%%%%%%%%%%%%%%%%%%%%%%%%%%%%%%%%%%%%%%%%%%%%%%%%%%%
%%%%%%%%%%%%%%%%%%%%%%%%%%%%%%%%%%%%%%%%%%%%%%%%%%%%%%%%%%%%%%%%%%%%%%%%
%%%%%%%%%%%%%%%%%%%%%%%%%%%%%%%%%%%%%%%%%%%%%%%%%%%%%%%%%%%%%%%%%%%%%%%%
\section{Examples}

%%%%%%%%%%%%%%%%%%%%%%%%%%%%%%%%%%%%%%%%%%%%%%%%%%%%%%%%%%%%%%%%%%%%%%%%
%%%%%%%%%%%%%%%%%%%%%%%%%%%%%%%%%%%%%%%%%%%%%%%%%%%%%%%%%%%%%%%%%%%%%%%%
\subsection{Iterating a variable $\var{A}$ of an array type}
\label{sec:iterArrayExamples}

\begin{exm}
We want to compute an iterated value of an array variable \var{A} for the following loop
\begin{verbatim}
for (int i = 0; i < n; ++i)
  A[i] = i;
\end{verbatim}
After symbolic execution of loop's body variable \var{A} has a value $\lambda \chi~.~\ite(\chi = \sym{i}, \sym{i}, \sym{A}(\chi))$. And suppose that variable \var{i} was already iterated, i.e.~$\theta^{\vec{\kappa}}(\var{i}) = \kappa + \sym{i}$. Then expressions $e$ and $e'$ from Algorithm~\ref{alg:iterateVar} are assigned as follows
\begin{align*}
e & \equiv \lambda \chi~.~\sym{A}(\chi) \\
e' & \equiv \lambda \chi~.~\ite(\chi = \kappa + \sym{i}, \kappa + \sym{i}, \sym{A}(\chi)).
\end{align*}
Since there is only single backbone path $l$ in the backbone tree of the program, $pc(l) \equiv \true$. According to Table~\ref{tab:iterateVarValue_array} we receive the following values:
\begin{align*}
w_1 & \equiv (\kappa + \sym{i})[\kappa/\tau] = \tau + \sym{i} \\
g_1 & \equiv \kappa \\
\zeta(\{ \tau + \sym{i} \},1,1) & \equiv \chi \ne \tau + \sym{i} \\
\phi_1(\{ \tau + \sym{i} \},\{ \kappa \}) & \equiv \forall \tau' (\tau < \tau' < \kappa \rightarrow \chi \ne \tau' + \sym{i}) \\
\gamma_1 & \equiv \true \\
\psi & \equiv \true \\
h_1 & \equiv \exists \tau~.~\chi = \tau + \sym{i} \wedge 0 \leq \tau < \kappa \wedge \forall \tau' (\tau < \tau' < \kappa \rightarrow \chi \ne \tau' + \sym{i}) \\
t_1 & \equiv (\kappa + \sym{i})[\kappa/\tau]\{w_1/\chi\} = (\tau + \sym{i})\{(\tau + \sym{i})/\chi\} = \chi \\
\end{align*}
Therefore the resulting iterated value for array \var{A} is
\begin{equation*}
\lambda \chi~.~\ite(\exists \tau~.~\chi = \tau + \sym{i} \wedge 0 \leq \tau < \kappa \wedge \forall \tau' (\tau < \tau' < \kappa \rightarrow \chi \ne \tau' + \sym{i}), \chi, \sym{A}(\chi))
\end{equation*}
Note that in special cases like this one, we can simply detect, that $\tau + \sym{i}$ is a monotone function. Since there are no other writes to the array, the condition $\phi_1$ is redundant and the expression can be simplified into
\begin{equation*}
\lambda \chi~.~\ite(\exists \tau~.~\chi = \tau + \sym{i} \wedge 0 \leq \tau < \kappa, \chi, \sym{A}(\chi))
\end{equation*}
\end{exm}

\begin{exm}
Let us consider the following C++ program:
\begin{verbatim}
for (int i = 1; i < n; ++i) {
  A[i-1] = i;
  A[i] = i;
}
\end{verbatim}
After symbolic execution of loop's body variable \var{A} has a value $\lambda \chi~.~\ite(\chi = \sym{i}, \sym{i}, \ite(\chi = \sym{i} - 1, \sym{i},\sym{A}(\chi)))$. And suppose that variable \var{i} was already iterated, i.e.~$\theta^{\vec{\kappa}}(\var{i}) = \kappa + \sym{i}$. Then expressions $e$ and $e'$ from Algorithm~\ref{alg:iterateVar} are assigned as follows
\begin{align*}
e & \equiv \lambda \chi~.~\sym{A}(\chi) \\
e' & \equiv \lambda \chi~.~\ite(\chi = \kappa + \sym{i}, \kappa + \sym{i}, \ite(\chi = \kappa + \sym{i} - 1, \kappa + \sym{i},\sym{A}(\chi))).
\end{align*}
Since there is only single backbone path $l$ in the backbone tree of the program, $pc(l) \equiv \true$. According to Table~\ref{tab:iterateVarValue_array} we receive the following values:
\begin{align*}
w_1 & \equiv (\kappa + \sym{i})[\kappa/\tau] = \tau + \sym{i} \\
w_2 & \equiv (\kappa + \sym{i} - 1)[\kappa/\tau] = \tau + \sym{i} - 1 \\
g_1 & \equiv \kappa \\
g_2 & \equiv \kappa \\
\zeta(\{ \tau + \sym{i} \},1,1) & \equiv \chi \ne \tau + \sym{i} \\
\zeta(\{ \tau + \sym{i} - 1 \},1,2) & \equiv \chi \ne \tau + \sym{i} - 1 \\
\phi_1(\{ \tau + \sym{i}, \tau + \sym{i} - 1 \},\{ \kappa, \kappa \}) & \equiv \forall \tau' (\tau < \tau' < \kappa \rightarrow (\chi \ne \tau' + \sym{i} \wedge \chi \ne \tau' + \sym{i} - 1)) \\
\zeta(\{ \tau + \sym{i} \},2,1) & \equiv \chi \ne \tau + \sym{i} \\
\zeta(\{ \tau + \sym{i} - 1 \},2,2) & \equiv \chi \ne \tau + \sym{i} - 1 \\
\phi_2(\{ \tau + \sym{i}, \tau + \sym{i} - 1 \},\{ \kappa, \kappa \}) & \equiv \forall \tau' (\tau < \tau' < \kappa \rightarrow (\chi \ne \tau' + \sym{i} \wedge \chi \ne \tau' + \sym{i} - 1)) \\
\gamma_1 & \equiv \true \\
\psi & \equiv \true \\
h_1 & \equiv \exists \tau~.~\chi = \tau + \sym{i} \wedge 0 \leq \tau < \kappa~\wedge \\ & \hspace{1.3cm} \forall \tau' (\tau < \tau' < \kappa \rightarrow (\chi \ne \tau' + \sym{i} \wedge \chi \ne \tau' + \sym{i} - 1)) \\
h_2 & \equiv \exists \tau~.~\chi = \tau + \sym{i} - 1 \wedge 0 \leq \tau < \kappa~\wedge \\ & \hspace{1.3cm} \forall \tau' (\tau < \tau' < \kappa \rightarrow (\chi \ne \tau' + \sym{i} \wedge \chi \ne \tau' + \sym{i} - 1)) \\
t_1 & \equiv (\kappa + \sym{i})[\kappa/\tau]\{w_1/\chi\} = (\tau + \sym{i})\{(\tau + \sym{i})/\chi\} = \chi \\
t_2 & \equiv (\kappa + \sym{i})[\kappa/\tau]\{w_2/\chi\} = (\tau + \sym{i})\{(\tau + \sym{i} - 1)/\chi\} = \chi + 1 \\
\end{align*}
Therefore the resulting iterated value for array \var{A} is
\begin{align*}
\lambda \chi~.~\ite(& \exists \tau~.~\chi = \tau + \sym{i} \wedge 0 \leq \tau < \kappa \wedge \forall \tau' (\tau < \tau' < \kappa \rightarrow (\chi \ne \tau' + \sym{i} \wedge \chi \ne \tau' + \sym{i} - 1)), \chi, \\
& \ite(\exists \tau~.~\chi = \tau + \sym{i} - 1 \wedge 0 \leq \tau < \kappa \wedge \forall \tau' (\tau < \tau' < \kappa \rightarrow (\chi \ne \tau' + \sym{i} \wedge \chi \ne \tau' + \sym{i} - 1)), \chi + 1, \\ & \hspace{0.7cm} \sym{A}(\chi)))
\end{align*}
The condition of the outer $\ite$ expression can be satisfied only for $\tau = \kappa - 1$. On the other hand, the condition of the nested $\ite$ expression is satisfiable for all values of $\tau$. Also note that $\tau + \sym{i} - 1$ is a monotone function. Therefore, we can simplify the iterated value into
\begin{equation*}
\lambda \chi~.~\ite(\chi = \kappa - 1 + \sym{i}, \chi, \ite(\exists \tau~.~\chi = \tau + \sym{i} - 1 \wedge 0 \leq \tau < \kappa, \chi + 1, \sym{A}(\chi)))
\end{equation*}
\end{exm}

\begin{exm}
Let us consider the following C++ program:
\begin{verbatim}
for (int i = 0; i < n; ++i)
  if (i % 2 == 0)  // i.e. is 'i' even? 
    A[i] = 2*i + 1;
  else
    A[i] = 5;
\end{verbatim}
There are two backbone paths along the loop. The first one $l_1$ goes through positive branch and the second path $l_2$ through the negative one. After symbolic execution of the loop's body variable \var{A} has the following values:
\begin{align*}
\Theta(l_1)(\var{A}) & = \lambda \chi~.~\ite(\chi = \sym{i}, 2\sym{i} + 1, \sym{A}(\chi)) \\
\Theta(l_2)(\var{A}) & = \lambda \chi~.~\ite(\chi = \sym{i}, 5, \sym{A}(\chi))
\end{align*}
And suppose that variable \var{i} was already iterated, i.e.~$\theta^{\vec{\kappa}}(\var{i}) = \kappa_1 + \kappa_2 + \sym{i}$. The path counters $\kappa_1$ and $\kappa_2$ are newly introduced path counters for backbone paths $l_1$ and $l_2$ respectively. The loop at Algorithm~\ref{alg:iterateVar} is executed twice. First for backbone path $l_1$ and then for $l_2$. At the first execution, for the backbone path $l_1$, the expressions $e$ and $e'$ are as follows
\begin{align*}
e & \equiv \lambda \chi~.~\sym{A}(\chi) \\
e' & \equiv \lambda \chi~.~\ite(\chi = \kappa_1 + \kappa_2 + \sym{i}, 2(\kappa_1 + \kappa_2 + \sym{i}) + 1, \sym{A}(\chi)),
\end{align*}
For the backbone path $l_1$ we have $pc(l_1)\theta^{\vec{\kappa}} \equiv (\kappa_1 + \kappa_2 + \sym{i}) \mod 2 = 0$. According to Table~\ref{tab:iterateVarValue_array} we receive the following values:
\begin{align*}
w_1 & \equiv (\kappa_1 + \kappa_2 + \sym{i})[\vec{\kappa}/\vec{\tau}] = \tau_1 + \tau_2 + \sym{i} \\
g_1 & \equiv (\kappa_1, \kappa_2)^T \\
\zeta(\{ \tau_1 + \tau_2 + \sym{i} \},1,1) & \equiv \chi \ne \tau_1 + \tau_2 + \sym{i} \\
\phi_1(\{ \tau_1 + \tau_2 + \sym{i} \},\{ \matr{\kappa_1 \\ \kappa_2} \}) & \equiv \forall \matr{\tau'_1 \\ \tau'_2} \left( \matr{\tau_1 \\ \tau_2} < \matr{\tau'_1 \\ \tau'_2} < \matr{\kappa_1 \\ \kappa_2} \rightarrow \chi \ne \tau'_1 + \tau'_2 + \sym{i} \right) \\
\gamma_1 & \equiv \true \\
\psi[\vec{\kappa}/\vec{\tau}] & \equiv (pc(l_1)\theta^{\vec{\kappa}})[\vec{\kappa}/\vec{\tau}] \equiv (\tau_1 + \tau_2 + \sym{i}) \mod 2 = 0 \\
h_1 & \equiv \exists \matr{\tau_1 \\ \tau_2}~.~\chi = \tau_1 + \tau_2 + \sym{i} \wedge \vec{0} \leq \matr{\tau_1 \\ \tau_2} < \matr{\kappa_1 \\ \kappa_2} \wedge \phi_1(\{ \tau_1 + \tau_2 + \sym{i} \},\{ \matr{\kappa_1 \\ \kappa_2} \})) \\
t_1 & \equiv (2(\kappa_1 + \kappa_2 + \sym{i}) + 1)[\vec{\kappa}/\vec{\tau}]\{(\tau_1 + \tau_2 + \sym{i})/\chi\} = 2\chi + 1 \\
\end{align*}
Therefore new value of $e$ is
\begin{align*}
\lambda \chi~.~\ite(& \exists \matr{\tau_1 \\ \tau_2}~.~\chi = \tau_1 + \tau_2 + \sym{i} \wedge \vec{0} \leq \matr{\tau_1 \\ \tau_2} < \matr{\kappa_1 \\ \kappa_2} \wedge \\ & \hspace{1.5cm} \forall \matr{\tau'_1 \\ \tau'_2} \left( \matr{\tau_1 \\ \tau_2} < \matr{\tau'_1 \\ \tau'_2} < \matr{\kappa_1 \\ \kappa_2} \rightarrow \chi \ne \tau'_1 + \tau'_2 + \sym{i} \right) \wedge \\ & \hspace{1.5cm} (\tau_1 + \tau_2 + \sym{i}) \mod 2 = 0, \\ & 2\chi + 1, \sym{A}(\chi))
\end{align*}
Since $\tau_1 + \tau_2 + \sym{i}$ is a monotone function, the condition $\phi_1$ is redundant and the expression $e$ can be simplified into
\begin{align*}
\lambda \chi~.~\ite(\exists \matr{\tau_1 \\ \tau_2}~.~\chi = \tau_1 + \tau_2 + \sym{i} \wedge \vec{0} \leq \matr{\tau_1 \\ \tau_2} < \matr{\kappa_1 \\ \kappa_2} \wedge (\tau_1 + \tau_2 + \sym{i}) \mod 2 = 0, 2\chi + 1, \sym{A}(\chi))
\end{align*}
At the second execution of the loop at Algorithm~\ref{alg:iterateVar}, for the backbone path $l_2$, the expressions $e$ and $e'$ are as follows
\begin{align*}
e & \equiv \lambda \chi~.~\ite(\exists \matr{\tau_1 \\ \tau_2}~.~\chi = \tau_1 + \tau_2 + \sym{i} \wedge \vec{0} \leq \matr{\tau_1 \\ \tau_2} < \matr{\kappa_1 \\ \kappa_2} \wedge (\tau_1 + \tau_2 + \sym{i}) \mod 2 = 0, 2\chi + 1, \sym{A}(\chi)) \\
e' & \equiv \lambda \chi~.~\ite(\chi = \kappa_1 + \kappa_2 + \sym{i}, 5, \sym{A}(\chi)).
\end{align*}
For the backbone path $l_2$ we have $pc(l_2)\theta^{\vec{\kappa}} \equiv (\kappa_1 + \kappa_2 + \sym{i}) \mod 2 \ne 0$. According to Table~\ref{tab:iterateVarValue_array} we receive the following values:
\begin{align*}
w_1 & \equiv \tau_1 + \tau_2 + \sym{i} \\
w_2 & \equiv (\kappa_1 + \kappa_2 + \sym{i})[\vec{\kappa}/\vec{\tau}] = \tau_1 + \tau_2 + \sym{i} \\
g_1 & \equiv (\kappa_1, \kappa_2)^T \\
g_2 & \equiv (\kappa_1, \kappa_2)^T \\
\zeta(\{ \tau_1 + \tau_2 + \sym{i} \},1,1) & \equiv \chi \ne \tau_1 + \tau_2 + \sym{i} \\
\zeta(\{ \tau_1 + \tau_2 + \sym{i} \},1,2) & \equiv \chi \ne \tau_1 + \tau_2 + \sym{i} \\
\phi_1(\{ w_1, w_2 \},\{ g_1, g_2 \}) & \equiv \phi_1(\{ w_1\},\{ g_1\})~~~ \textit{(since $w_1 = w_2$ and $g_1 = g_2$)} \\ & \equiv \forall \matr{\tau'_1 \\ \tau'_2} \left( \matr{\tau_1 \\ \tau_2} < \matr{\tau'_1 \\ \tau'_2} < \matr{\kappa_1 \\ \kappa_2} \rightarrow \chi \ne \tau'_1 + \tau'_2 + \sym{i} \right) \\
\zeta(\{ \tau_1 + \tau_2 + \sym{i} \},2,1) & \equiv \chi \ne \tau_1 + \tau_2 + \sym{i} \\
\zeta(\{ \tau_1 + \tau_2 + \sym{i} \},2,2) & \equiv \chi \ne \tau_1 + \tau_2 + \sym{i} \\
\phi_2(\{ w_1, w_2 \},\{ g_1, g_2 \}) & \equiv \phi_2(\{ w_1\},\{ g_1\})~~~ \textit{(since $w_1 = w_2$ and $g_1 = g_2$)} \\ & \equiv \forall \matr{\tau'_1 \\ \tau'_2} \left( \matr{\tau_1 \\ \tau_2} < \matr{\tau'_1 \\ \tau'_2} < \matr{\kappa_1 \\ \kappa_2} \rightarrow \chi \ne \tau'_1 + \tau'_2 + \sym{i} \right) \\
\gamma_1 & \equiv (\tau_1 + \tau_2 + \sym{i}) \mod 2 = 0 \\
\psi[\vec{\kappa}/\vec{\tau}] & \equiv (pc(l_1)\theta^{\vec{\kappa}})[\vec{\kappa}/\vec{\tau}] \equiv (\tau_1 + \tau_2 + \sym{i}) \mod 2 \ne 0 \\
h_1 & \equiv \exists \matr{\tau_1 \\ \tau_2}~.~\chi = \tau_1 + \tau_2 + \sym{i} \wedge \vec{0} \leq \matr{\tau_1 \\ \tau_2} < \matr{\kappa_1 \\ \kappa_2} \wedge \\ & \hspace{2cm} \phi_1(\{ w_1, w_2 \},\{ g_1, g_2 \}) \wedge (\tau_1 + \tau_2 + \sym{i}) \mod 2 = 0 \\
h_2 & \equiv \exists \matr{\tau_1 \\ \tau_2}~.~\chi = \tau_1 + \tau_2 + \sym{i} \wedge \vec{0} \leq \matr{\tau_1 \\ \tau_2} < \matr{\kappa_1 \\ \kappa_2} \wedge \\ & \hspace{2cm} \phi_2(\{ w_1, w_2 \},\{ g_1, g_2 \}) \wedge (\tau_1 + \tau_2 + \sym{i}) \mod 2 \ne 0 \\
t_1 & \equiv 2\chi + 1 \\
t_2 & \equiv 5[\vec{\kappa}/\vec{\tau}]\{(\tau_1 + \tau_2 + \sym{i})/\chi\} = 5 \\
\end{align*}
Therefore the resulting iterated value for array \var{A} is
\begin{align*}
\lambda \chi~.~\ite(& \exists \matr{\tau_1 \\ \tau_2}~.~\chi = \tau_1 + \tau_2 + \sym{i} \wedge \vec{0} \leq \matr{\tau_1 \\ \tau_2} < \matr{\kappa_1 \\ \kappa_2} \wedge \\ & \hspace{1.5cm} \forall \matr{\tau'_1 \\ \tau'_2} \left( \matr{\tau_1 \\ \tau_2} < \matr{\tau'_1 \\ \tau'_2} < \matr{\kappa_1 \\ \kappa_2} \rightarrow \chi \ne \tau'_1 + \tau'_2 + \sym{i} \right) \wedge \\ & \hspace{1.5cm} (\tau_1 + \tau_2 + \sym{i}) \mod 2 = 0, 2\chi + 1, \\ & \ite( \exists \matr{\tau_1 \\ \tau_2}~.~\chi = \tau_1 + \tau_2 + \sym{i} \wedge \vec{0} \leq \matr{\tau_1 \\ \tau_2} < \matr{\kappa_1 \\ \kappa_2} \wedge \\ & \hspace{2.1cm} \forall \matr{\tau'_1 \\ \tau'_2} \left( \matr{\tau_1 \\ \tau_2} < \matr{\tau'_1 \\ \tau'_2} < \matr{\kappa_1 \\ \kappa_2} \rightarrow \chi \ne \tau'_1 + \tau'_2 + \sym{i} \right) \wedge \\ & \hspace{2.1cm} (\tau_1 + \tau_2 + \sym{i}) \mod 2 \ne 0, 5, \sym{A}(\chi))
\end{align*}
We can see, that conditions $\phi_1$ and $\phi_2$ are redundant in the expression (since $w_1 = w_2$ and $g_1 = g_2$). Therefore we can simplify the resulting value into
\begin{align*}
\lambda \chi~.~\ite(& \exists \matr{\tau_1 \\ \tau_2}~.~\chi = \tau_1 + \tau_2 + \sym{i} \wedge \vec{0} \leq \matr{\tau_1 \\ \tau_2} < \matr{\kappa_1 \\ \kappa_2} \wedge (\tau_1 + \tau_2 + \sym{i}) \mod 2 = 0, 2\chi + 1, \\ & \ite( \exists \matr{\tau_1 \\ \tau_2}~.~\chi = \tau_1 + \tau_2 + \sym{i} \wedge \vec{0} \leq \matr{\tau_1 \\ \tau_2} < \matr{\kappa_1 \\ \kappa_2} \wedge (\tau_1 + \tau_2 + \sym{i}) \mod 2 \ne 0, 5, \sym{A}(\chi))
\end{align*}
\end{exm}

\begin{exm}
Let us consider the following C++ program:
\begin{verbatim}
for (int i = 0; i < m; ++i) {
  id = B[i*(n+1)+1];
  for (int j = 0; j < n; ++j)
    A[id][j] = B[id*(n+1)+j+2];
}
\end{verbatim}
The program consists of two nested loops. We first express an iterated value of 2D array \var{A} after the inner loop. Then we apply the same procedure to express \var{A} after the outer one.

After symbolic execution of the inner loop's body variable \var{A} has a value $\lambda \vec{\chi}~.~\ite(\vec{\chi} = (\sym{id}, \sym{j}), \sym{B}(\sym{id} (\sym{n} + 1) + \sym{j} + 2), \sym{A}(\chi))$. And suppose that variable \var{j} was already iterated, i.e.~$\theta^{\vec{\kappa}}(\var{j}) = \kappa_1 + \sym{j}$, where $\kappa_1$ is a path counter introduced for the only backbone path of inner loop. Then expressions $e$ and $e'$ from Algorithm~\ref{alg:iterateVar} are assigned as follows
\begin{align*}
e & \equiv \lambda \vec{\chi}~.~\sym{A}(\vec{\chi}) \\
e' & \equiv \lambda \vec{\chi}~.~\ite(\vec{\chi} = (\sym{id}, \kappa_1 + \sym{j})^T, \sym{B}(\sym{id} (\sym{n} + 1) + \kappa_1 + \sym{j} + 2), \sym{A}(\chi)).
\end{align*}
Note that $\vec{\chi} = (\chi_1, \chi_2)$. Since there is only single backbone path $l$ in the inner loop, $pc(l) \equiv \true$. According to Table~\ref{tab:iterateVarValue_array} we receive the following values:
\begin{align*}
\vec{w}_1 \equiv &~ (\sym{id}, \kappa_1 + \sym{j})^T[\kappa_1/\tau_1] = (\sym{id}, \tau_1 + \sym{j})^T \\
g_1 \equiv &~ \kappa_1 \\
\zeta(\{ \matr{\sym{id} \\ \tau_1 + \sym{j}} \},1,1) \equiv &~ \vec{\chi} \ne \matr{\sym{id} \\ \tau_1 + \sym{j}} \\
\phi_1(\{ \matr{\sym{id} \\ \tau_1 + \sym{j}} \},\{ \kappa_1 \}) \equiv &~ \forall \tau'_1 \left( \tau_1 < \tau'_1 < \kappa_1 \rightarrow \vec{\chi} \ne \matr{\sym{id} \\ \tau_1 + \sym{j}} \right) \\
\gamma_1 \equiv &~ \true \\
\psi \equiv &~ \true \\
h_1 \equiv &~ \exists \tau_1~.~\vec{\chi} = \matr{\sym{id} \\ \tau_1 + \sym{j}} \wedge 0 \leq \tau_1 < \kappa_1 \wedge \forall \tau'_1 \left( \tau_1 < \tau'_1 < \kappa_1 \rightarrow \vec{\chi} \ne \matr{\sym{id} \\ \tau_1 + \sym{j}} \right) \\
t_1 \equiv &~ \sym{B}(\sym{id} (\sym{n} + 1) + \kappa_1 + \sym{j} + 2)[\kappa_1/\tau_1]\{\vec{w}_1/\vec{\chi}\} = \\ &~ \sym{B}(\sym{id} (\sym{n} + 1) + \tau_1 + \sym{j} + 2)\{(\sym{id}, \tau_1 + \sym{j})/\vec{\chi}\} = \\ &~ \sym{B}(\chi_1 (\sym{n} + 1) + \chi_2 + 2) \\
\end{align*}
Therefore the resulting iterated value for array \var{A} from the inner loop is
\begin{align*}
\lambda \vec{\chi}~.~\ite(& \exists \tau_1~.~\vec{\chi} = (\sym{id}, \tau_1 + \sym{j})^T \wedge 0 \leq \tau_1 < \kappa_1 \wedge \forall \tau'_1 (\tau_1 < \tau'_1 < \kappa_1 \rightarrow \vec{\chi} \ne (\sym{id}, \tau_1 + \sym{j})^T), \\ & \sym{B}(\chi_1 (\sym{n} + 1) + \chi_2 + 2), \sym{A}(\chi))
\end{align*}
Since both functions $\sym{id}$ and $\tau_1 + \sym{j}$ are monotone and there is not other write to \var{A}, we can simplify the expression into
\begin{equation*}
\lambda \vec{\chi}~.~\ite(\exists \tau_1~.~\vec{\chi} = (\sym{id}, \tau_1 + \sym{j})^T \wedge 0 \leq \tau_1 < \kappa_1, \sym{B}(\chi_1 (\sym{n} + 1) + \chi_2 + 2), \sym{A}(\chi))
\end{equation*}

We may proceed to the outer loop. There we first eliminate imported path counter $\kappa_1$ such that we substitute all its occurrences by an expression $\max\{ 0, \sym{n} \}$. We discuss the elimination of imported path counters in Section~\ref{sec:IterateTheta} and computation of an expression to be substituted in Section~\ref{sec:iterNestedLoop}. We also describe the computation of the expression $\max\{ 0, \sym{n} \}$ in Example~\ref{exm:linFnOne}. Also note that values of variables \var{j} and \var{id} are set to $0$ and $\sym{B}(\sym{i} (\sym{n} + 1) + 1)$ respectively, before entering the inner loop. Therefore, after symbolic execution of the outer loop's body the variable \var{A} has a value
\begin{equation*}
\lambda \vec{\chi}~.~\ite(\exists \tau_1~.~\vec{\chi} = (\sym{B}(\sym{i} (\sym{n} + 1) + 1), \tau_1)^T \wedge 0 \leq \tau_1 < \max\{ 0, \sym{n} \}, \sym{B}(\chi_1 (\sym{n} + 1) + \chi_2 + 2), \sym{A}(\chi)).
\end{equation*}
And suppose that variable \var{i} was already iterated, i.e.~$\theta^{\vec{\kappa}}(\var{i}) = \kappa + \sym{i}$, where $\kappa$ is a path counter introduced for the only backbone path of the outer loop. Then expressions $e$ and $e'$ from Algorithm~\ref{alg:iterateVar} are assigned as follows
\begin{align*}
e & \equiv \lambda \vec{\chi}~.~\sym{A}(\vec{\chi}) \\
e' & \equiv \lambda \vec{\chi}~.~\ite(\exists \tau_1~.~\vec{\chi} = (\sym{B}((\kappa + \sym{i}) (\sym{n} + 1) + 1), \tau_1)^T \wedge 0 \leq \tau_1 < \max\{ 0, \sym{n} \}, \\ & \hspace{2cm} \sym{B}(\chi_1 (\sym{n} + 1) + \chi_2 + 2), \sym{A}(\chi)).
\end{align*}
Since there is only a single backbone path $l$ in the inner loop, $pc(l) \equiv \true$. According to Table~\ref{tab:iterateVarValue_array} we receive the following values:
\begin{align*}
\vec{w}_1 \equiv &~ (\sym{B}((\kappa + \sym{i}) (\sym{n} + 1) + 1), \tau_1)^T[\kappa/\tau] = (\sym{B}((\tau + \sym{i}) (\sym{n} + 1) + 1), \tau_1)^T \\
\vec{g}_1 \equiv &~ (\max\{ 0, \sym{n} \}, \kappa)^T \\
\zeta(\{ \vec{w}_1 \},1,1) \equiv &~ \vec{\chi} \ne (\sym{B}((\tau + \sym{i}) (\sym{n} + 1) + 1), \tau_1)^T \\
\phi_1(\{ \vec{w}_1 \},\{ \vec{g}_1 \}) \equiv &~ \forall \matr{\tau'_1 \\ \tau'} \left( \matr{\tau_1 \\ \tau} < \matr{\tau'_1 \\ \tau'} < \matr{\max\{ 0, \sym{n} \} \\ \kappa} \rightarrow \vec{\chi} \ne \matr{\sym{B}((\tau' + \sym{i}) (\sym{n} + 1) + 1) \\ \tau'_1} \right) \\
\gamma_1 \equiv &~ \true \\
\psi \equiv &~ \true \\
h_1 \equiv &~ \exists \matr{\tau_1 \\ \tau}~.~\vec{\chi} = \matr{\sym{B}((\tau + \sym{i}) (\sym{n} + 1) + 1) \\ \tau_1} \wedge \vec{0} \leq \matr{\tau_1 \\ \tau} < \matr{\max\{ 0, \sym{n} \} \\ \kappa} \wedge \phi_1(\{ \vec{w}_1 \},\{ \vec{g}_1 \}) \\
t_1 \equiv &~ \sym{B}(\chi_1 (\sym{n} + 1) + \chi_2 + 2)[\kappa/\tau]\{\vec{w}_1/\vec{\chi}\} = \sym{B}(\chi_1 (\sym{n} + 1) + \chi_2 + 2)
\end{align*}
Therefore the resulting iterated value for array \var{A} is
\begin{align*}
\lambda \vec{\chi}~.~\ite(& \exists \matr{\tau_1 \\ \tau}~.~\vec{\chi} = \matr{\sym{B}((\tau + \sym{i}) (\sym{n} + 1) + 1) \\ \tau_1} \wedge \vec{0} \leq \matr{\tau_1 \\ \tau} < \matr{\max\{ 0, \sym{n} \} \\ \kappa} \wedge, \\ & \hspace{1.5cm} \forall \matr{\tau'_1 \\ \tau'} \left( \matr{\tau_1 \\ \tau} < \matr{\tau'_1 \\ \tau'} < \matr{\max\{ 0, \sym{n} \} \\ \kappa} \rightarrow \vec{\chi} \ne \matr{\sym{B}((\tau' + \sym{i}) (\sym{n} + 1) + 1) \\ \tau'_1} \right) \\ & \sym{B}(\chi_1 (\sym{n} + 1) + \chi_2 + 2), \sym{A}(\chi))
\end{align*}
\end{exm}

%%%%%%%%%%%%%%%%%%%%%%%%%%%%%%%%%%%%%%%%%%%%%%%%%%%%%%%%%%%%%%%%%%%%%%%%
%%%%%%%%%%%%%%%%%%%%%%%%%%%%%%%%%%%%%%%%%%%%%%%%%%%%%%%%%%%%%%%%%%%%%%%%
\subsection{Building formula $S_\gamma$ and using SMT solver on it}
\label{sec:linRelExamples}

\begin{exm} \label{exm:linFnOne}
Let us consider the following C++ program
\begin{verbatim}
for (int i = 0; i < m; ++i)
  for (int j = 0; j < n; ++j)
    A[i][j] = 0;
\end{verbatim}
There is only single backbone path (the body of the outer loop) in $\bbt$ of the program. Let $v_1, \ldots, v_6$ be all its vertices. Then $v_1$ is $l_s$, $v_7$ is $l_t$, $v_4$ is the only component vertex of $\bbt$, and $\lst(v_5)$ is the only exit vertex of SCC $\mathcal{C}_{v_4}$. Then after a symbolic execution of $\bbt$ the map $\Psi$ has the following content:
\begin{itemize}
\item $\Psi(v_1) = \true$
\item $\Psi(v_2) = \sym{i} < \sym{m}$
\item $\Psi(v_3) = \true$
\item $\Psi(v_4) = \forall \tau_{v_4}~(0 \leq \tau_{v_4} < \kappa_{v_4} \rightarrow \tau_{v_4} < \sym{n})$
\item $\Psi(v_5) = \kappa_{v_4} \geq \sym{n}$
\item $\Psi(v_6) = \true$
\item $\Psi(v_7) = \true$
\end{itemize}
After elimination of imported counter $\kappa_{v_4}$ we receive (only changes are shown)
\begin{itemize}
\item $\Psi(v_4) = \forall s~(0 \leq s < \sym{s}_{v_4} \rightarrow s < \sym{n})$
\item $\Psi(v_5) = \sym{s}_{v_4} \geq \sym{n}$
\end{itemize}
The symbolic state $\Theta(v_7)$ stores the following values of variables \var{i} and \var{j}
\begin{itemize}
\item $\Theta(v_7)(\var{i}) = \sym{i} + 1$
\item $\Theta(v_7)(\var{j}) = \kappa_{v_4}$
\end{itemize}
After elimination of imported counter $\kappa_{v_4}$ we receive (only changes are shown)
\begin{itemize}
\item $\Theta(v_7)(\var{j}) = \sym{s}_{v_4}$
\end{itemize}
Let us now suppose that $\theta^{\vec{\kappa}}$ stores the following value of \var{i}
\begin{itemize}
\item $\theta^{\vec{\kappa}}(\var{i}) = \kappa + \sym{i}$
\end{itemize}
Note that $\kappa$ is a fresh path counter introduced for our backbone path. We are ready to build formula $S_{v_4}$. We start with $\Gamma(v_4, \Omega, \theta^{\vec{\kappa}})$:
\begin{equation*}
\Gamma(v_4, \Omega, \theta^{\vec{\kappa}}) \equiv (0 \leq \sym{s}_{v_4} - 1 \rightarrow \sym{s}_{v_4} - 1 < \sym{n}) \wedge \sym{s}_{v_4} \geq \sym{n}
\end{equation*}
Note that $\Omega(v_4)$ is evaluated according to the first case, since $v_4$ is a component vertex. But $\Omega(v_5)$ is evaluated according to the first case. Also note that substitution of $\theta^{\vec{\kappa}}$ into formulae did not incorporate introduced counter $\kappa$ into the resulting formula. Therefore we do not need to introduce matrix $M$. Since $\Gamma(v_4, \Omega, \theta^{\vec{\kappa}})$ contains only single basic symbol $\sym{n}$, thus $\vec{\sym{a}} = (\sym{n})^T = \sym{n}$. And vector $\vec{w} = (w_1, w_2)^T$, because we have the only basic symbol in the formula. The formula $S_{v_4}$ looks as follows
\begin{equation*}
S_{v_4} \equiv \exists w_1, w_2 \forall n, \sym{s}_{v_4} \left( (\sym{s}_{v_4} \ge 0 \wedge \Gamma(v_4, \Omega, \theta^{\vec{\kappa}})) \rightarrow \sym{s}_{v_4} = \max\{ 0, w_1\sym{n} + w_2 \} \right)
\end{equation*}
And when we substitute formula $\Gamma$ into $S_{v_4}$ we obtain
\begin{equation*}
S_{v_4} \equiv \exists w_1, w_2 \forall n, \sym{s}_{v_4} \left( (\sym{s}_{v_4} \ge 0 \wedge (0 \leq \sym{s}_{v_4} - 1 \rightarrow \sym{s}_{v_4} - 1 < \sym{n}) \wedge \sym{s}_{v_4} \geq \sym{n}) \rightarrow \sym{s}_{v_4} = \max\{ 0, w_1\sym{n} + w_2 \} \right)
\end{equation*}
Then we ask an SMT solver, whether the formula is satisfiable or not. And if so we further ask for a model to get values of integers $w_1$ and $w_2$.
We see, that formula is satisfiable and $w_1 = 1$ and $w_2 = 0$. Therefore we return a symbolic expression:
\begin{equation*}
\max\{ 0, n \}
\end{equation*}
\end{exm}

\begin{exm}
Let us consider the following C++ program
\begin{verbatim}
for (int i = 0; i < m; ++i)
  for (int j = i; j < n; ++j)
    A[i][j] = 0;
\end{verbatim}
There is only single backbone path (the body of the outer loop) in $\bbt$ of the program. Let $v_1, \ldots, v_6$ be all its vertices. Then $v_1$ is $l_s$, $v_7$ is $l_t$, $v_4$ is the only component vertex of $\bbt$, and $\lst(v_5)$ is the only exit vertex of SCC $\mathcal{C}_{v_4}$. Then after a symbolic execution of $\bbt$ the map $\Psi$ has the following content:
\begin{itemize}
\item $\Psi(v_1) = \true$
\item $\Psi(v_2) = \sym{i} < \sym{m}$
\item $\Psi(v_3) = \true$
\item $\Psi(v_4) = \forall \tau_{v_4}~(0 \leq \tau_{v_4} < \kappa_{v_4} \rightarrow \tau_{v_4} + \sym{i} < \sym{n})$
\item $\Psi(v_5) = \kappa_{v_4} + \sym{i} \geq \sym{n}$
\item $\Psi(v_6) = \true$
\item $\Psi(v_7) = \true$
\end{itemize}
After elimination of imported counter $\kappa_{v_4}$ we receive (only changes are shown)
\begin{itemize}
\item $\Psi(v_4) = \forall s~(0 \leq s < \sym{s}_{v_4} \rightarrow s + \sym{i} < \sym{n})$
\item $\Psi(v_5) = \sym{s}_{v_4} + \sym{i} \geq \sym{n}$
\end{itemize}
The symbolic state $\Theta(v_7)$ stores the following values of variables \var{i} and \var{j}
\begin{itemize}
\item $\Theta(v_7)(\var{i}) = \sym{i} + 1$
\item $\Theta(v_7)(\var{j}) = \kappa_{v_4} + \sym{i}$
\end{itemize}
After elimination of imported counter $\kappa_{v_4}$ we receive (only changes are shown)
\begin{itemize}
\item $\Theta(v_7)(\var{j}) = \sym{s}_{v_4} + \sym{i}$
\end{itemize}
Let us now suppose that $\theta^{\vec{\kappa}}$ stores the following value of \var{i} and \var{j}
\begin{itemize}
\item $\theta^{\vec{\kappa}}(\var{i}) = \kappa + \sym{i}$
\item $\theta^{\vec{\kappa}}(\var{j}) = \sym{s}_{v_4} + \kappa + \sym{i}$
\end{itemize}
Note that $\kappa$ is a fresh path counter introduced for our backbone path. We are ready to build formula $S_{v_4}$. We start with $\Gamma(v_4, \Omega, \theta^{\vec{\kappa}})$:
\begin{equation*}
\Gamma(v_4, \Omega, \theta^{\vec{\kappa}}) \equiv (0 \leq \sym{s}_{v_4} - 1 \rightarrow \sym{s}_{v_4} - 1 + \kappa + \sym{i} < \sym{n}) \wedge \sym{s}_{v_4} + \kappa + \sym{i} \geq \sym{n}
\end{equation*}
Note that $\Omega(v_4)$ is evaluated according to the first case, since $v_4$ is a component vertex. But $\Omega(v_5)$ is evaluated according to the first case. Also note that substitution of $\theta^{\vec{\kappa}}$ into formulae incorporated the introduced counter $\kappa$ into the resulting formula. Therefore we have $\vec{\kappa} = (\kappa)^T = \kappa$ and matrix $M = (m_1, m_2, m_3)^T$ is of type $(2+1) \times 1$, since there are two basic symbols $\sym{n},\sym{i}$ and just one counter $\kappa$ in the formula $\Gamma$. Further we have $\vec{\sym{a}} = (\sym{n}, \sym{i})^T$, and vector $\vec{w} = (w_1, w_2, w_3)^T$. The formula $S_{v_4}$ looks as follows
\begin{equation*}
S_{v_4} \equiv \exists M,\vec{w}, \forall \vec{a},\kappa, \sym{s}_{v_4} \left( (\kappa \geq 0 \wedge \sym{s}_{v_4} \geq 0 \wedge \Gamma(v_4, \Omega, \theta^{\vec{\kappa}})) \rightarrow \sym{s}_{v_4} = \max\{ 0, (M \kappa + \vec{w})^T \matr{\vec{\sym{a}} \\ 1} \} \right)
\end{equation*}
And when we substitute formula $\Gamma$ into $S_{v_4}$ we obtain
\begin{align*}
S_{v_4} \equiv \exists M,\vec{w}, \forall \vec{a},\kappa, \sym{s}_{v_4} & \left( (\kappa \geq 0 \wedge \sym{s}_{v_4} \geq 0 \wedge (0 \leq \sym{s}_{v_4} - 1 \rightarrow \sym{s}_{v_4} - 1 + \kappa + \sym{i} < \sym{n}) \wedge \sym{s}_{v_4} + \kappa + \sym{i} \geq \sym{n}) \right. \\ & \left.
\rightarrow \sym{s}_{v_4} = \max\{ 0, (M \kappa + \vec{w})^T \matr{\vec{\sym{a}} \\ 1} \} \right)
\end{align*}
Then we ask Z3 SMT solver, whether the formula is satisfiable or not. And if so we further ask for a model to get values of integers $m_i$ and $w_j$.
We see, that formula is satisfiable and $m_1 = m_2 = 0, m_3 = -1$, and $w_1 = 1, w_2 = -1, w_3 = 0$. Therefore we return a symbolic expression:
\begin{equation*}
\max\{ 0, -\kappa + n - i \}
\end{equation*}
\end{exm}

%%%%%%%%%%%%%%%%%%%%%%%%%%%%%%%%%%%%%%%%%%%%%%%%%%%%%%%%%%%%%%%%%%%%%%%%
%%%%%%%%%%%%%%%%%%%%%%%%%%%%%%%%%%%%%%%%%%%%%%%%%%%%%%%%%%%%%%%%%%%%%%%%
\subsection{Hello}

\begin{verbatim}
char H[6] = "Hello";
int h = 0;
for (int i = 0; A[i] != 0; ++i) {
  int j = i, k = 0;
  while (H[k] != 0 && A[j] != 0 && A[j] == H[k]) {
    ++j;
    ++k;
  }
  if (H[k] == 0) { h = 1; break; }
  if (A[j] == 0) break;
}
if (h == 1)
  assert(false);
\end{verbatim}

\begin{figure}%
\begin{center}
\begin{tabular}{c}
\includegraphics[width=11cm, height=16cm]{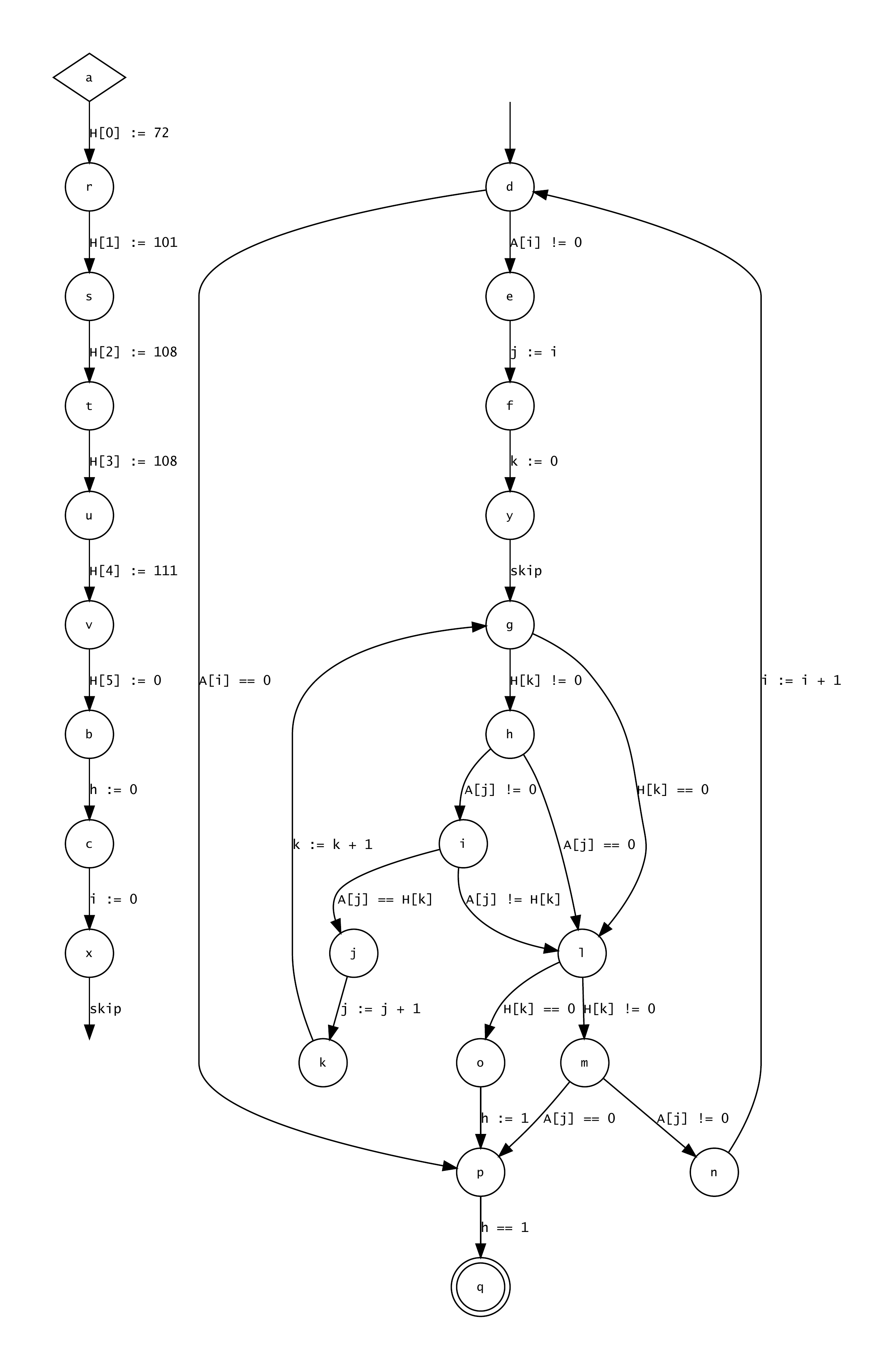}
\end{tabular}
\end{center}
\caption{Program \texttt{Hello}}%
\label{fig:Hello}%
\end{figure}

\begin{figure}%
\begin{center}
\begin{tabular}{cc}
\includegraphics[width=7.5cm, height=15cm]{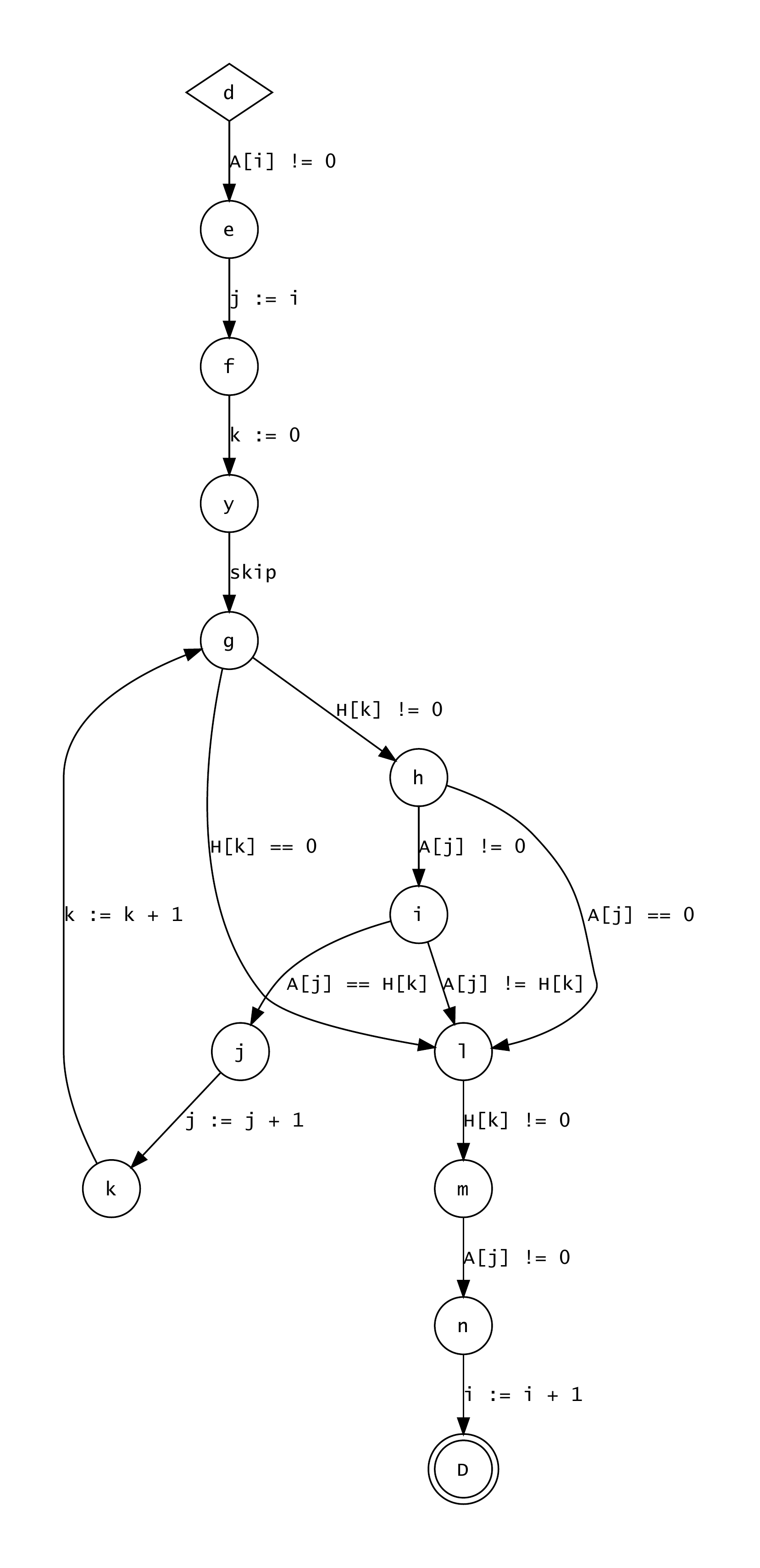} &
\includegraphics[width=3cm, height=8cm]{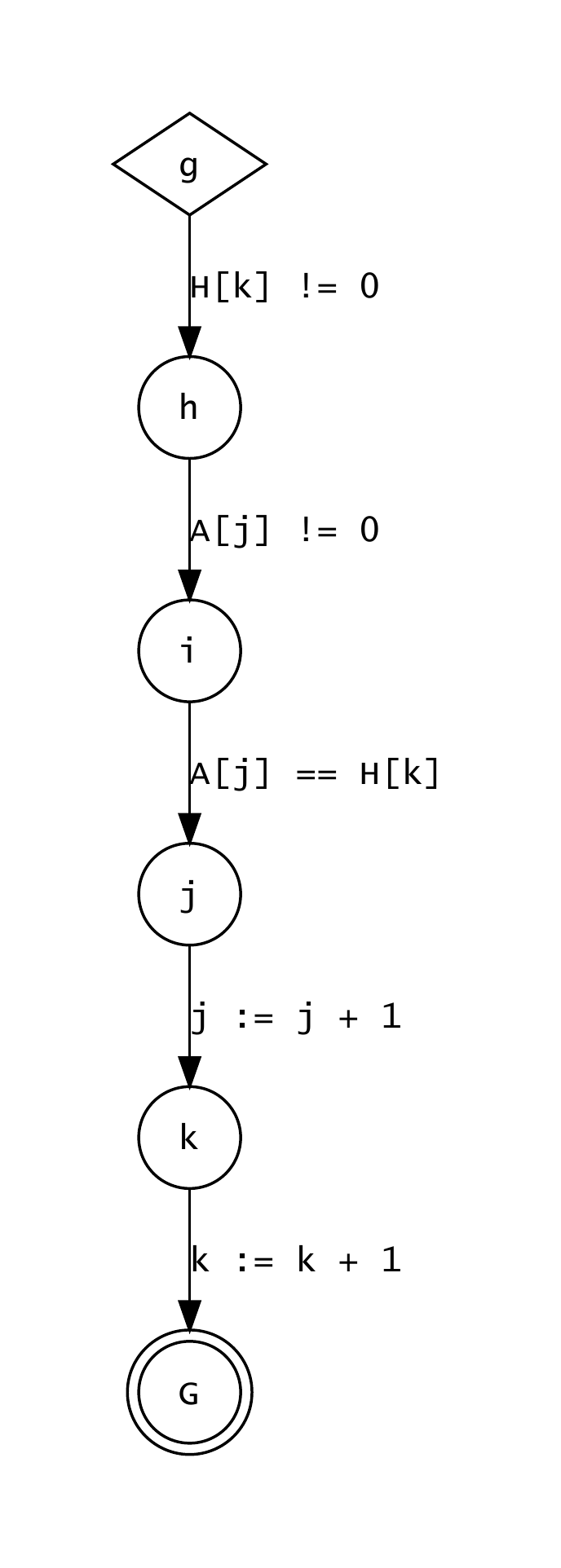}
\end{tabular}
\end{center}
\caption{Induced program of loops in \texttt{Hello} (a) Outer loop (b) Inner loop}%
\label{fig:HelloOuterInner}%
\end{figure}

\paragraph{Inner loop} After symbolic execution of backbone tree of the induced program we receive the following properties. Note that resulting backbone tree has only single backbone path $ghijkG$. \\
Function $\Psi$:
\begin{align*}
\Psi(g) & = \true \\
\Psi(gh) & = \sym{H}(\sym{k}) \ne 0 \\
\Psi(ghi) & = \sym{A}(\sym{j}) \ne 0 \\
\Psi(ghij) & = \sym{A}(\sym{j}) = \sym{H}(\sym{k}) \\
\Psi(ghijk) & = \true \\
\Psi(ghijkG) & = \true \\
\end{align*}
Function $\Theta$:
\begin{align*}
\Theta(ghijkG)(\var{j}) & = \sym{j} + 1 \\
\Theta(ghijkG)(\var{k}) & = \sym{k} + 1 \\
\Theta(ghijkG)(\var{A}) & = \lambda \chi~.~\sym{A}(\chi) \\
\Theta(ghijkG)(\var{H}) & = \lambda \chi~.~\sym{H}(\chi) \\
\end{align*}
Note that $\bar{\Psi} \equiv \Psi$ and $\bar{\Theta} \equiv \Theta$. Therefore, a symbolic state $\theta^{\vec{\kappa}}$ is:
\begin{align*}
\theta^{\vec{\kappa}}(\var{j}) & = \kappa_1 + \sym{j} \\
\theta^{\vec{\kappa}}(\var{k}) & = \kappa_1 + \sym{k} \\
\theta^{\vec{\kappa}}(\var{A}) & = \lambda \chi~.~\sym{A}(\chi) \\
\theta^{\vec{\kappa}}(\var{H}) & = \lambda \chi~.~\sym{H}(\chi) \\
\end{align*}
And a looping condition $\varphi^{\vec{\kappa}}$ is:
\begin{equation*}
\varphi^{\vec{\kappa}} \equiv \forall \tau_1 \left[ 0 \leq \tau_1 < \kappa_1 \rightarrow \left(\sym{H}(\tau_1 + \sym{k}) \ne 0 \wedge \sym{A}(\tau_1 + \sym{j}) \ne 0 \wedge \sym{A}(\tau_1 + \sym{j}) = \sym{H}(\tau_1 + \sym{k}) \right) \right]
\end{equation*}

\paragraph{Outer loop} After symbolic execution of backbone tree of the induced program we receive the following properties. Note that resulting backbone tree has only single backbone path $defyghilmnD$, since backbone paths going through program edges $(g,l)$ or $(h,l)$ are infeasible.\\
Function $\Psi$:
\begin{align*}
\Psi(d) & = \true \\
\Psi(de) & = \sym{A}(\sym{i}) \ne 0 \\
\Psi(def) & = \true \\
\Psi(defy) & = \true \\
\Psi(defyg) & = \forall \tau_1 \left[ 0 \leq \tau_1 < \kappa_1 \rightarrow %\right. \\ & \hspace{2.5cm} \left.
\left(\sym{H}(\tau_1) \ne 0 \wedge \sym{A}(\tau_1 + \sym{i}) \ne 0 \wedge \sym{A}(\tau_1 + \sym{i}) = \sym{H}(\tau_1) \right) \right] \\
\Psi(defygh) & = \sym{H}(\kappa_1) \ne 0 \\
\Psi(defyghi) & = \sym{A}(\kappa_1 + \sym{i}) \ne 0 \\
\Psi(defyghil) & = \sym{A}(\kappa_1 + \sym{i}) \ne \sym{H}(\kappa_1) \\
\Psi(defyghilm) & = \sym{H}(\kappa_1) \ne 0 \\
\Psi(defyghilmn) & = \sym{A}(\kappa_1 + \sym{i}) \ne 0 \\
\Psi(defyghilmnD) & = \true \\
\end{align*}
Function $\Theta$:
\begin{align*}
\Theta(defyghilmnD)(\var{i}) & = \sym{i} + 1 \\
\Theta(defyghilmnD)(\var{j}) & = \kappa_1 + \sym{i} \\
\Theta(defyghilmnD)(\var{k}) & = \kappa_1 \\
\Theta(defyghilmnD)(\var{A}) & = \lambda \chi~.~\sym{A}(\chi) \\
\Theta(defyghilmnD)(\var{H}) & = \lambda \chi~.~\sym{H}(\chi)
\end{align*}
Function $\bar{\Psi}$:
\begin{align*}
\bar{\Psi}(d) & = \true \\
\bar{\Psi}(de) & = \sym{A}(\sym{i}) \ne 0 \\
\bar{\Psi}(def) & = \true \\
\bar{\Psi}(defy) & = \true \\
\bar{\Psi}(defyg) & = \forall s \left[ 0 \leq s < \sym{s}_1 \rightarrow %\right. \\ & \hspace{2.5cm} \left.
\left(\sym{H}(\sym{s}_1) \ne 0 \wedge \sym{A}(\sym{s}_1 + \sym{i}) \ne 0 \wedge \sym{A}(\sym{s}_1 + \sym{i}) = \sym{H}(\sym{s}_1) \right) \right] \\
\bar{\Psi}(defygh) & = \sym{H}(\sym{s}_1) \ne 0 \\
\bar{\Psi}(defyghi) & = \sym{A}(\sym{s}_1 + \sym{i}) \ne 0 \\
\bar{\Psi}(defyghil) & = \sym{A}(\sym{s}_1 + \sym{i}) \ne \sym{H}(\sym{s}_1) \\
\bar{\Psi}(defyghilm) & = \sym{H}(\sym{s}_1) \ne 0 \\
\bar{\Psi}(defyghilmn) & = \sym{A}(\sym{s}_1 + \sym{i}) \ne 0 \\
\bar{\Psi}(defyghilmnD) & = \true \\
\end{align*}
Function $\bar{\Theta}$:
\begin{align*}
\bar{\Theta}(defyghilmnD)(\var{i}) & = \sym{i} + 1 \\
\bar{\Theta}(defyghilmnD)(\var{j}) & = \sym{s}_1 + \sym{i} \\
\bar{\Theta}(defyghilmnD)(\var{k}) & = \sym{s}_1 \\
\bar{\Theta}(defyghilmnD)(\var{A}) & = \lambda \chi~.~\sym{A}(\chi) \\
\bar{\Theta}(defyghilmnD)(\var{H}) & = \lambda \chi~.~\sym{H}(\chi) \\
\bar{\Theta}(defyghilmnD)(\var{s}_1) & = \sym{s}_1
\end{align*}
Symbolic state $\theta^{\vec{\kappa}}$ after iteration of regular variables:
\begin{align*}
\theta^{\vec{\kappa}}(\var{i}) & = \kappa_2 + \sym{i} \\
\theta^{\vec{\kappa}}(\var{j}) & = \star \\
\theta^{\vec{\kappa}}(\var{k}) & = \star \\
\theta^{\vec{\kappa}}(\var{A}) & = \lambda \chi~.~\sym{A}(\chi) \\
\theta^{\vec{\kappa}}(\var{H}) & = \lambda \chi~.~\sym{H}(\chi),
\end{align*}
where $\kappa_2$ is a fresh counter introduce that single backbone path.\\
Formula $\Gamma(defyg, \Omega, \theta^{\vec{\kappa}}, \sym{x})$ looks as follows:
\begin{align*}
\Gamma(defyg,\Omega,\theta^{\vec{\kappa}}, \sym{x}) \equiv & ( 0 \leq \sym{s}_1 - 1 \rightarrow (\sym{H}(\sym{s}_1 - 1, \sym{x}) \ne 0~\wedge \\ & \hspace{2.6cm} \sym{A}(\sym{s}_1 - 1 + \kappa_2 + \sym{i}, \sym{x}) \ne 0~\wedge \\ & \hspace{2.6cm} \sym{A}(\sym{s}_1 - 1 + \kappa_2 + \sym{i}, \sym{x}) = \sym{H}(\sym{s}_1 - 1, \sym{x})))~\wedge \\ & \sym{H}(\sym{s}_1, \sym{x}) \ne 0~\wedge \\ & \sym{A}(\sym{s}_1 + \kappa_2 + \sym{i}, \sym{x}) \ne 0~\wedge \\ & \sym{A}(\sym{s}_1 + \kappa_2 + \sym{i}, \sym{x}) \ne \sym{H}(\sym{s}_1, \sym{x})
\end{align*}
Note that we used $\bar{\Psi}$, $\bar{\Theta}$, and $\theta^{\vec{\kappa}}[\var{s}_1 \rightarrow \sym{s}_1]$ to compose $\Gamma(defyg,\Omega,\theta^{\vec{\kappa}}, \sym{x})$. Then formula $S_1$ is:
\begin{align*}
S_1 \equiv
\exists \matr{m_1 \\ m_2 \\ m_3 \\ m_4}, \matr{w_1 \\ w_2 \\ w_3 \\ w_4}
  \forall \matr{\sym{i} \\ \sym{k} \\ \sym{x}}, \kappa_2, \sym{s}_1~ & ( 0 \leq \kappa_2 \wedge 0 \leq \sym{s}_1 \wedge              \Gamma(defyg,\Omega,\theta^{\vec{\kappa}}|_\mathcal{V}, \sym{x}) ) \rightarrow  \\ & \sym{s}_1 = \max \left\{ 0, \left( \matr{m_1 \\ m_2 \\ m_3 \\ m_4} \kappa_2 + \matr{w_1 \\ w_2 \\ w_3 \\ w_4} \right)^T \matr{\sym{i} \\ \sym{k} \\ \sym{x} \\ 1} \right\}
\end{align*}
Since $S_1$ is not satisfiable we have $\sym{s}_1 = \star$. Therefore a fix-point $\theta^{\vec{\kappa}}$ is
\begin{align*}
\theta^{\vec{\kappa}}(\var{i}) & = \kappa_2 + \sym{i} \\
\theta^{\vec{\kappa}}(\var{j}) & = \star \\
\theta^{\vec{\kappa}}(\var{k}) & = \star \\
\theta^{\vec{\kappa}}(\var{A}) & = \lambda \chi~.~\sym{A}(\chi) \\
\theta^{\vec{\kappa}}(\var{H}) & = \lambda \chi~.~\sym{H}(\chi) \\
\theta^{\vec{\kappa}}(\var{s}_1) & = \star
\end{align*}
And looping condition $\varphi^{\vec{\kappa}}$ is:
\begin{align*}
\varphi^{\vec{\kappa}} \equiv \forall \tau_2 [ & 0 \leq \tau_2 < \kappa_2 \rightarrow ( \\ & \sym{A}(\tau_2 + \sym{i}) \ne 0~\wedge \\ & \exists \kappa_1 (0 \leq \kappa_1~\wedge \\ & \hspace{0.8cm} (\forall \tau_1 [ 0 \leq \tau_1 < \kappa_1 \rightarrow (\sym{H}(\tau_1) \ne 0 \wedge \sym{A}(\tau_1 + \tau_2 + \sym{i}) \ne 0 \wedge \sym{A}(\tau_1 + \tau_2 + \sym{i}) = \sym{H}(\tau_1) )]
)~\wedge \\ & \hspace{0.8cm} \sym{H}(\kappa_1) \ne 0~\wedge \\ & \hspace{0.8cm} \sym{A}(\kappa_1 + \tau_2 + \sym{i}) \ne 0~\wedge \\ & \hspace{0.8cm} \sym{A}(\kappa_1 + \tau_2 + \sym{i}) = \sym{H}(\kappa_1))) ]
\end{align*}

\paragraph{Whole program} After symbolic execution of a backbone tree of the program we receive following function $\Psi$:
\begin{align*}
\Psi(a) & = \true \\
\Psi(ar) & = \true \\
\Psi(ars) & = \true \\
\Psi(arst) & = \true \\
\Psi(\ldots u) & = \true \\
\Psi(\ldots v) & = \true \\
\Psi(\ldots b) & = \true \\
\Psi(\ldots c) & = \true \\
\Psi(\ldots x) & = \true \\
\Psi(\ldots d) & = \forall \tau_2 [ 0 \leq \tau_2 < \kappa_2 \rightarrow ( \\ & \hspace{1.25cm} \sym{A}(\tau_2) \ne 0~\wedge \\ & \hspace{1.25cm} \exists \kappa_1 (0 \leq \kappa_1~\wedge \\ & \hspace{2.05cm} (\forall \tau_1 [ 0 \leq \tau_1 < \kappa_1 \rightarrow ( \sym{H}'(\tau_1) \ne 0 \wedge \sym{A}(\tau_1 + \tau_2) \ne 0 \wedge \sym{A}(\tau_1 + \tau_2) = \sym{H}'(\tau_1) )]
)~\wedge \\ & \hspace{2.05cm} \sym{H}'(\kappa_1) \ne 0~\wedge \\ & \hspace{2.05cm} \sym{A}(\kappa_1 + \tau_2) \ne 0~\wedge \\ & \hspace{2.05cm} \sym{A}(\kappa_1 + \tau_2) = \sym{H}'(\kappa_1))) ] \\
\Psi(\ldots e) & = \sym{A}(\kappa_2) \ne 0 \\
\Psi(\ldots f) & = \true \\
\Psi(\ldots y) & = \true \\
\Psi(\ldots g) & = \forall \tau_3 \left[ 0 \leq \tau_3 < \kappa_3 \rightarrow \left(\sym{H}'(\tau_3) \ne 0 \wedge \sym{A}(\tau_3 + \kappa_2) \ne 0 \wedge \sym{A}(\tau_3 + \kappa_2) = \sym{H}'(\tau_3) \right) \right] \\
\Psi(\ldots l) & = \sym{H}'(\kappa_3) = 0 \\
\Psi(\ldots o) & = \sym{H}'(\kappa_3) = 0 \\
\Psi(\ldots p) & = \true \\
\Psi(\ldots q) & = \true \\
\end{align*}
Note that a backbone tree of the program has only a single backbone path $arstuvbcxdefyglopq$ after its symbolic execution, since each backbone path going through some of program edges $(d,p), (g,h), (l,m)$ is not feasible. Because of space limitations we have abbreviated vertices of the backbone tree. From the same reasons we have introduced new function symbol $\sym{H}': \texttt{int} \rightarrow \texttt{int}$ representing content of array $\sym{H}$ and it is defined as follows $\forall \tau~\sym{H}'(\tau) = \ite(\tau = 0, 72, \ite(\tau = 1, 101, \ite(\tau = 2, 108, \ite(\tau = 3, 108, \ite(\tau = 4, 111, \ite(\tau = 5, 0, \sym{H}(\tau)))))))$. Note that at vertex $\ldots g$ we have recycled the result of the analysis of the inner loop. We have introduced a fresh path counter $\kappa_3$.\\
Finally the abstraction $\hat{\varphi}$ looks as follows
\begin{align*}
\hat{\varphi} \equiv & \exists \kappa_2 (0 \leq \kappa_2~\wedge \\ & \hspace{0.8cm} \forall \tau_2 [ 0 \leq \tau_2 < \kappa_2 \rightarrow ( \\ & \hspace{1.5cm} \sym{A}(\tau_2) \ne 0~\wedge \\ & \hspace{1.5cm} \exists \kappa_1 (0 \leq \kappa_1~\wedge \\ & \hspace{2.3cm} (\forall \tau_1 [ 0 \leq \tau_1 < \kappa_1 \rightarrow ( \sym{H}'(\tau_1) \ne 0 \wedge \sym{A}(\tau_1 + \tau_2) \ne 0 \wedge \sym{A}(\tau_1 + \tau_2) = \sym{H}'(\tau_1) )]
)~\wedge \\ & \hspace{2.3cm} \sym{H}'(\kappa_1) \ne 0~\wedge \\ & \hspace{2.3cm} \sym{A}(\kappa_1 + \tau_2) \ne 0~\wedge \\ & \hspace{2.3cm} \sym{A}(\kappa_1 + \tau_2) = \sym{H}'(\kappa_1))) ]~\wedge \\ & \hspace{0.8cm} \sym{A}(\kappa_2) \ne 0~\wedge \\ & \hspace{0.8cm} \exists \kappa_3 (0 \leq \kappa_3~\wedge \\ & \hspace{1.6cm} \forall \tau_3 \left[ 0 \leq \tau_3 < \kappa_3 \rightarrow \left(\sym{H}'(\tau_3) \ne 0 \wedge \sym{A}(\tau_3 + \kappa_2) \ne 0 \wedge \sym{A}(\tau_3 + \kappa_2) = \sym{H}'(\tau_3) \right) \right]~\wedge \\ & \hspace{1.6cm} \sym{H}'(\kappa_3) = 0))
\end{align*}
We ask an SMT solver to get model. A model of the formula define symbolic input for array \var{A} to contain a string \texttt{"Hello"}, which would navigate symbolic execution directly to the target location.

%%%%%%%%%%%%%%%%%%%%%%%%%%%%%%%%%%%%%%%%%%%%%%%%%%%%%%%%%%%%%%%%%%%%%%%%
%%%%%%%%%%%%%%%%%%%%%%%%%%%%%%%%%%%%%%%%%%%%%%%%%%%%%%%%%%%%%%%%%%%%%%%%
\subsection{MatrIR}

It took $1$min $36$s for \textsc{Pex} to reach the target location in the following program:
\begin{verbatim}
int w = 0;
for (int i = 0; i < m; ++i) {
  int k = 0;
  for (int j = i; j < n; ++j)
    if (A[i][j] > 10 && A[i][j] < 100)
      ++k;
  if (k > 15) {
    w = 1;
    break;
  }
}
if (m > 15 && n > 20 && w == 1)
  assert(false);
\end{verbatim}

\begin{figure}
\begin{center}
\begin{tabular}{c}
\includegraphics[width=9cm, height=17cm]{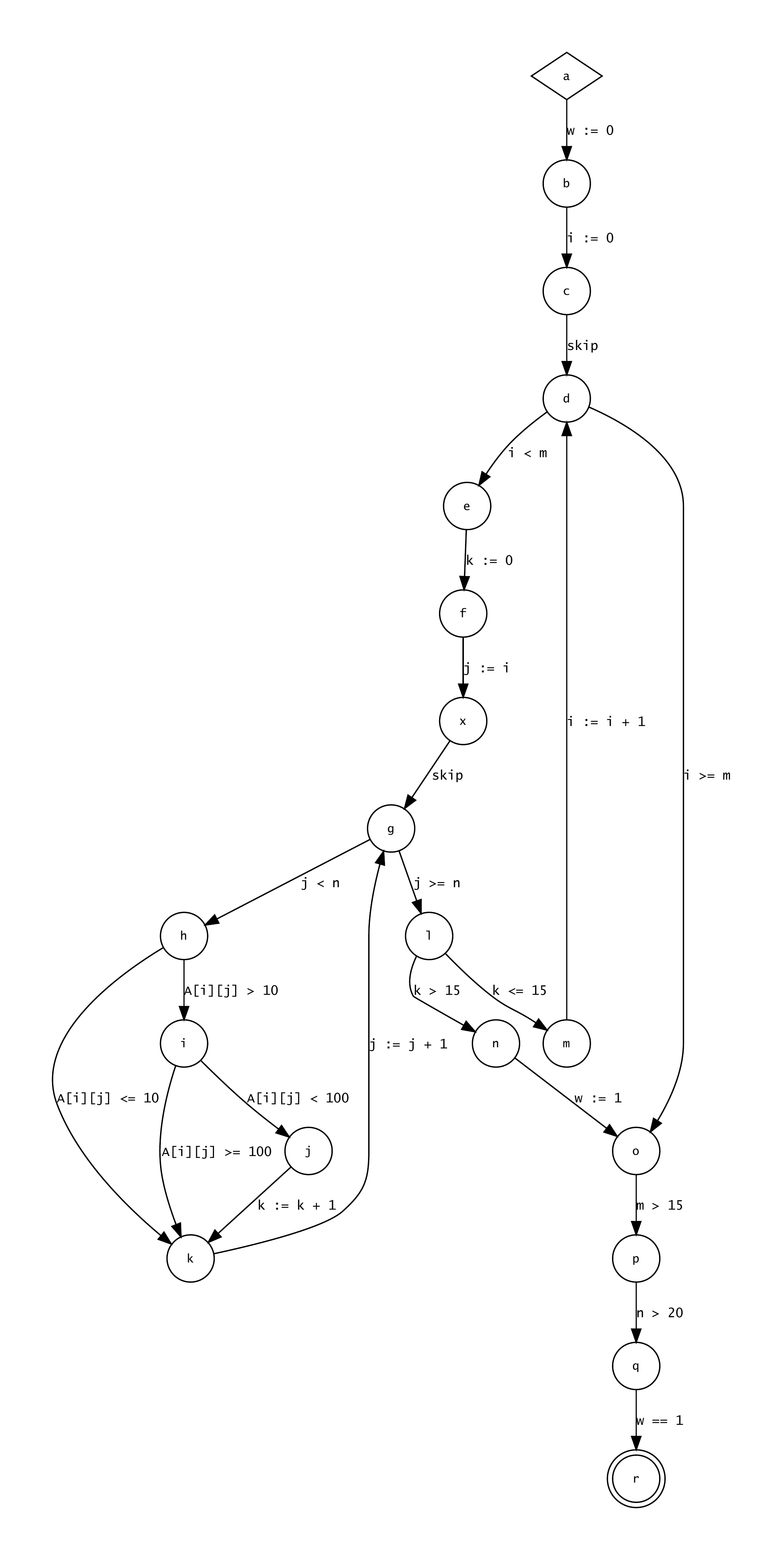}
\end{tabular}
\end{center}
\caption{Program \texttt{MatrIR}}%
\label{fig:MatrIR}%
\end{figure}

\begin{figure}
\begin{center}
\begin{tabular}{cc}
\includegraphics[width=8cm, height=12cm]{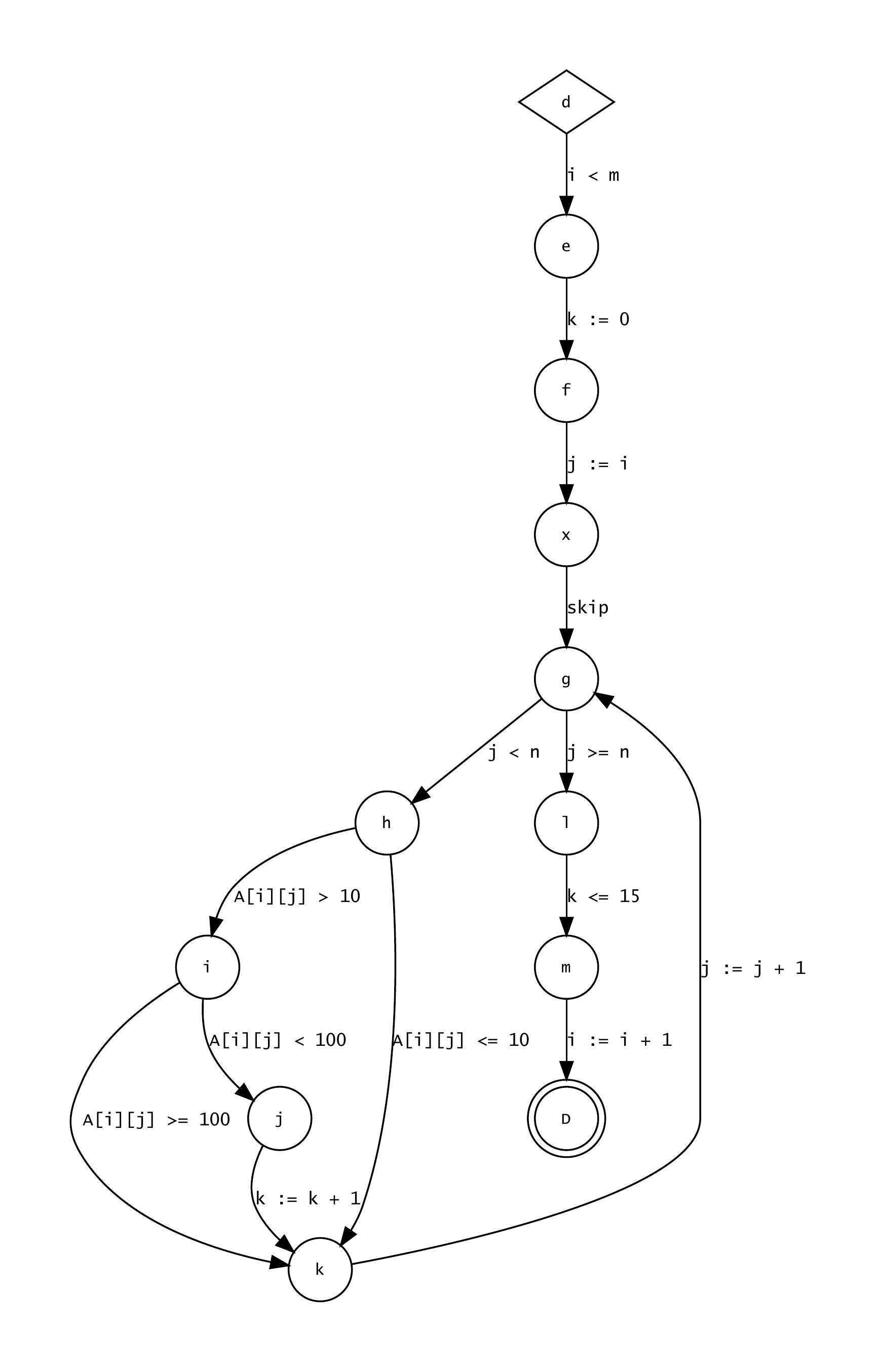} &
\includegraphics[width=5.5cm, height=8.1cm]{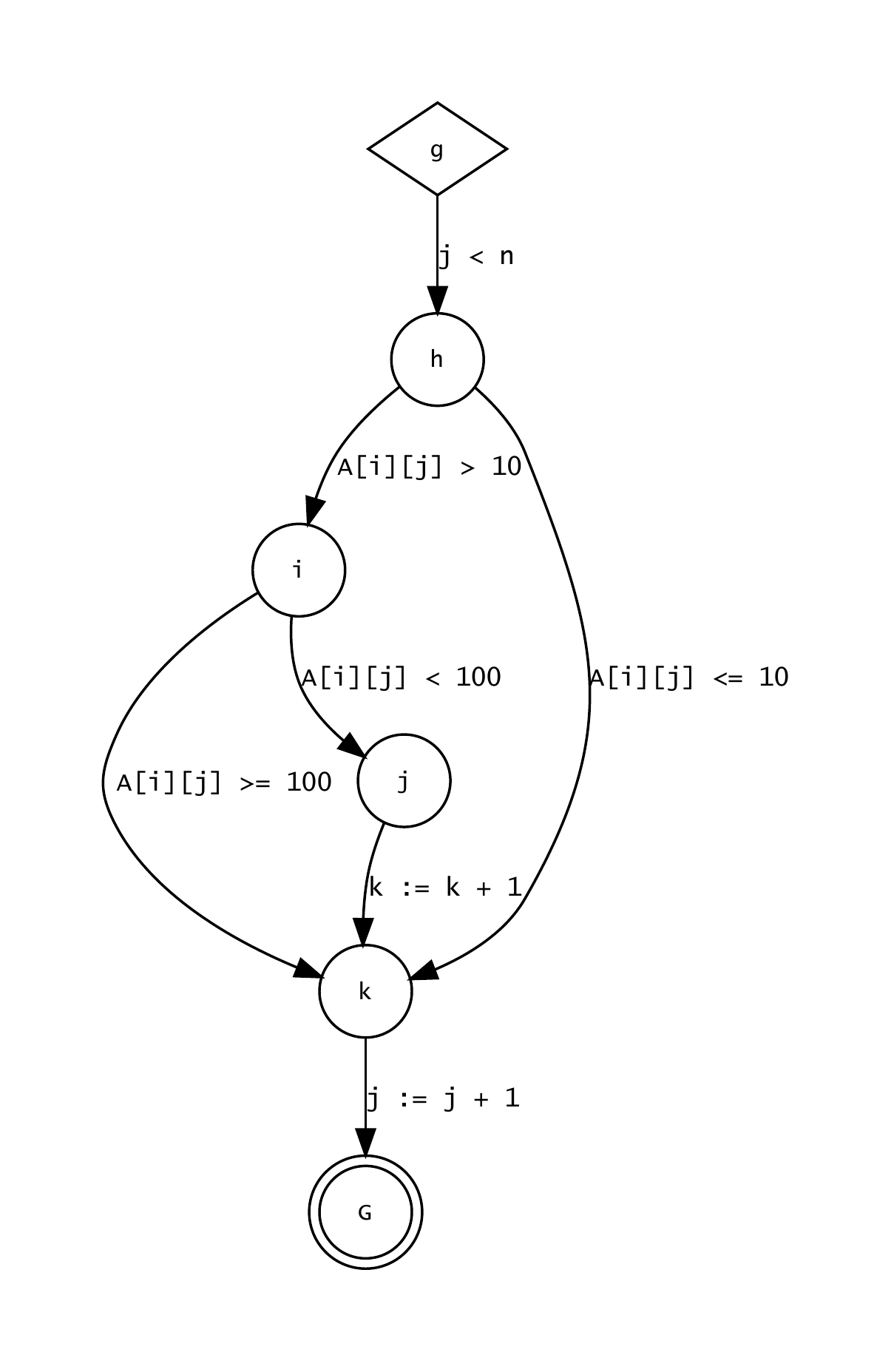} \\
(a) & (b)
\end{tabular}
\end{center}
\caption{Induced programs at loops of program \texttt{MatrIR} (a) Outer loop, (b) Inner loop}%
\label{fig:MatrIROuterInner}%
\end{figure}

\paragraph{Inner loop} There are three backbone paths $ghijkG$, $ghikG$, and $ghkG$ in a backbone tree of the inner loop. \\
Function $\Psi$ looks as follows
\begin{align*}
\Psi(g) & = \true \\
\Psi(gh) & = \sym{j} < \sym{n} \\
\Psi(ghi) & = \sym{A}(\sym{i},\sym{j}) > 10 \\
\Psi(ghij) & = \sym{A}(\sym{i},\sym{j}) < 100 \\
\Psi(ghijk) & = \true \\
\Psi(ghijkG) & = \true \\
\Psi(ghik) & = \sym{A}(\sym{i},\sym{j}) >= 100 \\
\Psi(ghikG) & = \true \\
\Psi(ghk) & = \sym{A}(\sym{i},\sym{j}) <= 10 \\
\Psi(ghkG) & = \true
\end{align*}
Function $\Theta$ is
\begin{center}
\begin{tabular}{lll}
$\Theta(ghijkG)(\var{i}) = \sym{i}$ &
$\Theta(ghikG)(\var{i}) = \sym{i}$ &
$\Theta(ghkG)(\var{i}) = \sym{i}$ \\

$\Theta(ghijkG)(\var{j}) = \sym{j} + 1$ &
$\Theta(ghikG)(\var{j}) = \sym{j} + 1$ &
$\Theta(ghkG)(\var{j}) = \sym{j} + 1$ \\

$\Theta(ghijkG)(\var{k}) = \sym{k} + 1$ &
$\Theta(ghikG)(\var{k}) = \sym{k}$ &
$\Theta(ghkG)(\var{k}) = \sym{k}$ \\

$\Theta(ghijkG)(\var{n}) = \sym{n}$ &
$\Theta(ghikG)(\var{n}) = \sym{n}$ &
$\Theta(ghkG)(\var{n}) = \sym{n}$ \\

$\Theta(ghijkG)(\var{A}) = \lambda \vec{\chi}~.~\sym{A}(\vec{\chi})$ &
$\Theta(ghikG)(\var{A}) = \lambda \vec{\chi}~.~\sym{A}(\vec{\chi})$ &
$\Theta(ghkG)(\var{A}) = \lambda \vec{\chi}~.~\sym{A}(\vec{\chi})$ \\
\end{tabular}
\end{center}
Since $\bar{\Psi} \equiv \Psi$ and $\bar{\Theta} \equiv \Theta$, we receive the following iterated symbolic state $\theta^{\vec{\kappa}}$
\begin{align*}
\theta^{\vec{\kappa}}(\var{i}) & = \sym{i} \\
\theta^{\vec{\kappa}}(\var{j}) & = \kappa_{1,1} + \kappa_{1,2} + \kappa_{1,3} + \sym{j} \\
\theta^{\vec{\kappa}}(\var{k}) & = \kappa_{1,1} + \sym{k} \\
\theta^{\vec{\kappa}}(\var{n}) & = \sym{n} \\
\theta^{\vec{\kappa}}(\var{A}) & = \lambda \vec{\chi}~.~\sym{A}(\vec{\chi})
\end{align*}
Note that we introduced path counters $\kappa_{1,1}, \kappa_{1,2}, \kappa_{1,3}$ for the backbone paths $ghijkG, ghikG, ghkG$ respectively. And looping condition $\varphi^{\vec{\kappa}}$ looks as follows
\begin{align*}
\varphi^{\vec{\kappa}} \equiv (\forall \tau_{1,1}~0 \leq \tau_{1,1} < \kappa_{1,1} \rightarrow \exists \matr{\tau_{1,2} \\ \tau_{1,3}}~& \matr{0 \\ 0} \leq \matr{\tau_{1,2} \\ \tau_{1,3}} \leq \matr{\kappa_{1,2} \\ \kappa_{1,3}}~\wedge \\ & \tau_{1,1} + \tau_{1,2} + \tau_{1,3} + \sym{j} < \sym{n}~\wedge \\ & \sym{A}(\sym{i}, \tau_{1,1} + \tau_{1,2} + \tau_{1,3} + \sym{j}) > 10~\wedge \\ & \sym{A}(\sym{i}, \tau_{1,1} + \tau_{1,2} + \tau_{1,3} + \sym{j}) < 100)~\wedge \\
(\forall \tau_{1,2}~0 \leq \tau_{1,2} < \kappa_{1,2} \rightarrow \exists \matr{\tau_{1,1} \\ \tau_{1,3}}~& \matr{0 \\ 0} \leq \matr{\tau_{1,1} \\ \tau_{1,3}} \leq \matr{\kappa_{1,1} \\ \kappa_{1,3}}~\wedge \\ & \tau_{1,1} + \tau_{1,2} + \tau_{1,3} + \sym{j} < \sym{n}~\wedge \\ & \sym{A}(\sym{i}, \tau_{1,1} + \tau_{1,2} + \tau_{1,3} + \sym{j}) \geq 100)~\wedge \\
(\forall \tau_{1,3}~0 \leq \tau_{1,3} < \kappa_{1,3} \rightarrow \exists \matr{\tau_{1,1} \\ \tau_{1,2}}~& \matr{0 \\ 0} \leq \matr{\tau_{1,1} \\ \tau_{1,2}} \leq \matr{\kappa_{1,1} \\ \kappa_{1,2}}~\wedge \\ & \tau_{1,1} + \tau_{1,2} + \tau_{1,3} + \sym{j} < \sym{n}~\wedge \\ & \sym{A}(\sym{i}, \tau_{1,1} + \tau_{1,2} + \tau_{1,3} + \sym{j}) \leq 10)
\end{align*}

\paragraph{Outer loop} There is a single backbone paths $defxglmo$ in a backbone tree of the inner loop. \\
Function $\Psi$ looks as follows
\begin{align*}
\Psi(d) = &~ \true \\
\Psi(de) = &~ \sym{i} < \sym{m} \\
\Psi(def) = &~ \true \\
\Psi(\ldots x) = &~ \true \\
\Psi(\ldots g) = &~ (\forall \tau_{1,1}~0 \leq \tau_{1,1} < \kappa_{1,1} \rightarrow \\ & \hspace{1.6cm} \exists \matr{\tau_{1,2} \\ \tau_{1,3}}~\matr{0 \\ 0} \leq \matr{\tau_{1,2} \\ \tau_{1,3}} \leq \matr{\kappa_{1,2} \\ \kappa_{1,3}}~\wedge \\ & \hspace{3.2cm} \tau_{1,1} + \tau_{1,2} + \tau_{1,3} + \sym{i} < \sym{n}~\wedge \\ & \hspace{3.2cm} \sym{A}(\sym{i}, \tau_{1,1} + \tau_{1,2} + \tau_{1,3} + \sym{i}) > 10~\wedge \\ & \hspace{3.2cm} \sym{A}(\sym{i}, \tau_{1,1} + \tau_{1,2} + \tau_{1,3} + \sym{i}) < 100)~\wedge \\ &
(\forall \tau_{1,2}~0 \leq \tau_{1,2} < \kappa_{1,2} \rightarrow \\ & \hspace{1.6cm} \exists \matr{\tau_{1,1} \\ \tau_{1,3}}~\matr{0 \\ 0} \leq \matr{\tau_{1,1} \\ \tau_{1,3}} \leq \matr{\kappa_{1,1} \\ \kappa_{1,3}}~\wedge \\ & \hspace{3.2cm} \tau_{1,1} + \tau_{1,2} + \tau_{1,3} + \sym{i} < \sym{n}~\wedge \\ & \hspace{3.2cm}  \sym{A}(\sym{i}, \tau_{1,1} + \tau_{1,2} + \tau_{1,3} + \sym{i}) \geq 100)~\wedge \\ &
(\forall \tau_{1,3}~0 \leq \tau_{1,3} < \kappa_{1,3} \rightarrow \\ & \hspace{1.6cm} \exists \matr{\tau_{1,1} \\ \tau_{1,2}}~\matr{0 \\ 0} \leq \matr{\tau_{1,1} \\ \tau_{1,2}} \leq \matr{\kappa_{1,1} \\ \kappa_{1,2}}~\wedge \\ & \hspace{3.2cm} \tau_{1,1} + \tau_{1,2} + \tau_{1,3} + \sym{i} < \sym{n}~\wedge \\ & \hspace{3.2cm} \sym{A}(\sym{i}, \tau_{1,1} + \tau_{1,2} + \tau_{1,3} + \sym{i}) \leq 10) \\
\Psi(\ldots l) = &~ \kappa_{1,1} + \kappa_{1,2} + \kappa_{1,3} + \sym{i} \geq \sym{n} \\
\Psi(\ldots m) = &~ \kappa_{1,1} \leq 15 \\
\Psi(\ldots o) = &~ \true
\end{align*}
Function $\Theta$ is
\begin{align*}
\Theta(defxglmo)(\var{i}) & = \sym{i} + 1 \\
\Theta(defxglmo)(\var{j}) & = \kappa_{1,1} + \kappa_{1,2} + \kappa_{1,3} + \sym{i} \\
\Theta(defxglmo)(\var{k}) & = \kappa_{1,1} \\
\Theta(defxglmo)(\var{m}) & = \sym{m} \\
\Theta(defxglmo)(\var{n}) & = \sym{n} \\
\Theta(defxglmo)(\var{A}) & = \lambda \vec{\chi}~.~\sym{A}(\vec{\chi}) \\
\end{align*}
Function $\bar{\Psi}$ looks as follows
\begin{align*}
\bar{\Psi}(d) & = \true \\
\bar{\Psi}(de) & = \sym{i} < \sym{m} \\
\bar{\Psi}(def) & = \true \\
\bar{\Psi}(\ldots x) & = \true \\
\bar{\Psi}(\ldots g) & = \forall s~(0 < s < \sym{s}_1 \rightarrow s + \sym{i} < \sym{n}) \\
\bar{\Psi}(\ldots l) & = \sym{s}_1 + \sym{i} \geq \sym{n} \\
\bar{\Psi}(\ldots m) & = \star \leq 15 \\
\bar{\Psi}(\ldots o) & = \true
\end{align*}
Function $\bar{\Theta}$ is
\begin{align*}
\bar{\Theta}(defxglmo)(\var{i}) & = \sym{i} + 1 \\
\bar{\Theta}(defxglmo)(\var{j}) & = \sym{s}_1 + \sym{i} \\
\bar{\Theta}(defxglmo)(\var{k}) & = \star \\
\bar{\Theta}(defxglmo)(\var{m}) & = \sym{m} \\
\bar{\Theta}(defxglmo)(\var{n}) & = \sym{n} \\
\bar{\Theta}(defxglmo)(\var{A}) & = \lambda \vec{\chi}~.~\sym{A}(\vec{\chi}) \\
\end{align*}
After iteration of regular program variables we receive the following $\theta^{\vec{\kappa}}$
\begin{align*}
\theta^{\vec{\kappa}}(\var{i}) & = \kappa_2 + \sym{i} \\
\theta^{\vec{\kappa}}(\var{j}) & = \star \\
\theta^{\vec{\kappa}}(\var{k}) & = \star \\
\theta^{\vec{\kappa}}(\var{m}) & = \sym{m} \\
\theta^{\vec{\kappa}}(\var{n}) & = \sym{n} \\
\theta^{\vec{\kappa}}(\var{A}) & = \lambda \vec{\chi}~.~\sym{A}(\vec{\chi}) \\
\end{align*}
Function $\Gamma(defxg, \Omega, \theta^{\vec{\kappa}}(\var{A}), \sym{x})$ looks as follows
\begin{align*}
\Gamma(defxg, \Omega, \theta^{\vec{\kappa}}(\var{A}), \sym{x}) \equiv & (0 \leq \sym{s}_1 - 1 \rightarrow \sym{s}_1 - 1 + \kappa_2 + \sym{i} < \sym{n}) \wedge \sym{s}_1 + \kappa_2 + \sym{i} \geq \sym{n}
\end{align*}
And therefore the formula $S_1$ is
\begin{align*}
S_1 \equiv & \exists \matr{m_1 \\ m_2 \\ m_3}, \matr{w_1 \\ w_2 \\ w_3} \forall \matr{\sym{n} \\ \sym{i}}, \kappa_2, \sym{s}_1~(\kappa_2 > 0 \wedge \sym{s}_1 > 0 \wedge (0 \leq \sym{s}_1 - 1 \rightarrow \sym{s}_1 - 1 + \kappa_2 + \sym{i} < \sym{n})~\wedge \\ & \hspace{2cm} \sym{s}_1 + \kappa_2 + \sym{i} \geq \sym{n}) \rightarrow \sym{s}_1 = \max \left\{ 0, \left( \matr{m_1 \\ m_2 \\ m_3} \kappa_2 + \matr{w_1 \\ w_2 \\ w_3} \right)^T \matr{\sym{n} \\ \sym{i} \\ 1} \right\}
\end{align*}
After we get a model from an SMT solver we build the following solution for $\sym{s}_1$
\begin{align*}
\sym{s}_1 = \max\{ 0, \sym{n} - \kappa_2 - \sym{i} \}
\end{align*}
Therefore fix-point $\theta^{\vec{\kappa}}$ is
\begin{align*}
\theta^{\vec{\kappa}}(\var{i}) & = \kappa_2 + \sym{i} \\
\theta^{\vec{\kappa}}(\var{j}) & = \max\{ 0, \sym{n} - \kappa_2 - \sym{i} \} + \sym{i} \\
\theta^{\vec{\kappa}}(\var{k}) & = \star \\
\theta^{\vec{\kappa}}(\var{m}) & = \sym{m} \\
\theta^{\vec{\kappa}}(\var{n}) & = \sym{n} \\
\theta^{\vec{\kappa}}(\var{A}) & = \lambda \vec{\chi}~.~\sym{A}(\vec{\chi}) \\
\end{align*}
Looping condition $\varphi^{\vec{\kappa}}$ of the outer loop looks as follows
\begin{align*}
\varphi^{\vec{\kappa}} \equiv ~& \forall \tau_2~0 \leq \tau_2 < \kappa_2 \rightarrow (\kappa_2 + \sym{i} < \sym{m} \wedge \exists \matr{\kappa_{1,1} \\ \kappa_{1,2} \\ \kappa_{1,3}} \matr{0 \\ 0 \\ 0} \leq \matr{\kappa_{1,1} \\ \kappa_{1,2} \\ \kappa_{1,3}}~\wedge \\ & \Psi(defxg) \wedge \kappa_{1,1} + \kappa_{1,2} + \kappa_{1,3} + \sym{i} \geq \sym{n} \wedge \kappa_{1,1} \leq 15)
\end{align*}

\paragraph{Whole program} The backbone tree of the program has only single backbone path $abcdefxglnopqr$ after its symbolic execution, since a backbone path going through program edge $(d,o)$ is infeasible. Therefore function $\Psi$ looks as follows:
\begin{align*}
\Psi(a) & = \true \\
\Psi(ab) & = \true \\
\Psi(abc) & = \true \\
\Psi(\ldots d) & = \forall \tau_2~0 \leq \tau_2 < \kappa_2 \rightarrow (\tau_2 < \sym{m}~\wedge \\ & \hspace{1cm} \exists \matr{\kappa_{1,1} \\ \kappa_{1,2} \\ \kappa_{1,3}} \matr{0 \\ 0 \\ 0} \leq \matr{\kappa_{1,1} \\ \kappa_{1,2} \\ \kappa_{1,3}}~\wedge \\ & \hspace{2cm}
(\forall \tau_{1,1}~0 \leq \tau_{1,1} < \kappa_{1,1} \rightarrow \exists \matr{\tau_{1,2} \\ \tau_{1,3}}~\matr{0 \\ 0} \leq \matr{\tau_{1,2} \\ \tau_{1,3}} \leq \matr{\kappa_{1,2} \\ \kappa_{1,3}}~\wedge \\ & \hspace{7.7cm} \tau_{1,1} + \tau_{1,2} + \tau_{1,3} + \tau_2 < \sym{n}~\wedge \\ & \hspace{7.7cm} \sym{A}(\tau_2, \tau_{1,1} + \tau_{1,2} + \tau_{1,3} + \tau_2) > 10~\wedge \\ & \hspace{7.7cm} \sym{A}(\tau_2, \tau_{1,1} + \tau_{1,2} + \tau_{1,3} + \tau_2) < 100)~\wedge \\ & \hspace{2cm}
(\forall \tau_{1,2}~0 \leq \tau_{1,2} < \kappa_{1,2} \rightarrow \exists \matr{\tau_{1,1} \\ \tau_{1,3}}~\matr{0 \\ 0} \leq \matr{\tau_{1,1} \\ \tau_{1,3}} \leq \matr{\kappa_{1,1} \\ \kappa_{1,3}}~\wedge \\ & \hspace{7.7cm} \tau_{1,1} + \tau_{1,2} + \tau_{1,3} + \tau_2 < \sym{n}~\wedge \\ & \hspace{7.7cm} \sym{A}(\tau_2, \tau_{1,1} + \tau_{1,2} + \tau_{1,3} + \tau_2) \geq 100)~\wedge \\ & \hspace{2cm}
(\forall \tau_{1,3}~0 \leq \tau_{1,3} < \kappa_{1,3} \rightarrow \exists \matr{\tau_{1,1} \\ \tau_{1,2}}~\matr{0 \\ 0} \leq \matr{\tau_{1,1} \\ \tau_{1,2}} \leq \matr{\kappa_{1,1} \\ \kappa_{1,2}}~\wedge \\ & \hspace{7.7cm} \tau_{1,1} + \tau_{1,2} + \tau_{1,3} + \tau_2 < \sym{n}~\wedge \\ & \hspace{7.7cm} \sym{A}(\tau_2, \tau_{1,1} + \tau_{1,2} + \tau_{1,3} + \tau_2) \leq 10)~\wedge \\ & \hspace{2cm} \kappa_{1,1} + \kappa_{1,2} + \kappa_{1,3} + \tau_2 \geq \sym{n} \wedge \kappa_{1,1} \leq 15) \\
\Psi(\ldots e) & = \kappa_2 < \sym{m}\\
\Psi(\ldots f) & = \true \\
\Psi(\ldots x) & = \true \\
\Psi(\ldots g) & = (\forall \tau_{3,1}~0 \leq \tau_{3,1} < \kappa_{3,1} \rightarrow \exists \matr{\tau_{3,2} \\ \tau_{3,3}}~\matr{0 \\ 0} \leq \matr{\tau_{3,2} \\ \tau_{3,3}} \leq \matr{\kappa_{3,2} \\ \kappa_{3,3}}~\wedge \\ & \hspace{6.2cm} \tau_{3,1} + \tau_{3,2} + \tau_{3,3} + \kappa_2 < \sym{n}~\wedge \\ & \hspace{6.2cm} \sym{A}(\kappa_2, \tau_{3,1} + \tau_{3,2} + \tau_{3,3} + \kappa_2) > 10~\wedge \\ & \hspace{6.2cm} \sym{A}(\kappa_2, \tau_{3,1} + \tau_{3,2} + \tau_{3,3} + \kappa_2) < 100)~\wedge \\ & \hspace{0.5cm}
(\forall \tau_{3,2}~0 \leq \tau_{3,2} < \kappa_{3,2} \rightarrow \exists \matr{\tau_{3,1} \\ \tau_{3,3}}~\matr{0 \\ 0} \leq \matr{\tau_{3,1} \\ \tau_{3,3}} \leq \matr{\kappa_{3,1} \\ \kappa_{3,3}}~\wedge \\ & \hspace{6.2cm} \tau_{3,1} + \tau_{3,2} + \tau_{3,3} + \kappa_2 < \sym{n}~\wedge \\ & \hspace{6.2cm} \sym{A}(\kappa_2, \tau_{3,1} + \tau_{3,2} + \tau_{3,3} + \kappa_2) \geq 100)~\wedge \\ & \hspace{0.5cm}
(\forall \tau_{3,3}~0 \leq \tau_{3,3} < \kappa_{3,3} \rightarrow \exists \matr{\tau_{3,1} \\ \tau_{3,2}}~\matr{0 \\ 0} \leq \matr{\tau_{3,1} \\ \tau_{3,2}} \leq \matr{\kappa_{3,1} \\ \kappa_{3,2}}~\wedge \\ & \hspace{6.2cm} \tau_{3,1} + \tau_{3,2} + \tau_{3,3} + \kappa_2 < \sym{n}~\wedge \\ & \hspace{6.2cm} \sym{A}(\kappa_2, \tau_{3,1} + \tau_{3,2} + \tau_{3,3} + \kappa_2) \leq 10) \\
\Psi(\ldots l) & = \kappa_{3,1} + \kappa_{3,2} + \kappa_{3,3} + \kappa_2 \geq \sym{n} \\
\Psi(\ldots n) & = \kappa_{3,1} > 15 \\
\Psi(\ldots o) & = \true \\
\Psi(\ldots p) & = \sym{m} > 15 \\
\Psi(\ldots q) & = \sym{n} > 20 \\
\Psi(\ldots r) & = \true \\
\end{align*}
And we receive the following abstraction $\hat{\varphi}$
\begin{align*}
\hat{\varphi} \equiv~& \exists \kappa_2~(0 \leq \kappa_2~\wedge
\\ & \hspace{0.9cm}
\forall \tau_2~0 \leq \tau_2 < \kappa_2 \rightarrow ( \\ & \hspace{2cm} \tau_2 < \sym{m}~\wedge \\ & \hspace{2cm} \exists \matr{\kappa_{1,1} \\ \kappa_{1,2} \\ \kappa_{1,3}} \matr{0 \\ 0 \\ 0} \leq \matr{\kappa_{1,1} \\ \kappa_{1,2} \\ \kappa_{1,3}}~\wedge \\ & \hspace{3cm}
(\forall \tau_{1,1}~0 \leq \tau_{1,1} < \kappa_{1,1} \rightarrow \exists \matr{\tau_{1,2} \\ \tau_{1,3}}~\matr{0 \\ 0} \leq \matr{\tau_{1,2} \\ \tau_{1,3}} \leq \matr{\kappa_{1,2} \\ \kappa_{1,3}}~\wedge \\ & \hspace{8.7cm} \tau_{1,1} + \tau_{1,2} + \tau_{1,3} + \tau_2 < \sym{n}~\wedge \\ & \hspace{8.7cm} \sym{A}(\tau_2, \tau_{1,1} + \tau_{1,2} + \tau_{1,3} + \tau_2) > 10~\wedge \\ & \hspace{8.7cm} \sym{A}(\tau_2, \tau_{1,1} + \tau_{1,2} + \tau_{1,3} + \tau_2) < 100)~\wedge \\ & \hspace{3cm}
(\forall \tau_{1,2}~0 \leq \tau_{1,2} < \kappa_{1,2} \rightarrow \exists \matr{\tau_{1,1} \\ \tau_{1,3}}~\matr{0 \\ 0} \leq \matr{\tau_{1,1} \\ \tau_{1,3}} \leq \matr{\kappa_{1,1} \\ \kappa_{1,3}}~\wedge \\ & \hspace{8.7cm} \tau_{1,1} + \tau_{1,2} + \tau_{1,3} + \tau_2 < \sym{n}~\wedge \\ & \hspace{8.7cm} \sym{A}(\tau_2, \tau_{1,1} + \tau_{1,2} + \tau_{1,3} + \tau_2) \geq 100)~\wedge \\ & \hspace{3cm}
(\forall \tau_{1,3}~0 \leq \tau_{1,3} < \kappa_{1,3} \rightarrow \exists \matr{\tau_{1,1} \\ \tau_{1,2}}~\matr{0 \\ 0} \leq \matr{\tau_{1,1} \\ \tau_{1,2}} \leq \matr{\kappa_{1,1} \\ \kappa_{1,2}}~\wedge \\ & \hspace{8.7cm} \tau_{1,1} + \tau_{1,2} + \tau_{1,3} + \tau_2 < \sym{n}~\wedge \\ & \hspace{8.7cm} \sym{A}(\tau_2, \tau_{1,1} + \tau_{1,2} + \tau_{1,3} + \tau_2) \leq 10)~\wedge \\ & \hspace{3cm} \kappa_{1,1} + \kappa_{1,2} + \kappa_{1,3} + \tau_2 \geq \sym{n}~\wedge \\ & \hspace{3cm} \kappa_{1,1} \leq 15)~\wedge
\\ & \hspace{0.9cm}
\kappa_2 < \sym{m}~\wedge
\\ & \hspace{0.9cm}
\exists \matr{\kappa_{3,1} \\ \kappa_{3,2} \\ \kappa_{3,3}}~(\matr{0 \\ 0 \\ 0} \leq \matr{\kappa_{3,1} \\ \kappa_{3,2} \\ \kappa_{3,3}}~\wedge
\\ & \hspace{2.6cm}
(\forall \tau_{3,1}~0 \leq \tau_{3,1} < \kappa_{3,1} \rightarrow \exists \matr{\tau_{3,2} \\ \tau_{3,3}}~\matr{0 \\ 0} \leq \matr{\tau_{3,2} \\ \tau_{3,3}} \leq \matr{\kappa_{3,2} \\ \kappa_{3,3}}~\wedge \\ & \hspace{6.2cm} \tau_{3,1} + \tau_{3,2} + \tau_{3,3} + \kappa_2 < \sym{n}~\wedge \\ & \hspace{6.2cm} \sym{A}(\kappa_2, \tau_{3,1} + \tau_{3,2} + \tau_{3,3} + \kappa_2) > 10~\wedge \\ & \hspace{6.2cm} \sym{A}(\kappa_2, \tau_{3,1} + \tau_{3,2} + \tau_{3,3} + \kappa_2) < 100)~\wedge
\\ & \hspace{2.6cm}
(\forall \tau_{3,2}~0 \leq \tau_{3,2} < \kappa_{3,2} \rightarrow \exists \matr{\tau_{3,1} \\ \tau_{3,3}}~\matr{0 \\ 0} \leq \matr{\tau_{3,1} \\ \tau_{3,3}} \leq \matr{\kappa_{3,1} \\ \kappa_{3,3}}~\wedge \\ & \hspace{6.2cm} \tau_{3,1} + \tau_{3,2} + \tau_{3,3} + \kappa_2 < \sym{n}~\wedge \\ & \hspace{6.2cm} \sym{A}(\kappa_2, \tau_{3,1} + \tau_{3,2} + \tau_{3,3} + \kappa_2) \geq 100)~\wedge
\\ & \hspace{2.6cm}
(\forall \tau_{3,3}~0 \leq \tau_{3,3} < \kappa_{3,3} \rightarrow \exists \matr{\tau_{3,1} \\ \tau_{3,2}}~\matr{0 \\ 0} \leq \matr{\tau_{3,1} \\ \tau_{3,2}} \leq \matr{\kappa_{3,1} \\ \kappa_{3,2}}~\wedge \\ & \hspace{6.2cm} \tau_{3,1} + \tau_{3,2} + \tau_{3,3} + \kappa_2 < \sym{n}~\wedge \\ & \hspace{6.2cm} \sym{A}(\kappa_2, \tau_{3,1} + \tau_{3,2} + \tau_{3,3} + \kappa_2) \leq 10)~\wedge
\\ & \hspace{2.6cm}
\kappa_{3,1} + \kappa_{3,2} + \kappa_{3,3} + \kappa_2 \geq \sym{n}~\wedge
\\ & \hspace{2.6cm}
\kappa_{3,1} > 15~\wedge
\\ & \hspace{2.6cm}
\sym{m} > 15~\wedge
\\ & \hspace{2.6cm}
\sym{n} > 20))
\end{align*}
From a model returned from an SMT solver we can see, that we have input to the program which reaches the target location. Note that although Z3~SMT solver correctly computed content of the array \var{A} (i.e.~there are numbers $11$ everywhere), the size $n$ of the array \var{A} is unnecessarily large $1257$.

\end{document}